\documentclass[11pt,letterpaper]{article}
\usepackage{setspace} % Load before hyperref so footnote links are retained.

\usepackage{natbib}
\usepackage{bibunits}
\defaultbibliographystyle{ormsv080}
\defaultbibliography{main}
\usepackage[pagebackref=false, colorlinks=true, citecolor=blue, anchorcolor=blue, urlcolor=blue]{hyperref}
\usepackage[letterpaper,margin=1in]{geometry}
\usepackage{caption}
\usepackage{hyperref}
\usepackage{amsmath}
\usepackage{amsthm}
\usepackage{tabularx} % for automatic column resizing
\usepackage{longtable}
\usepackage{graphicx}
\usepackage{amsfonts}
\usepackage{amssymb}
\usepackage{comment}
\usepackage{environ}
\newif\ifincludeextra
\includeextratrue
\ifdefined\OmitExtraAppendix\includeextrafalse\fi
\newif\ifshowreviewcomments
\showreviewcommentstrue
\ifdefined\HideReviewComments\showreviewcommentsfalse\fi
\usepackage[dvipsnames]{xcolor}
\usepackage[most]{tcolorbox}
\definecolor{editbrown}{RGB}{150,75,0}

\NewEnviron{authornote}{\ifshowreviewcomments\par\begingroup\normalfont\bfseries\color{blue}\hypersetup{linkcolor=blue,citecolor=blue,urlcolor=blue}\noindent[Author comment: \BODY\unskip\nobreak]\par\endgroup\fi}
\newtcolorbox{redproofbox}{
    colframe=red,
    colback=white,
    boxrule=1pt,
    arc=0pt,
    left=2pt,
    right=2pt,
    top=6pt,
    bottom=6pt,
    breakable
}
\usepackage{makecell}
\usepackage{subfigure}
\usepackage{multirow, booktabs}
\usepackage{array}
\usepackage{float}
\usepackage[capitalise,noabbrev]{cleveref}
\usepackage{enumerate}
\usepackage{enumitem}
\usepackage{changepage}
\newcolumntype{C}{>{\centering\arraybackslash}X}
\newcolumntype{R}[1]{>{\centering\arraybackslash}p{#1}}
\usepackage{algorithm,algpseudocode}

\renewcommand{\eqref}[1]{(\ref{#1})}
\def\bb{\textbf{b}}

\def\be{\textbf{e}}
\def\bff{\mathbf{f}}

\def\bi{\textbf{i}}

\def\bs{\textbf{s}}

\def\bw{\mathrm{\textbf{\textit{w}}}}

\def\bx{\mathbf{x}}  %{\mbox{\boldmath $\lambda$}}

\def\bz{\textbf{z}}

\def\bB{\textbf{B}}
\def\bC{\textbf{C}}

\def\bF{\textbf{F}}

\def\bH{\textbf{H}}
\def\bI{\textbf{I}}

\def\bM{\textbf{M}}

\def\bW{\textbf{W}}

\def\bfDelta{{\boldsymbol{\Delta}}}

\def\bfxi{{\boldsymbol{\xi}}}

\def\dx{\bfDelta \bx}

\def\cF{\mathcal{F}}

\def\bsig{\boldsymbol{\Sigma}}

\def\mE{\mathbb{E}}

\def\mR{\mathbb{R}}

\def\smskip{\smallskip}

\def\texitem#1{\par\smskip\noindent\hangindent 25pt
               \hbox to 25pt {\hss #1 ~}\ignorespaces}

\newcommand{\BEAS}{\begin{eqnarray*}}
\newcommand{\EEAS}{\end{eqnarray*}}
\newcommand{\BEA}{\begin{eqnarray}}
\newcommand{\EEA}{\end{eqnarray}}
\newcommand{\BEQ}{\begin{eqnarray}}
\newcommand{\EEQ}{\end{eqnarray}}
\newcommand{\BIT}{\begin{itemize} \setlength{\itemsep}{0.5cm}}
\newcommand{\EIT}{\end{itemize}}
\newcommand{\BNUM}{\begin{enumerate}}
\newcommand{\ENUM}{\end{enumerate}}

\newcommand{\BA}{\begin{array}}
\newcommand{\EA}{\end{array}}

\newcommand{\argmin}{\mathop{\rm argmin}}
\newcommand{\argmax}{\mathop{\rm argmax}}

\usepackage{mathtools}
\usepackage{bbm}

\newcommand{\defeq}{\coloneq}
\newcommand{\Ht}[1]{\textbf{H}_{#1}}
\newcommand{\Ct}{\mathcal{C}}
\newcommand{\bbone}{\mathbbm{1}}

\newcommand\E{\mathbb{E}}

\makeatletter
\@ifundefined{eqslim}{%
}{}
\makeatother

\input epsf
\parskip\medskipamount

\newcommand{\rev}[1]{{\color{darkred} [Revision: #1]}}

\newcommand{\td}{\text{d}}

\newcommand{\tG}{{\mbox{\tiny G}}}

\definecolor{rose}{rgb}{1.0, 0.33, 0.64}
\newcommand{\zz}[1]{{\color{rose} [ZZ: {#1}]}}
\newcommand{\martin}[1]{\textcolor{purple}{[MH: #1]}}

\newtheorem{thm}{Theorem}[section]
\newtheorem{exm}{Example}%[section]
\newtheorem{example}[exm]{Example}
\crefname{example}{Example}{Examples}
\newtheorem{dfn}{Definition}[section]
\newtheorem{cor}[thm]{Corollary}
\newtheorem{lem}[thm]{Lemma}

\newtheorem{rem}[thm]{Remark}
\newtheorem{assum}{Assumption}%[section]

\def\bsym#1{{\boldsymbol{#1}}}

\newcounter{counterexample}
\newcommand{\bit}{\begin{itemize}}
\newcommand{\eit}{\end{itemize}}
\newcommand{\bnum}{\begin{enumerate}}
\newcommand{\enum}{\end{enumerate}}

\definecolor{OliveGreen}{cmyk}{0.64,0,0.95,0.40}
\definecolor{BrickRed}{cmyk}{0,0.89,0.94,0.28}
\definecolor{darkred}{HTML}{CC0000}

\newfloat{algorithm}{htp}{lop}
\floatname{algorithm}{Algorithm}
\crefname{algorithm}{Algorithm}{Algorithms}

\makeatletter
\newcommand{\setbibliographypart}[1]{%
  \def\@extra@b@citeb{.#1}\def\@extra@binfo{.#1}}
\makeatother
\newcommand{\startappendixpart}[1]{%
  \clearpage
  \setcounter{page}{1}\renewcommand{\thepage}{#1-\arabic{page}}%
  \setcounter{section}{0}\renewcommand{\thesection}{#1.\arabic{section}}%
  \renewcommand{\theHsection}{#1.\arabic{section}}%
  \setcounter{equation}{0}\renewcommand{\theequation}{#1.\arabic{equation}}%
  \renewcommand{\theHequation}{#1.\arabic{equation}}%
  \setcounter{figure}{0}\renewcommand{\thefigure}{#1.\arabic{figure}}%
  \renewcommand{\theHfigure}{#1.\arabic{figure}}%
  \setcounter{table}{0}\renewcommand{\thetable}{#1.\arabic{table}}%
  \renewcommand{\theHtable}{#1.\arabic{table}}%
  \setcounter{algorithm}{0}\renewcommand{\thealgorithm}{#1.\arabic{algorithm}}%
  \renewcommand{\theHalgorithm}{#1.\arabic{algorithm}}%
  \setcounter{exm}{0}\renewcommand{\theexm}{#1.\arabic{exm}}%
  \renewcommand{\theHexm}{#1.\arabic{exm}}%
  \setcounter{prop}{0}\renewcommand{\theprop}{#1.\arabic{prop}}%
  \renewcommand{\theHprop}{#1.\arabic{prop}}%
  \setcounter{assum}{0}\renewcommand{\theassum}{#1.\arabic{assum}}%
  \renewcommand{\theHassum}{#1.\arabic{assum}}%
}

\begin{document}
\setbibliographypart{main}
\begin{bibunit}

\parindent 0em

\title{\bf When Stress Tests Miss the Risk: Statistical Scenario Analysis for Financial Portfolios}

\author{{\normalsize\bfseries Zhongze Cai\textsuperscript{1} \qquad Martin B. Haugh\textsuperscript{1} \qquad Xiaocheng Li\textsuperscript{2}}\\[0.5em]
{\normalsize\textsuperscript{1}Department of Analytics, Marketing and Operations}\\
{\normalsize Imperial Business School, Imperial College}\\
{\normalsize London, SW7 2AZ, UK}\\[0.25em]
{\normalsize\textsuperscript{2} Independent}}

%{\sc Working Paper under Revision}}

%\date{This version: \today}
\date{}

%\doublespacing

\maketitle
\thispagestyle{empty}
\vspace{-0.7cm}
\begin{abstract}
\normalsize
Scenario analysis is widely used to stress test financial portfolios, yet conventional approaches often summarize scenario gains using point estimates that overlook dependence between stressed and unstressed risk factors. We develop a machine learning framework for quantifying uncertainty in realized next-day scenario gains through prediction intervals whose widths can be calibrated online. Adaptive conformal scenario analysis (ACSA) calibrates scenario-specific quantile predictions and provides a long-run empirical coverage guarantee over realized scenarios. Our main method, kernel scenario analysis (KSA), estimates scenario-conditional quantiles directly from the specified stress and current market information. KSA can also be combined with ACSA to obtain online-calibrated prediction intervals. In experiments designed to reflect real-world markets, conventional scenario analysis can substantially understate risk, including for portfolios that appear safe under standard stress tests. Our proposed methods achieve better calibration and sharper intervals than empirical-quantile baselines. Overall, the framework moves scenario analysis beyond point estimates by quantifying predictive uncertainty and providing tools for statistically validating scenario gains.

\end{abstract}

\newpage

\section{Introduction}
\label{sec:intro}

Scenario analysis is a cornerstone of financial risk management. Banks, asset managers, and
regulators routinely ask: \emph{What happens to the portfolio if a small set of market drivers moves
in an adverse but plausible way?} The output is typically a ``scenario table'' that reports an
estimated profit-and-loss (P\&L) for each prescribed stress or scenario. For an equity-derivatives portfolio,
for example, a scenario might specify that the S\&P 500 falls by 3\% while the implied volatility
surface shifts upward by several volatility points, and the scenario analysis report might state
that the portfolio would gain (or lose) a particular dollar amount in that scenario.

Despite its operational importance, standard scenario analysis is often treated as a purely
deterministic calculation. A small subset of risk-factor changes is ``shocked'', while all other
risk factors are implicitly held fixed. The resulting scenario P\&L is then interpreted---often
informally, but pervasively---as an \emph{expected} P\&L conditional on the scenario. This
interpretation is appealing because it turns a high-dimensional and uncertain forecasting problem
into a simple table of numbers. Unfortunately, it is also frequently wrong. There are at least
three structural reasons for this.

First, financial risk factors often exhibit strong statistical dependence. Conditioning on a stress in
one part of the market typically shifts the distribution of other risk factors, often materially.
Treating the unstressed factors as unchanged can therefore yield a scenario P\&L that is badly
miscalibrated even as a conditional mean, let alone as a full description of risk. Second, for
nonlinear portfolios, especially those containing options or other convex payoffs, even replacing
unstressed factors by their conditional means can miss important \emph{residual uncertainty} effects:
convexity can cause the conditional expectation of the portfolio value to differ substantially from
the portfolio value computed at the conditional expectation of the risk factors. Third, scenario
analysis outputs are rarely accompanied by any statistical validation. Although the prescribed
stress scenarios are typically low-probability events, \emph{some} scenario (possibly mild) is
realized every day. This makes it possible to backtest scenario analysis in an operationally
relevant way: roll the clock back one day, run the scenario analysis for the scenario that actually
occurred, and compare the predicted and realized P\&L. Yet such backtesting is not standard practice,
and, as a result, risk managers often lack a principled notion of how accurate their scenario
numbers really are.

This paper argues that scenario analysis should be treated as a statistical prediction problem, and
that risk managers should report \emph{uncertainty} alongside point estimates. Concretely, let
$G_{t+1}$ denote the one-step portfolio gain (daily P\&L in our experiments). Let $\bz$ denote a stress
scenario, defined in terms of shocks to selected common factors, and interpret the object of
interest as the \emph{conditional distribution} of $G_{t+1}$ given the scenario and today’s
information. Our goal is to produce, for each scenario $\bz$ under consideration, a prediction
interval $\Ct_t(\bz)$ for $G_{t+1}$ that is (i) informative in practice (not overly conservative) and
(ii) accompanied by provable coverage guarantees.

Our first method, \emph{Adaptive Conformal Scenario Analysis} (ACSA), wraps a scenario-wise
quantile predictor in an online calibration procedure. It updates the quantile level according
to whether the preceding interval covered the realized P\&L in the realized scenario, widening
intervals when empirical coverage is too low and shrinking them when it is too high. The
resulting intervals enjoy a long-run \emph{marginal} coverage guarantee across the realized
scenario sequence, under mild requirements on the predictor and without a correct parametric
specification. This provides a practical feedback mechanism for calibrating scenario uncertainty
when the underlying prediction model is imperfect.

Marginal calibration does not ensure adequate coverage in every scenario. A procedure can
achieve the correct average coverage while under-covering extreme regions, where losses may be large and nonlinear exposures are particularly important. Our second method, \emph{Kernel Scenario
Analysis} (KSA), estimates scenario-conditional quantiles by reweighting historical observations
according to their similarity to the current market state and hypothetical factor stress, using
a kernel-based notion of distance. Scenarios are defined in terms of a low-dimensional set of
common-factor moves, while today's market state supplies additional conditioning information
about tomorrow's distribution.

ACSA and KSA offer two complementary ways of improving scenario analysis. ACSA adjusts for systematic miscalibration as new outcomes are observed, while KSA uses the structure of the current market state and the proposed stress to distinguish between scenarios with different levels of uncertainty. They can be applied independently, or combined when both forms of adaptation are desirable.

The paper makes four contributions:
\begin{itemize}

\item \vspace{-1em} We formulate scenario analysis as a statistical prediction problem and augment scenario
tables with prediction intervals and backtestable coverage guarantees.

\item \vspace{-0.5em} We adapt conformal calibration to scenario tables through ACSA and obtain long-run
marginal coverage across realized scenarios without requiring a correctly specified
parametric model.

\item \vspace{-0.5em} We develop KSA, which estimates scenario-conditional quantiles by jointly conditioning
on the specified scenario and the current market state. We establish a theoretical bound on its conditional coverage error under some regularity assumptions.

\item \vspace{-0.5em} Through numerical studies based on S\&P 500 option portfolios, we demonstrate more
reliable and informative scenario assessments than standard scenario analysis, including settings where dependence, convexity, and adversarial exposures make its point estimates misleading.
\end{itemize}
\vspace{-1em}

The remainder of this paper is organized as follows. Section~\ref{sec:LitReview} reviews related work. Sections~\ref{sec:Prelims}--\ref{sec:MotivatingEG}
formulate the problem and develop the three-factor motivating example, which serves as a
controlled testbed for the methods.
Sections~\ref{sec:ACSA} and \ref{sec:KSA} present ACSA and KSA, respectively, while Section~\ref{sec:Numerics}
reports numerical experiments. The E-Companion contains proofs, implementation details, and a summary of our notation may be found in E-Companion~\ref{subsec:notation_summary}.

\section{Literature Review}
\label{sec:LitReview}

Much of financial risk management concerns scalar measures such as value-at-risk (VaR),
conditional value-at-risk (CVaR), and coherent risk measures; see \citet{MFE}. \citet{WangZiegel2021} study risk measures based on loss distributions
under multiple scenarios whereas our focus is on prediction intervals for
portfolio gains within each specified stress scenario. Following the \citet{ECB-2006}, scenarios may be \emph{historical} by replaying realized
market moves on today's portfolio, \emph{hypothetical} by applying a specified narrative,
\emph{probabilistic} in that they are generated from a joint model fitted to historical data, or \emph{reverse-engineered} where we work backwards from a portfolio-outcome threshold to find
scenarios likely to exceed it. These categories often overlap. The empirical Bayes approach
of \citet*{Glasserman15}, for example, combines reverse engineering with probabilistic
structure calibrated to historical data, while hypothetical narratives are commonly informed
by historical frequencies and co-movement patterns. Our work also uses historical data
to learn scenario-wise behavior, but we apply nonparametric methods rather than fitting a
parametric joint distribution, and we provide asymptotic coverage guarantees for scenario-gain uncertainty statements.

The \citet{Basel-2005} guidance emphasizes that scenarios should be \emph{plausible},
\emph{severe}, and \emph{suggestive} of risk-mitigating actions. One way to balance plausibility
and severity is to define a region of plausible risk-factor moves and search for the most
adverse portfolio outcome within it. \citet*{Breuer09} develop this idea and study
\emph{partial-scenario} stressing: they maximize plausibility by setting unstressed systematic
factors to their conditional means given the stressed factors. \citet{Haugh2020} propagate
the full conditional \emph{distribution} of unstressed factors, which matters for derivatives
portfolios because of convexity. A related asset-allocation perspective appears in
Black--Litterman \citep{Black7} and its extensions \citep{Meucci2008}, which compute risk-factor
distributions conditional on a specified view, typically a linear constraint in a Gaussian setting.

Other approaches to scenario selection and structuring include graphical
models for scenario selection and conditional propagation
\citep{rebonato2010-2,rebonato-2019}, systematic search over a plausible
risk-factor distribution \citep{Flood15}, hypothetical scenario design
\citep{Golubjpm.2018.1.079}, and scenario aggregation under model
uncertainty \citep{CambouFilipovic2017}. Additional contributions include \citet{Alfaro,Quag09,Bonti06},
while \citet{Borio-2012} review practice and policy considerations. Recent work also uses deep generative models for financial scenario
simulation. \citet*{ContEtAl2026TailGAN} train a generative adversarial
network to reproduce value-at-risk and expected-shortfall
characteristics of static portfolios and dynamic trading strategies,
while \citet*{ChenEtAl2025DiffusionFactors} incorporate latent factor
structure into diffusion models for high-dimensional asset-return
generation.

Factor structure underlies many probabilistic scenario models. Equity examples include the
capital asset pricing model (CAPM) and Fama--French three-factor model \citet{Fama}. For fixed income,
\citet{Diebold06} propose a three-factor yield-curve model, while \citet{Saroka30} use
principal component analysis (PCA) to identify likely yield-curve moves consistent with a
specified view. For equity derivatives, \citet{cont2002} develop a PCA model of implied
volatility surfaces. Since financial factors exhibit temporal dependence, time-series
components are also needed; see e.g. \citet{Haugh2020}. \citet{PelgerXiong2022} combine nonparametric estimation with PCA to
estimate factor models whose loadings vary with the economic state. We refer to \citet{MFE},\citet{ruppert} and \citet{tsay} for
textbook treatments and further references.

Our methodology draws on quantile regression \citep{koenker2005quantile}, uncertainty
quantification, and nonparametric statistics. Conditional quantiles provide prediction
intervals whose coverage can be studied conditionally or averaged over realizations of the
conditioning variables. Early analyses provide conditional consistency under restrictive
assumptions \citep{vovk2005algorithmic,takeuchi2006nonparametric,steinwart2011estimating}.
More recent work, motivated partly by flexible predictors such as deep networks, emphasizes
marginal guarantees and evaluates conditional performance empirically. Major streams include
conformal prediction \citep{angelopoulos2021gentle,papadopoulos2002inductive,lei2018distribution}
and uncertainty calibration \citep{roth2022uncertain}. For non-exchangeable panel data, \citet{TuGiesecke2026} combine
history-based similarity weights with adaptive miscoverage updates,
using contemporaneous outcomes from other units for calibration.

Kernel methods have also been used to learn dynamic portfolio values
from simulated cash flows \citep{BoudabsaFilipovic2022}.
Kernel conditional-quantile estimation for dependent data is well established.
\citet{cai2002regression} estimates time-series regression quantiles by inverting a
weighted Nadaraya--Watson conditional CDF estimator, while
\citet{hansen2008uniform} develops uniform convergence theory for kernel estimators
under strong mixing. 
More recently, \citet{lee2025kowcpi} use reweighted
Nadaraya--Watson estimation of nonconformity-score quantiles for
time-series prediction and establish a conditional-coverage-gap bound
under mixing assumptions. Our analysis provides an explicit
finite-sample bound on scenario-conditional coverage error, uniformly
over a compact operational feature region, under the stated dependence
and predictive-sufficiency assumptions.

Our contribution is the combination of these statistical ingredients with a financial scenario-analysis formulation. KSA conditions jointly on a hypothetical stressed-factor scenario and the current market state to produce an entire scenario table. Its residual-centering construction is compatible with security-level conditional-mean estimation and practical aggregation across changing portfolios. The latter is an implementation advantage, whereas Theorem~\ref{thm:individual_cali_thm} concerns a fixed portfolio and deterministic centering function. The theorem gives a uniform conditional-coverage bound over a compact operational feature region under its stated assumptions. ACSA adapts conformal calibration for time series \citep{gibbs2021adaptive,gibbs2024conformal} to scenario tables, and KSA can supply its quantile predictor. The financial-risk interpretation and adversarial-portfolio application further distinguish the resulting framework from ordinary one-step time-series prediction.

Evaluation and ongoing validation provide a further connection to financial risk forecasting.
Proper scoring rules reward calibration and sharpness \citep{Gneiting_Raftery}; quantile
forecasts can be assessed using pinball loss, while joint scoring of VaR and expected
shortfall (ES) enables comparison of tail-risk forecasts \citep{FisslerZiegel2016,HeKouPeng}.
\citet{NoldeZiegel2017} advocate comparative backtests alongside traditional pass/fail tests
to encourage accurate risk forecasts. Classic VaR backtests examine coverage and independence
of exceedances, often in fixed samples. Sequential monitoring instead seeks prompt detection
of forecast deterioration as observations arrive \citep{HogaDemetrescu2023} while e-backtesting
provides sequentially valid tests \citep{WangWangZiegel2025}. Related work develops
model-independent and conditional backtests for ES \citep{AcerbiSzekely2014,DuEscanciano2017}.
Although these approaches concern scalar risk forecasts, their emphasis on statistical
validation also applies to scenario tables. Our target is scenario-wise predictive uncertainty
for portfolio gains, evaluated through interval scores and backtesting, with an online
perspective suited to forecast accuracy that varies over time.

\section{Preliminaries and ``Standard'' Scenario Analysis}
\label{sec:Prelims}

In this section, we formulate the scenario analysis problem. We begin in Section \ref{sec:Examples} with a motivating portfolio setting where scenario analysis can be applied. In Section \ref{sec:Factor Structure} we introduce the factor model structure that typically underlies a scenario analysis. Then in Section \ref{sec:SSA} we describe how a so-called standard scenario analysis works. Finally, in Section \ref{sec:SSA_Weaknesses} we describe the main weaknesses of standard scenario analysis.

\subsection{A Motivating Portfolio} \label{sec:Examples}\label{eg:USEquityOption}

We assume we have a fixed portfolio of securities that will not change over the horizon under consideration, e.g., 1 day. In principle the portfolio could include any combination of securities from any asset class, although in practice we are typically limited to reasonably liquid securities for which historical price data are available. Specifically, consider a U.S. equity and options portfolio.
The time $t$ portfolio value $V_t$ satisfies
\begin{equation} \label{eq:EquityOptions1}
V_t = \sum_{i=1}^n u_i S_t^{(i)} + \sum_{j=1}^M u_j^d C_t^{(j)},
\end{equation}
where the first sum in (\ref{eq:EquityOptions1}) represents  holdings (the $u_i$'s) in domestic stocks (with time $t$ prices $S_t^{(i)}$), and the second sum represents holdings ($u_j^d$'s) in $M$ options (with prices $C_t^{(j)}$) of various strikes, maturities and underlying securities. We might assume a simple CAPM-style factor model for the domestic market so that
\begin{equation} \label{eq:CAPM_for_St}
R_{t+1}^{(i)} := \log(S^{(i)}_{t+1}/S_t^{(i)}) = r_f + \beta_i (R^m_{t+1} - r_f) + \epsilon_{t+1}^{(i)},
\end{equation}
where $R^m_{t+1}$ is the domestic market log return between dates $t$ and $t+1$, $r_f$ the domestic (and assumed constant) risk-free rate, and the $\epsilon_{t+1}^{(i)}$'s are independent random noise. This makes $R^m_{t+1}$ the only common equity factor return. We may also assume that there are three volatility factors for the market, and that the market volatility model also drives the volatility surface for the individual stocks. This means there are a total of four common factors. In our numerical experiments in Section \ref{sec:Numerics} we will use this setting with just a single underlying security.

%There are of course many other possible settings\ifincludeextra\ -- see Extra Appendix~\ref{apx:example_of_Section_3} for additional examples --\fi\ and depending on the nature of the portfolio, a more or less refined factor structure might be required.
%For example, the inclusion of Fama-French or momentum factors might be required in Example \ref{eg:USEquityOption} if the portfolio was constructed with a view to trading these factors.

\subsection{The Factor Model Structure} \label{sec:Factor Structure}

We now provide a more formal description of the market structure and present the notation that we will use throughout the remainder of the paper. This market structure serves only as a vehicle for describing how a (standard) scenario analysis proceeds in our numerical experiments, but not for the presentation of our proposed algorithms. Knowing the correct market structure, including the factor model and dynamics, is {\em not required} for the ACSA algorithm of Section \ref{sec:ACSA}. Similarly, while the KSA algorithm of Section \ref{sec:KSA} makes stronger assumptions than the ACSA algorithm and incorporates some knowledge of market factors, it does {\em not require} knowledge of the precise factor dynamics.

Let $V_t$ denote the time $t$ value of the portfolio.  The \textit{portfolio gain} (also referred to as the portfolio profit and loss, P\&L for short) at time $t+1$ is $G_{t+1}\defeq V_{t+1} - V_{t}$. In our financial context, time is measured in days so that the portfolio gain would then be a daily P\&L. Assume $V_t$ is
$\cF_t$-adapted where $\mathbb{F}:=\{\cF_t\}_{t\geq 0}$ denotes the filtration generated by all relevant information (e.g. security prices, risk factors, etc) in the market. We use bold font and lower case to denote vectors, bold font and upper case to denote matrices, and non-bold font lower case to denote scalars. All vectors are assumed to be column vectors unless otherwise stated. A superscript $^{(i)}$ denotes the $i^{\mathrm{th}}$ component of a vector, and a superscript $^{(i:j)}$ refers to the sub-vector comprising elements from index $i$ to index $j$. For a matrix, a superscript $^{(i:j,:)}$ denotes the sub-matrix containing rows $i$ through $j$, and $^{(:, i:j)}$ denotes the sub-matrix containing columns $i$ through $j$. More generally, a superscript $^{(i_1:i_2, j_1:j_2)}$ denotes the sub-matrix formed by rows $i_1$ through $i_2$ and columns $j_1$ through $j_2$.

As is standard in the risk management literature (see e.g. \cite{MFE}), we will assume the portfolio value $V_t$ is a function of $\cF_t$-adapted \textbf{risk factors}, whose time $t$ values we denote by the vector $\bx_t \in \mathbb{R}^n$ for some $n$. For the equity and options portfolio in Section~\ref{sec:Examples}, we take $\bx_t$ to consist of the stock prices $S_t^{(i)}$'s and the $M$ implied volatilities that are required to compute the option prices $C_t^{(j)}$'s.
In any case it follows that $V_t = v(\bx_t)$ for some known pricing function $v: \mathbb{R}^n \to \mathbb{R}$, with the portfolio gain given by
\begin{equation}\label{eq:portpnl}
G_{t+1} = v(\bx_{t+1}) - v(\bx_t) = \sum_{i=1}^{n_s} u_i \big( P_{t+1}^{(i)} - P_t^{(i)}\big)
\end{equation}
where $n_s$ is the total number of securities in the portfolio, $u_i$ is the number of units of the $i^{\mathrm{th}}$ security\footnote{We omit the dependence of the time $t$ portfolio holdings on $t$ since they are assumed to be constant in any scenario analysis backtest. For the same reason, we also omit the dependence of the portfolio gain in (\ref{eq:portpnl}) on the portfolio holdings.}, and $P_t^{(i)} = P_t^{(i)}(\bx_t)$ is\footnote{It is tempting to write $P_t^{(i)} = P^{(i)}(\bx_t)$ for some fixed pricing function $P^{(i)}$ but we maintain the dependence on $t$ as many securities have a fixed maturity and so given $\bx_t$, security prices will still in general depend on $t$.} the time $t$ price of the $i^{\text{th}}$ security. We generally suppress the dependence of $P_t^{(i)}$ on $\bx_t$ but note that it typically depends on just a few (and often just one) components of $\bx_t$.

In financial risk management, it is common to employ a \textbf{common factor model}, in which a set of possibly latent factors $\bff_t$ drives the dynamics of observed risk factors $\bx_t$. A typical specification is
\begin{equation}\label{eq:factormodel}
\Delta \bx_{t+1} = \bsym{\mu} + \bB \, \Delta \bff_{t+1} + \bfxi_{t+1}, \quad t = 1,2,\dots,
\end{equation}
where $\bx_t \in \mathbb{R}^n$ is the vector of \textbf{risk factors}, with $\Delta \bx_{t+1} = \bx_{t+1} - \bx_t$, and $\bsym{\mu} \in \mathbb{R}^n$ is its per-period drift. The vector $\bff_t \in \mathbb{R}^m$ contains the \textbf{common factors}, with $\Delta \bff_{t+1} = \bff_{t+1} - \bff_t$. The factor loading matrix is $\bB = [\bb_1 \, \cdots \, \bb_m] \in \mathbb{R}^{n \times m}$, with $\bb_i \in \mathbb{R}^n$ denoting its $i^{\mathrm{th}}$ column. Finally, $\{\bfxi_t\}_{t=1,2,\ldots}$ is an i.i.d.\ sequence of zero-mean random vectors in $\mathbb{R}^n$, representing idiosyncratic shocks that are assumed to be independent of all other variables. In what follows, for any time series-based variable, we always use $\Delta$ to denote its change from time $t$ to $t+1$. The equity and options setting above can be modeled according to (\ref{eq:factormodel}).

\subsection{Standard Scenario Analysis}
\label{sec:SSA}

A fundamental goal of risk management is to understand the distribution of the portfolio gain $G_{t+1}$ or portfolio return $R_{t+1}$ conditional on $\cF_t$. Since computing or estimating the distribution is often too challenging in practice, we often aim instead to compute some functional of the distribution, such as the portfolio's value-at-risk (VaR) or conditional value-at-risk (CVaR). Alternatively, we may wish to conduct a scenario analysis (the topic of this paper) where the goal is to understand the distribution of $G_{t+1}$ conditional on $\cF_t$ and specific time $t+1$ risk scenarios.
In a standard scenario analysis (hereafter SSA), the risk manager would identify various stresses to apply to the changes (i.e., the returns) in the common factors $\bff_t$.
We can define a scenario by jointly stressing any number $k \leq m$ of the common factor returns.
In the case of a portfolio of futures and options on the S\&P 500 in the setting of Section~\ref{sec:Examples} (with a single underlying security), these stresses might include shifts to the value of the underlying, i.e. the S\&P 500, as well as some combination of parallel\footnote{A parallel shift may not be the best choice of stress for an implied volatility surface as it is well known that shorter maturity implied volatilities tend to move a lot more than longer maturity ones. It would, of course, be easy to define a common factor return that handles this phenomenon. It is perhaps worth emphasizing at this point, however, that we are agnostic as to how common factors and associated stresses / scenarios are defined. Rather, our goal in standard scenario analysis is to estimate the (expected) gains {\em given} these scenarios.} shifts to the implied volatility surface and a steepening or flattening of the skew (or term) structure of the implied volatilities. For example, suppose the first component of $\bx_t$ (i.e., $x_t^{(1)}$) represents the log S\&P 500 spot price, and the other components of $\bx_t$ represent the implied volatilities for the options that appear in the portfolio with various strike-maturity combinations.
If $ f_{t}^{(1)}$ also represents the log S\&P 500 spot price, then $\mu^{(1)} = 0, \bb_1 = (1, 0, \ldots, 0)$ and $\xi_t^{(1)}=0$ w.p. 1. Similarly, if the changes in $f_t^{(2)}$ correspond to the parallel shifts to the implied volatility surface, then the second column of $\bB$ would be $\bb_2 = (0 , 1 , \ldots \, 1)$. We can now consider a scenario where the returns of $ f_{t}^{(1)}$ and $ f_{t}^{(2)}$ are simultaneously stressed. For example, a scenario of interest might be one where $\Delta f_{t+1}^{(1)} = -5\%$ and $\Delta f_{t+1}^{(2)} = +10$, corresponding to a 5\% fall in the S\&P 500 spot price and a 10 volatility point increase across its entire implied volatility surface. The SSA approach computes the portfolio gain via (\ref{eq:portpnl}) with $\bx_{t+1}$ determined by (\ref{eq:factormodel}) and the scenario setting: $ \Delta f_{t+1}^{(1)} = -5\%$, $  \Delta f_{t+1}^{(2)} = +10$,  $ \Delta f_{t+1}^{(3)} = \Delta f_{t+1}^{(4)} = 0$ and $\bfxi_{t+1} = \mathbf{0}$.

In practice, multiple two- or even three-dimensional tables of the estimated portfolio gain can be computed, corresponding to the simultaneous stressing of $k=2$ or $k=3$ different common factor returns. It is important to emphasize, especially in the context of portfolios containing derivatives, that the typical risk / portfolio manager employing SSA does not have an explicit model like (\ref{eq:factormodel}) at hand nor does she need one. Indeed, most of the modeling effort in the derivatives space goes into developing dynamic pricing models under an equivalent martingale measure ${\cal Q}$ and then developing numerical procedures for computing prices and the so-called Greeks within these models.
Furthermore, it may be the case that only a subset of factors, say the first $l \leq m$, are ever considered for stressing. In that case, SSA works with a ``model'' of the form
\begin{equation}\label{eq:StanModel}
\Delta {\bx}_{t+1}  = \bsym{\mu} + \bB^{(:, 1:l)}   \Delta \bff_{t+1}^{(1:l)}.
\end{equation}
The important feature of (\ref{eq:StanModel}) is that noise terms, i.e., the $\bfxi_{t+1}$'s, and probability distributions / factor dynamics are not specified and in fact play no role in it.  It therefore follows that (\ref{eq:StanModel}) is {\em not} a probabilistic model for the risk factor ${\bx}_{t+1}$, and this is typically the case in SSA applications.

\subsubsection{The Observability of Common Factor Returns}
\label{sec:Observ}

At this point, we emphasize that throughout the paper, both $\bx_t$ and $\bff_t$ are assumed to be $\cF_t$-adapted and therefore observable at time $t$. This is true for $\bx_{t}$ since we can view $\bx_{t}$ as simply an alternative way to quote security prices, e.g., in terms of bond yields, or log prices and implied volatilities. We also assume that $\bff_t$ is observable since the stress scenarios will be defined in terms of $ \bff_{t}$, and it is necessary to define scenarios and also recognize what actual scenario transpired on any given day. That said, in some applications, a possible objection to this assumption is that some common factors $\bff_t$ might be latent. For example, in so-called cross-sectional or statistical factor models, e.g. \cite{ruppert}, the factors are assumed to be latent and need to be estimated. Clearly, it would be very difficult to evaluate a scenario analysis if some or all of the common factors that are used to define scenarios are latent, since then one can never tell for sure which scenario actually transpired.

An easy solution to this problem is to work with observable estimates or proxies of the common factors and to define scenarios by stressing the returns of these estimates / proxies. For the implied volatility surface of the S\&P 500 index, we can easily define observable proxies for any latent common factors. For example, we can easily use the observed implied volatilities to define the average level of implied volatility, the skew, and the steepness of the term-structure of volatilities, respectively. Hence, the changes in these common factor proxies between $t$ and $t+1$ will be observable. In the remainder of this paper we will therefore assume all common risk factors are observable.

%Ultimately, in order to conduct and evaluate any kind of scenario analysis, we need to be able to observe the stressed common factors $\bff_{t+1}^{(1:l)}$ in order to see what scenario actually transpired.

% In the following context, we do not assume (\ref{eq:factormodel}) is the underlying true model, although knowledge of the matrix $\bB$ in (\ref{eq:factormodel}) is required in order to conduct an SSA (since SSA requires a mechanism that translates a stress on the common factor returns to a stress on the risk factors). At a high level, we can think of the {\em true} relationship being
% $$
% \begin{aligned}
%     \bx_{t+1} &= h\left(\bff_{t+1}\right) + \bfxi_{t+1} \\
%     &= h\left(\bff_{t} + (\bff_{t+1} - \bff_{t})\right) + \bfxi_{t+1}
% \end{aligned}
% $$
% for some possibly non-linear and unknown function $h(\cdot)$, where the function $h(\cdot)$ translates a stress on $\bff_{t+1} - \bff_{t}$ to a stress on $\bx_{t+1}$. In contrast to SSA, the algorithms we propose later in Sections \ref{sec:ACSA} and \ref{sec:KSA} can be viewed as implicitly attempting to {\em learn} the above relationship with a view to providing more accurate portfolio gain estimates as well as coverage guarantees.

\subsection{The Problems with Standard Scenario Analysis}
\label{sec:SSA_Weaknesses}

Before proceeding, we point out several weaknesses of SSA which of course relate to the fact that (\ref{eq:StanModel}) is not a probabilistic model for the risk factors ${\bx}_{t+1}$. Some of these weaknesses were also highlighted by \cite{Haugh2020}.
\begin{enumerate}
\item SSA only produces a point estimate - typically interpreted as an expectation - of the portfolio gain in a given scenario. But an expectation only gives a partial view on the gain risk in a given scenario and it would be far more informative to provide a prediction interval for the realized scenario gain.

\item Let $\bs = \{s_1, \ldots , s_k\}$ where $1\leq s_1 < s_2 \leq \cdots < s_k \leq l$, and let $ \bff^{(\bs)}_{t+1} \defeq (f^{(s_1)}_{t+1}, \ldots, f^{(s_k)}_{t+1})$ denote the subset of common factors whose returns are stressed in a given scenario.  SSA implicitly assumes that the unstressed common factors are unchanged between $t$ and $t+1$, i.e., $ \Delta \bff^{(\bs^c)}_{t+1} = \mathbf{0}$ (we use $\bs^c$ to denote the complement of $\bs$, i.e. the indices of the non-stressed common factors in the scenario). But this is not justified and can often lead to a very inaccurate estimation for the given scenario.

%\footnote{We use $\mE_t[\cdot]$ throughout to denote expectations that are conditional on $\cF_t$. Hence $\mE_t[\cdot \mid \bff^{(\bs)}_{t+1}]$ denotes expectations that are conditional on $\cF_t$ {\em and } $\bff^{(\bs)}_{t+1}$.}

\item An obvious solution to the problem raised in the previous point would be to set the unstressed factors $ \bff^{(\bs^c)}_{t+1}$ equal to their conditional expectation $\mE_t[ \bff^{(\bs^c)}_{t+1} \mid  \bff^{(\bs)}_{t+1}]$\footnote{We use $\mE_t[\cdot]$ throughout to denote expectations that are conditional on $\cF_t$. Hence $\mE_t[\cdot \mid \bff^{(\bs)}_{t+1}]$ denotes expectations that are conditional on $\cF_t$ {\em and } $\bff^{(\bs)}_{t+1}$.}  when estimating the portfolio gain. This should certainly be an improvement over SSA, e.g. see \cite{Haugh2020}, but it ignores the uncertainty in $\bfxi_{t+1}$ and $\bff^{(\bs^c)}_{t+1} \mid (\cF_t, \,\bff^{(\bs)}_{t+1}) $. This uncertainty may be significant, particularly for portfolios containing securities whose values depend non-linearly on ${\bx}_{t+1}$.
But even if we ignore this uncertainty, setting $\bff^{(\bs^c)}_{t+1} = \mE_t[ \bff^{(\bs^c)}_{t+1} \mid  \bff^{(\bs)}_{t+1}]$ is not a straightforward task as it requires a model for estimating these conditional expectations.

\item SSA is rarely if ever accompanied by any sort of backtesting or statistical validation. Perhaps the main reason for this is that each of the scenarios considered by an SSA is (typically) a zero probability event, and none of the considered scenarios will have actually occurred on day $t+1$. However, this problem is easy to resolve. Specifically, on day $t+1$ we could ``see'' exactly what scenario has just transpired. For example, we could see what the return in the S\&P 500 and what the parallel change in the implied volatility surface were over the period $[t,t+1]$. We could then roll time back one day, rerun the scenario analysis for exactly this realized scenario and then compare the estimated and realized portfolio gains.

\item Another weakness with SSA is that it relies on a reduced-form model such as (\ref{eq:StanModel}).  But if a portfolio is constructed with a view to (i) being neutral relative to all the common factor returns that are ever considered for stressing but (ii) highly exposed to non-stressed common factor returns, then an SSA will drastically underestimate the risk of the portfolio. This, of course, is related to other points raised above but the problem here is that the risk manager may not even be aware of the non-stressed common factors, whereas the trader or portfolio manager is aware of them. We refer to this as {\em adversarial portfolio selection}.  A more benign possibility is simply that one or more of the common factors are misspecified, which can also lead to reported expected scenario gains that are misleading.

%For example, instead of recognizing that short-term implied volatilities tend to move a lot more than longer-term implied volatilities, we might mistakenly assume that the first implied volatility common factor models {\em parallel} moves in the implied volatility surface. A portfolio might then appear to be neutral to parallel moves but in fact be quite exposed to such moves. To see this, consider a portfolio of options that is long short-term implied volatility but takes an off-setting position in long-term implied volatility options. If all implied volatilities move up or down by the same amount then the portfolio will be (approximately neutral). However, the observed parallel move will simply be the average move of  implied vols across all pre-specified strike-maturity combinations. If we condition on this (as we would do in a scenario), then the resulting conditional distribution would likely show more movement in the short-term vols due to a square-root of time ground truth.
\end{enumerate}

\cite{Haugh2020} addressed some of these weaknesses by embedding the scenario analysis within a dynamic factor model for the underlying risk factors. Their approach required multivariate state-space models for modeling the time-series behavior of (sometimes latent) risk factors, and that were sufficiently tractable as to enable the computation (possibly via simulation) of the conditional distribution of unstressed risk factors. They showed how SSA and their approach could lead to dramatically different results, particularly in the case of adversarial portfolio selection.
They adopted a parametric approach, however, which meant their reported expected scenario gains were only accurate to the extent that their modeling was correct. In contrast, in this paper we propose a more holistic approach to scenario analysis that does not rely on the correct specification of a parametric model defining the risk factor and common risk factor dynamics. In addition, a key feature of our framework is the provision of coverage guarantees for the prediction intervals we report for our portfolio gains.

\section{A Motivating Example: A 3-Factor Log-Linear Equity Model} \label{sec:MotivatingEG}

In this section, we present a simple motivating example that illustrates how inaccurate SSA can be. We then revisit this example in Sections \ref{sec:ACSA} and \ref{sec:KSA}, where we show how ACSA and KSA produce calibrated prediction intervals for the realized gain under each scenario.

%Throughout the paper, we evaluate the portfolio gain for each day $t+1$ using the previous day’s ($t$) portfolio value normalized to a constant $V_0$ (e.g., $5000$). This quantity is equivalent to the daily portfolio return, but for consistency with the commonly used terminology, we refer to it as the portfolio gain.

\paragraph{Market Dynamics}
We assume the market consists of a single non-dividend paying equity index, e.g. the S\&P 500, with time $t$ price $S_t$ and daily return dynamics that are IID and satisfy
\begin{equation} \label{eq:StockReturn1}
\log \left(\frac{S_{t+1}}{S_t}\right)
      = \delta\cdot  \left(\mu_\delta-\frac{\sigma_\epsilon^2}{2}\right) + \bsym{\beta}^{\top}\bF_{t+1} + \epsilon_{t+1},
\end{equation}
where $\delta = 1/252$ denotes the length in years of one trading day,
$\mathbf F_{t+1} \sim \mathcal{N}(\mathbf 0,\delta\, \bsym{\Sigma})$
is a 3-dimensional observable factor return, with $\bsym{\Sigma}$ denoting
its annualized covariance matrix, and
$\epsilon_{t+1} \sim \mathcal{N}\bigl(0,\delta\,\sigma_\epsilon^2\bigr)$
is an idiosyncratic random noise variable with annualized volatility
$\sigma_\epsilon$.
The $-\sigma_\epsilon^2/2$ term in \eqref{eq:StockReturn1} reflects the usual
It\^{o} correction associated with
a geometric Brownian motion type specification.\footnote{For a geometric Brownian motion satisfying
$\frac{\mathrm{d}S_t}{S_t}=\mu\,\mathrm{d}t+\sigma\,\mathrm{d}W_t$,
It\^{o}'s lemma gives
$\mathrm{d}\log S_t=(\mu-\sigma^2/2)\,\mathrm{d}t+\sigma\,\mathrm{d}W_t$.}
The common-factor contribution $\boldsymbol{\beta}^{\top}\mathbf F_{t+1}$ is specified in log-return form. Accordingly, we
interpret $\mu_\delta$ as the annualized drift parameter appearing in
\eqref{eq:StockReturn1}.
We interpret the three factor returns as \textit{oil}, \textit{interest rate},
and \textit{credit spread} shocks. We take $\mu_\delta = 0.03$, and the factor
loadings $\boldsymbol{\beta} =: (\beta_1,\beta_2,\beta_3)$, together with their
interpretations, are given in Table~\ref{tb:Loadings}.

\begin{table}[h]
\begin{center}
\begin{tabular}{@{}lccc@{}}
\toprule
Factor & Daily observable shock & $\beta_k$ \\ \midrule
Oil & \% daily change in Brent/WTI spot & $-0.20$ \\
Rate & Daily change in 10y U.S.\ yield (bp) & $-0.15$ \\
Credit & Daily change in ICE/BofA US HY OAS (bp) & $-0.30$ \\ \bottomrule
\end{tabular}
\end{center}
\vspace{-0.5cm}
\caption{The three factor loadings and their interpretation.}
\label{tb:Loadings}
\end{table}
We note that \eqref{eq:StockReturn1} agrees with the model structure specified in \eqref{eq:factormodel}, with $\bx_t = \log S_t$, $\bsym{\mu} =  \delta\cdot  \left(\mu_\delta-\frac{\sigma_\epsilon^2}{2}\right)$, $\bB = \bsym{\beta}^\top$, $\Delta \bff_{t+1} = \bF_{t+1}$ and $\bsym{\xi}_{t+1} = \epsilon_{t+1}$. It follows from (\ref{eq:StockReturn1}) that the $T$-day return satisfies
\begin{equation} \label{eq:StockReturn2}
\log \left(\frac{S_{T}}{S_0}\right)
      = T\cdot \delta\,  \left(\mu_\delta-\frac{\sigma_\epsilon^2}{2}\right) + \bsym{\beta}^\top\bF_{1:T} + \epsilon_{1:T}
\end{equation}
where   $\mathbf F_{1:T} = \sum_{t=1}^T F_t \sim \mathcal{N}(\mathbf 0,T \cdot \delta\, \bsym{\Sigma})$ is the $T$-day factor return and
$\epsilon_{1:T} =\sum_{t=1}^T \epsilon_{t} \sim \mathcal{N} \Bigl(0,\,
T \cdot \delta\, \sigma_\epsilon^2\Bigr)$. We assume $S_0=5000$ and the annual covariance matrix (corresponding to taking $T=252$ days so that $T\cdot \delta = 1$) is given by $\bsym{\Sigma}=
\begin{bmatrix}
0.0900 & 0.0009 & -0.0240\\
0.0009 & 0.0001 & -0.0016\\
-0.0240 & -0.0016 & 0.1600
\end{bmatrix}.$ This in turn implies correlations  $\rho_{\text{\tiny Oil,Rate}}=0.30$, $\rho_{\text{\tiny Oil,Credit}}=-0.20$ and $\rho_{\text{\tiny Rate,Credit}}=-0.40$. Finally, we take $\sigma_\epsilon = 18\%$ so that the total annual volatility is
\begin{equation}\label{eq:expression_of_sig_total}
\sigma_{\text{tot}} = \sqrt{\bsym{\beta}^{\mathsf T}\bsym{\Sigma}\boldsymbol{\beta}+ \sigma_{\epsilon}^2} \approx 21.8\%.
\end{equation}
Throughout this example, $\bsym{\Sigma}$ and $\sigma_\epsilon^2$ denote
annualized covariance and variance parameters, respectively. One-day variances are obtained by multiplying the corresponding annualized
quantities by $\delta=1/252$.

\paragraph{Portfolio}
Our portfolio consists of a single short position in a European call option on the index with a time to maturity of one month (i.e., $\tau = 21/252$) and a fixed moneyness (the strike price divided by the underlying asset's price) $m=1.05$, i.e., the call option is 5\% out-of-the-money. Given that $S_0=5000$, the strike price would be $K = 1.05 \times S_0 = 5,250$. We take the continuously compounded interest rate to be $r=4\%$ so that the initial value of the option is then given by the Black-Scholes formula and equals $\$44.2$. (The use of the Black-Scholes formula to price options is consistent with the log-normal dynamics of $S_T$ as evidenced by (\ref{eq:StockReturn2}).)

\paragraph{Scenario Analysis} We consider stress testing the portfolio over a single trading day by considering joint stresses to the oil and rate factors.
In particular, we consider percentage shocks of size $\{-20, -12, -6, 0, +6, +12, +20\}$ to the oil factor and of size $\{-30, -20, -10, 0, +10, +20, +30\}$ b.p.'s to the rate factor. We use $\bz = (z_{\text{\tiny Oil}}, z_{\text{\tiny Rate}})$ to denote a scenario so, for example, $\bz = (+12\%,-20\, \mathrm{b.p.})$ denotes a stress of $+12\%$ to the oil factor return and $-20\, \mathrm{b.p.}$ to the rate factor return. While some of these stresses are considerably more extreme than would be implied by their 1-day volatilities,  scenario analysis stresses in practice are often much larger than those implied by the size of typical daily moves. For example, oil price shocks of -20\% occurred in 2020 during the COVID crash, and +15\% days occurred on days of OPEC supply cuts. Similarly, $\pm 30$ b.p. moves in 10-year yields can occur when the U.S. Federal Reserve surprises the market, and $\pm 10$ b.p. moves are quite common on FOMC days.

\paragraph{Computing the Conditional 1-Day Log Return Distribution}

The credit factor is not stressed in any scenario, and its conditional mean and annualized conditional variance in a given scenario $\bz$ are
\begin{eqnarray}\label{eq:define_of_condmu}
\mu_{\text{\tiny Cred}|s} &=&\boldsymbol{\Sigma}_{\text{\tiny Cred},s}\boldsymbol{\Sigma}_{ss}^{-1}\bz \label{eq:CondCreditMean} \\
\sigma^{2}_{\text{\tiny Cred}|s} &=& \Sigma_{\text{\tiny Cred,Cred}}
-\bsig_{\text{\tiny Cred},s}\bsig_{ss}^{-1}\bsig_{s,\text{\tiny Cred}} \nonumber
\end{eqnarray}
where $\Sigma$ is partitioned according to
$
\bsym{\Sigma} = \left[
\begin{array}{ll}
 \bsig_{ss} & \bsig_{s,\text{\tiny Cred}} \\
\bsig_{\text{\tiny Cred},s} & \Sigma_{\text{\tiny Cred}, \text{\tiny Cred}}
\end{array}
\right].
$
We note that $\mu_{\text{\tiny Cred}\mid s}$ and $\sigma_{\text{\tiny Cred}\mid s}$ depend on the stress in distinct ways: $\sigma_{\text{\tiny Cred}\mid s}$ is determined by the {\em choice} of stressed factors (i.e., the subset $\bs$ of common factors), whereas $\mu_{\text{\tiny Cred}\mid s}$ also reflects the specific stress vector $\bz$ applied to those factors. Let $\mathbf{F}_s=(F_{\text{\tiny Oil}},F_{\text{\tiny Rate}})$ be the
stressed pair\footnote{We allow for a slight abuse of notation here. In particular, $\mathbf{F}_s \in \mathbb{R}^2$ denotes the two stressed factor returns over a single day while we have also used $\mathbf{F}_{t}  \in \mathbb{R}^3$ to denote the random factor return vector between times $t-1$ and $t$.} of factors and $F_{\text{\tiny Cred}}$ the unstressed credit–spread shock.
Then
\begin{equation}\label{eq:logret_cond_dist}
    \log\left(S_{t+1}/S_t\right)\Bigl|\,
\left(\mathbf F_s=\bz \right)
\;\sim\;
\mathcal{N}\bigl(m_{\bz},\,v_{\text{\tiny Cred}}\bigr),
\end{equation}
where
\begin{equation}\label{eq:define_of_ms}
    m_{\bz}=
\delta\,  \left(\mu_\delta-\frac{\sigma_\epsilon^2}{2}\right)
+\beta_{\text{\tiny Oil}}\,z_{\text{\tiny Oil}}
+\beta_{\text{\tiny Rate}}\,z_{\text{\tiny Rate}}
+\beta_{\text{\tiny Cred}}\,\mu_{\text{\tiny Cred}\mid s},
\quad
v_{\text{\tiny Cred}}=  \delta\cdot \left(\beta_{\text{\tiny Cred}}^{2} \sigma_{\text{\tiny Cred} | s}^2
+\sigma_\epsilon^{2}\right).
\end{equation}

% Due to $\Delta \ll 1$, we can approximate the conditional distribution with
% $$
% m_{\bz}\approx \beta_{\text{\tiny Oil}}\,f_{\text{\tiny Oil}}
% +\beta_{\text{\tiny Rate}}\,f_{\text{\tiny Rate}}
% +\beta_{\text{\tiny Cred}}\,\mu_{\text{\tiny Cred}\mid s},\qquad v_s\approx \delta\cdot\sigma_\epsilon^{2}.
% $$
\begin{comment}
and we have omitted\footnote{This term makes no noticeable difference to the option prices and portfolio gains under the various scenarios in the tables below. It is also why we don't need to specify $\alpha$.} a negligible term in $m_{\bz}$ that arises due to time moving forward $1$ day when we reprice the options in the various scenarios.
\end{comment}
\paragraph{Repricing the Short Option Position} We can price the portfolio option using the Black-Scholes formula in every scenario, i.e., assuming a continuously compounded risk-free interest rate of $r=4\%$ and volatility $\sigma_{\text{tot}}$. Because the scenario stresses roll time forward by 1 day, the time to maturity\footnote{The short option position in the portfolio means the portfolio gains approximately $2.55$ dollars at the zero-stress spot level from the one-day theta effect.} when repricing the options in each scenario is now 20/252.

\paragraph{Results}Table \ref{tab:SSA} reports the estimated scenario P\&L or portfolio gain (in dollars) for SSA and conditional-mean scenario analysis (CSA). The CSA entries are obtained by setting the unstressed credit factor to its mean conditional on the scenario, i.e. $\mu_{\text{\tiny Cred}|s}$ from (\ref{eq:CondCreditMean}). Because the idiosyncratic noise $\epsilon$ is uncorrelated with the factor noise, we set it to its mean of zero in each scenario. Finally, in Table \ref{tab:Convex} we present results for the oracle scenario analysis where the portfolio gains are obtained via (\ref{eq:logret_cond_dist}) by integrating the Black-Scholes price of the option over the residual uncertainty of the credit factor conditional on the scenario.

\begin{table}[ht]
\centering
\setlength{\tabcolsep}{1.5pt}
\renewcommand{\arraystretch}{1.08}
\begin{tabular*}{\textwidth}{@{\extracolsep{\fill}}r|*{7}{r}|*{7}{r}@{}}
\hline
 & \multicolumn{7}{c|}{\bfseries SSA} & \multicolumn{7}{c}{\bfseries CSA} \\
\hline
 & \multicolumn{7}{c|}{Oil (\%)} & \multicolumn{7}{c}{Oil (\%)} \\
Rate (bp) & $-20$ & $-12$ & $-6$ & $0$ & $+6$ & $+12$ & $+20$
& $-20$ & $-12$ & $-6$ & $0$ & $+6$ & $+12$ & $+20$ \\
\hline
$-30$
& -70.7 & -35.3 & -14.5 & 1.9 & 14.6 & 24.0 & 32.6
& -27.0 & -6.2 & 6.1 & 15.9 & 23.5 & 29.4 & 35.0 \\

$-20$
& -70.3 & -35.0 & -14.3 & 2.1 & 14.7 & 24.1 & 32.6
& -35.0 & -12.5 & 1.0 & 11.9 & 20.4 & 27.0 & 33.4 \\

$-10$
& -69.9 & -34.8 & -14.0 & 2.3 & 14.8 & 24.2 & 32.7
& -43.7 & -19.3 & -4.6 & 7.4 & 16.9 & 24.4 & 31.6 \\

$0$
& -69.6 & -34.5 & -13.8 & 2.5 & 15.0 & 24.3 & 32.8
& -53.0 & -26.7 & -10.6 & 2.5 & 13.0 & 21.3 & 29.5 \\

$+10$
& -69.2 & -34.2 & -13.6 & 2.7 & 15.1 & 24.4 & 32.8
& -63.0 & -34.7 & -17.3 & -2.9 & 8.7 & 17.9 & 27.1 \\

$+20$
& -68.8 & -33.9 & -13.4 & 2.8 & 15.2 & 24.5 & 32.9
& -73.7 & -43.3 & -24.5 & -8.9 & 3.9 & 14.2 & 24.5 \\

$+30$
& -68.5 & -33.6 & -13.1 & 3.0 & 15.4 & 24.6 & 32.9
& -85.0 & -52.6 & -32.4 & -15.4 & -1.4 & 10.0 & 21.5 \\
\hline
\end{tabular*}
\caption{Scenario gains in dollars, rounded to one decimal place. The left panel reports standard scenario analysis (SSA), which sets the credit factor to zero. The right panel reports conditional-mean scenario analysis (CSA), which sets it to its scenario-conditional mean. Both set independent idiosyncratic noise to zero. Rate stresses are in basis points (bp) and oil stresses are in percent.}
\label{tab:SSA}
\label{tab:CondMean}
\end{table}

\begin{table}[ht]
\centering
\begin{tabular}{r|rrrrrrr}
\hline
Rate$\backslash$Oil & $-20\%$ & $-12\%$ & $-6\%$ & $0\%$ & $+6\%$ & $+12\%$ & $+20\%$\\
\hline
$-30$ bp & -29.85 & -8.73 & 3.84 & 13.94 & 21.91 & 28.08 & 34.03\\
$-20$ bp & -37.98 & -15.10 & -1.35 & 9.80 & 18.66 & 25.58 & 32.32\\
$-10$ bp & -46.73 & -22.02 & -7.04 & 5.21 & 15.04 & 22.77 & 30.38\\
$0$ bp   & -56.11 & -29.52 & -13.25 & 0.16 & 11.01 & 19.62 & 28.18\\
$+10$ bp & -66.14 & -37.62 & -20.01 & -5.38 & 6.55 & 16.10 & 25.69\\
$+20$ bp & -76.84 & -46.34 & -27.34 & -11.44 & 1.64 & 12.19 & 22.90\\
$+30$ bp & -88.20 & -55.69 & -35.27 & -18.04 & -3.76 & 7.86 & 19.76\\
\hline
\end{tabular}
\caption{Oracle Scenario Analysis. Expected value of option is computed in each scenario by
integrating the Black-Scholes price of the option over the residual uncertainty of the credit factor
and the idiosyncratic noise $\epsilon$ in that scenario.}
\label{tab:Convex}
\end{table}

Several observations are in order. First, we note there is a considerable difference in the scenario P\&Ls across the three methods, illustrating just how inaccurate an SSA can be. We can interpret the improvement from the SSA entries to the CSA entries in Table \ref{tab:SSA} as being due to a correlation correction induced by setting the credit factor $F_{\text{\tiny Cred}}$ to its conditional mean $\mu_{\text{\tiny Cred}|s}$ in each scenario. This correction can swing the  P\&L by up to approximately $\$44$, and generally results in less extreme scenario P\&L's because the credit factor is negatively correlated with the oil and rate factors and therefore provides a partial hedge in extreme scenarios. The differences (approximately $\$1.0$--$\$3.2$)  between the CSA entries in Table \ref{tab:SSA} and Table \ref{tab:Convex} can be interpreted as a convexity correction. In particular, while the CSA entries contain the correct ``delta'' correction (in contrast to the SSA entries), Table \ref{tab:Convex} may be viewed as having an additional ``gamma'' correction by explicitly recognizing the uncertainty in the index price in each scenario. Because the portfolio is short a European option and therefore short gamma, we see that the P\&L's in Table \ref{tab:Convex} are always smaller than
the corresponding CSA P\&L's in Table \ref{tab:SSA}. We also mention that the $(0,0)$ entries are not zero for either SSA or CSA in Table \ref{tab:SSA}. This is due to the ``theta'' effect whereby the portfolio gains when the short option position has 1 day less to maturity.

It's worth remarking that the difference between the (inaccurate) SSA entries in Table~\ref{tab:SSA} and (the accurate) Table~\ref{tab:Convex} (or even between the CSA entries in Table~\ref{tab:SSA} and Table~\ref{tab:Convex}) can be much greater depending on the portfolio in question. For example, an adversarial portfolio (as discussed in weakness \# 5 of SSA in Section \ref{sec:SSA_Weaknesses}) could yield wildly inaccurate estimates of scenario gains when these estimates are conducted via SSA.
Table \ref{tab:Convex} requires knowing the conditional distribution in each scenario (and integrating over residual uncertainty), which is infeasible in realistic markets. Sections \ref{sec:ACSA}-\ref{sec:KSA} show how ACSA and KSA learn scenario-wise prediction intervals from data and then calibrate them via backtesting, avoiding reliance on a correct parametric specification.

\section{Adaptive Conformal Scenario Analysis}
\label{sec:ACSA}

In this section, we introduce our {\bf adaptive conformal scenario analysis (ACSA)} algorithm. This algorithm resembles the adaptive conformal prediction algorithm recently proposed by \citet{gibbs2021adaptive}, and we apply it with some minor modifications to our scenario analysis setting. The major advantages of the ACSA algorithm are its simplicity and its wide applicability. The main downside of the algorithm is that its theoretical guarantee is only a marginal guarantee, which is not entirely satisfactory in a risk management setting where conditional guarantees would be more desirable. This desire for an algorithm with conditional guarantees motivated our development of the {\bf kernel scenario analysis (KSA)} algorithm in Section \ref{sec:KSA}. Nevertheless, as we shall explain later, we view the two algorithms as complementary to each other.

We use $\bz\in \mR^{k}$ to represent a specific stress scenario so that, for example, if a stress $\bz$ is applied to $\bff_t^{(\bs)}$, then the resulting common factors after being stressed satisfy $\Delta \bff_{t+1}^{(\bs)} = \bz$ (i.e., $\bff_{t+1}^{(\bs)} = \bff_t^{(\bs)} + \bz$). Take the equity and options portfolio in Section~\ref{sec:Examples}, for instance. If we apply the stress as specified in Section \ref{sec:SSA}, then $\bz = (-5\%, +10)$. We also use $\bz_{t+1}$ for the {\em actual scenario} that transpired at time $t+1$: once the realization of $\bff_{t+1}$ is observed, the corresponding value of $\bz_{t+1}$ is also revealed as $\bz_{t+1} = \bff_{t+1}^{(\bs)} - \bff_{t}^{(\bs)} = \Delta \bff_{t+1}^{(\bs)}$.
The ACSA algorithm (and indeed the KSA algorithm of Section \ref{sec:KSA}) will take as input a miscoverage rate $\alpha \in (0,1)$ and then aim to provide a $100 \times (1-\alpha)\%$ prediction interval for $G_{t+1}\mid (\Delta\bff_{t+1}^{(\bs)} =\bz,{\cal F}_t)$ (which we generally abbreviate as $G_{t+1}\mid \bz$) for each scenario $\bz$ under consideration.

\begin{comment}
\begin{rem}
We could compute a prediction interval for a fine grid of miscoverage rates $0 < \alpha_1 < \alpha_2 < \cdots < \alpha_M < 1$ and then use these intervals to approximate the distribution of  $G_{t+1}\mid \bz$ and estimate $\mE_t[G_{t+1}\mid \bz\,]$. Here and in Section \ref{sec:KSA}, however, we will focus on a single miscoverage rate $\alpha$.
\end{rem}
\end{comment}

\subsection{The ACSA Algorithm}
\label{sec:ACSA_Alg}

We present ACSA in Algorithm~\ref{alg:ACSA}. At each time $t$, a quantile
predictor $g_t(\bz;q)$ is used to construct a prediction interval for the
next-period portfolio gain $G_{t+1}$ under each scenario $\bz$. The target
miscoverage rate is $\alpha$, but ACSA does not keep the quantile level fixed
at $\alpha$. Instead, it maintains an adjusted level $\tilde{\alpha}_t$ and
updates it according to whether the most recent interval covered the realized
portfolio gain: a miss decreases $\tilde{\alpha}_t$ and therefore tends to
widen subsequent intervals, whereas successful coverage increases
$\tilde{\alpha}_t$ and tends to shrink them.

For any adjusted level $\tilde{\alpha}\in\mathbb R$, define the prediction set by
\begin{equation}
\label{eq:acsa_prediction_set}
\Ct_t(\bz;\tilde{\alpha})=
\begin{cases}
\mathbb R, & \tilde{\alpha}\leq 0,\\
\big[
\min\{g_t(\bz;\tilde{\alpha}/2),\,g_t(\bz;1-\tilde{\alpha}/2)\},\,\,
\max\{g_t(\bz;\tilde{\alpha}/2),\,g_t(\bz;1-\tilde{\alpha}/2)\}
\big],
& 0<\tilde{\alpha}<1,\\
\emptyset, & \tilde{\alpha}\geq 1.
\end{cases}
\end{equation}
The boundary cases in (\ref{eq:acsa_prediction_set}) ensure certain
coverage when $\tilde{\alpha}\leq0$ and certain miscoverage when
$\tilde{\alpha}\geq1$, as required by the proof of
Theorem~\ref{thm:empirical_coverage}. If the response is known to lie
almost surely in a finite interval, that interval can replace
$\mathbb R$ without changing the coverage indicator. In our numerical experiments, when
$\widetilde{\alpha}_t\leq 0$, we replace $\mathbb{R}$ by the empirical range of
historical portfolio gains recomputed under the current portfolio configuration\footnote{Since this empirical range is
not an almost-sure bound on future gains, this finite-interval substitution is an
implementation approximation of \eqref{eq:acsa_prediction_set}.
Additional implementation details are provided in E-Companion \ref{app:algo_details}}.

Apart from the mild regularity conditions stated in Assumption \ref{assump:g_t_basicAssump} below, the quantile predictor
$g_t$ in \eqref{eq:acsa_prediction_set} is otherwise unrestricted. It may be
obtained from linear or nonlinear quantile regression, a neural network, or the
KSA procedure introduced later in Section~\ref{sec:KSA}. The full procedure is
summarized in Algorithm~\ref{alg:ACSA}.

\begin{algorithm}[ht!]
    \caption{Adaptive Conformal Scenario Analysis (ACSA)}
    \label{alg:ACSA}
    \begin{algorithmic}[1]
    \Require Target coverage rate $1-\alpha$, step size $\gamma>0$,
    initial quantile predictor of the portfolio gain $g_1(\cdot)$,
    current time index $T$

    \State Set $\tilde{\alpha}_1 = \alpha$
    \label{line:OSA_init}

    \For{$t=2,\ldots,T$}
        \State Observe the realized portfolio gain $G_t$ and realized stress scenario $\bz_t$. Compute
        \label{line:cal_error}
        \begin{equation}
        \label{eq:error_t}
            \mathrm{err}_{t}
            \gets
            \bbone\{G_t\notin
            \Ct_{t-1}(\bz_t;\tilde{\alpha}_{t-1})\}
        \end{equation}

        \State Update $\tilde{\alpha}_{t}$ according to
        \label{line:updateBeta}
        \begin{equation}
        \label{eq:OSA_update}
            \tilde{\alpha}_{t}
            \gets
            \tilde{\alpha}_{t-1}
            +\gamma\cdot(\alpha-\mathrm{err}_{t})
        \end{equation}

        \State Update / retrain the quantile predictor $g_t$, if desired
        \label{line:alg_1_updateg}

        \State \textbf{Output:} For any stress scenario $\bz$, produce the prediction set following \eqref{eq:acsa_prediction_set}
        \label{line:produce_OSAinterval}
        \begin{equation}
        \label{eq:OSAinterval}
            \Ct_t(\bz)
            \gets
            \Ct_t(\bz;\tilde{\alpha}_t)
        \end{equation}

    \EndFor
    \end{algorithmic}
\end{algorithm}

When $0<\tilde{\alpha}_t<1$, the two endpoints in
\eqref{eq:acsa_prediction_set} correspond to estimated conditional
quantiles of $G_{t+1}$ under scenario $\bz$. If $g_t$ were perfectly
calibrated, the resulting interval would therefore have approximately
$100(1-\tilde{\alpha}_t)\%$ conditional coverage. ACSA does not require
this perfect calibration to be true. Its purpose is precisely to compensate for systematic
miscalibration in the underlying predictor. To see the feedback mechanism, suppose the interval reported at time
$t-1$ fails to cover $G_t$. Then $\mathrm{err}_t=1$ and
$\tilde{\alpha}_t
=
\tilde{\alpha}_{t-1}-\gamma(1-\alpha),$ so the adjusted miscoverage level decreases and the next interval will
typically be wider. Conversely, if the previous interval covers $G_t$,
then $\mathrm{err}_t=0$ and $\tilde{\alpha}_t
=
\tilde{\alpha}_{t-1}+\gamma\alpha$, which tends to make the next interval narrower. Thus ACSA continually
uses realized coverage errors to correct the quantile level supplied to
$g_t$. For example, suppose intervals constructed using the nominal
miscoverage level $\alpha=10\%$ achieve only $80\%$ coverage rather than
the desired $90\%$. Repeated misses will push $\tilde{\alpha}_t$ below
$10\%$, causing the procedure to request more extreme lower and upper
quantiles from $g_t$. An adjusted level near $5\%$, for instance, may
produce intervals whose nominal coverage is $95\%$ but whose realized
coverage is closer to the desired $90\%$. The adjusted level should
therefore be interpreted as a calibration parameter rather than as the
true miscoverage probability of the underlying quantile model. A more accurate predictor $g_t$ will generally yield sharper intervals,
so the quality of $g_t$ remains important in practice. The coverage
guarantee of ACSA, however, does not rely on $g_t$ being an accurate
conditional-quantile estimator. We require only the following basic
regularity condition.

\begin{assum}
\label{assump:g_t_basicAssump}
For each $t\geq1$, the predictor
$g_t:\mathbb R^k\times(0,1)\to\mathbb R$ is finite-valued and measurable
with respect to the information available at time $t$ and its scenario
and quantile-level arguments. The realized portfolio gains $G_t$ are
finite almost surely.
\end{assum}
Before stating the coverage result, we give two simple examples that
illustrate the dynamics of the adjusted level. A third example comparing
conditional and stationary quantile predictors at a common adjusted
level is provided in E-Companion~\ref{subapx:complex_example_ACSA}. For
clarity, we do not condition on any scenario in these examples.
\vspace{-0.15cm}
\begin{example}[Exact Conditional Quantiles]
\label{eg:Ideal1}
\sffamily\upshape
Suppose the conditional distribution of $G_t$ given
$\mathcal F_{t-1}$ is continuous and $g_{t-1}$ returns its exact
conditional quantiles at every level in $(0,1)$. Define $\pi(a):=\min\{1,\max\{0,a\}\}.$ For the prediction set defined in (\ref{eq:acsa_prediction_set}), the
conditional miscoverage probability is $\pi(\tilde{\alpha}_{t-1})$:
it is zero when $\tilde{\alpha}_{t-1}\leq0$ and one when
$\tilde{\alpha}_{t-1}\geq1$.
Hence
\begin{equation}
\label{eq:ACSA_Iter1}
\mathbb E_{t-1}[\tilde{\alpha}_t]
=
\tilde{\alpha}_{t-1}
+\gamma\bigl(\alpha-\pi(\tilde{\alpha}_{t-1})\bigr).
\end{equation}
When $0\leq\tilde{\alpha}_{t-1}\leq1$,  (\ref{eq:ACSA_Iter1}) gives
$\mathbb E_{t-1}[\tilde{\alpha}_t]
=(1-\gamma)\tilde{\alpha}_{t-1}+\gamma\alpha$. Thus, under the exact conditional-quantile predictor assumed here,
the adjusted level exhibits mean reversion toward the nominal level
$\alpha$: its conditional expected change is positive when
$\tilde{\alpha}_{t-1}<\alpha$ and negative when
$\tilde{\alpha}_{t-1}>\alpha$. If
$\tilde{\alpha}_{t-1}\leq0$, the next update increases the level by
$\gamma\alpha$; if $\tilde{\alpha}_{t-1}\geq1$, it decreases the level
by $\gamma(1-\alpha)$. Thus the boundary behavior naturally pushes the
adjusted level back toward $(0,1)$.
\end{example}\medskip
%\vspace{-0.6cm}
\begin{example}[Constant Quantile Predictor]
\label{eg:Ideal2}
\sffamily\upshape
Suppose $g_t(\bz;u)=C$ for every $t$, $\bz$, and $u\in(0,1)$, where
$C\in\mathbb R$ is fixed, and suppose each $G_t$ has a continuous
distribution. Whenever $0<\tilde{\alpha}_t<1$, the prediction set is
the singleton $\{C\}$ and therefore misses almost surely. Repeated
misses decrease $\tilde{\alpha}_t$ until it becomes nonpositive, at
which point the prediction set becomes $\mathbb R$ and the next
observation is covered with probability one. The subsequent update then
moves $\tilde{\alpha}_t$ upward again. The adjusted levels $\tilde{\alpha}_t$ remain bounded, so the telescoping
identity (\ref{eq:acsa-telescoping-identity}) from the proof of Theorem~\ref{thm:empirical_coverage}  implies that the empirical
coverage rate converges almost surely to $1-\alpha$, as stated in
Theorem~\ref{thm:empirical_coverage}. This example also
shows that coverage alone does not imply informative prediction
intervals: a poor predictor can satisfy the marginal guarantee only
through highly conservative boundary behavior.
\end{example}\medskip
\vspace{-0.4cm}
We are now ready to state the main theoretical guarantee. The result is
an application of the adaptive conformal argument of
\citet{gibbs2021adaptive}; the contribution here is its scenario-analysis
formulation and interpretation rather than a new conformal theorem.
At each time $t$, ACSA reports a collection of intervals, one for every
scenario $\bz$ under consideration. The coverage result, however, is
evaluated only along the sequence of scenarios that are actually
realized.

\begin{thm}
\label{thm:empirical_coverage}
Let $\alpha\in(0,1)$ and $\gamma>0$. Under
Assumption~\ref{assump:g_t_basicAssump}, Algorithm~\ref{alg:ACSA},
initialized with $\tilde{\alpha}_1=\alpha$ and using the prediction sets
in \eqref{eq:acsa_prediction_set}, satisfies
\[
\left|
\frac{1}{T}\sum_{t=1}^T
\bbone\{G_{t+1}\in\Ct_t(\bz_{t+1})\}
-(1-\alpha)
\right|
\leq
\frac{\max\{\alpha,1-\alpha\}+\gamma}{T\gamma}
\]
almost surely. In particular, as $T\to\infty$, the empirical coverage
rate converges to $1-\alpha$ almost surely.
\end{thm}

The proof follows from the update rule in
\eqref{eq:OSA_update}. Summing the updates over time expresses the
empirical miscoverage rate in terms of the difference between
$\tilde{\alpha}_{T+1}$ and its initial value. The boundary convention in
\eqref{eq:acsa_prediction_set} keeps $\tilde{\alpha}_t$ bounded, so this
remainder vanishes after division by $T$. We give the full argument in
E-Companion~\ref{subapx:proofThmEmpiricalCoverage}.

The step size $\gamma$ controls how quickly ACSA responds to recent
coverage errors. A larger value adapts more rapidly but may cause
$\tilde{\alpha}_t$ and the corresponding interval widths to fluctuate
more strongly. A smaller value produces smoother adaptation but reacts
more slowly to changes in calibration. The coverage conclusion above
holds for every fixed $\gamma>0$, so in practice $\gamma$ can be chosen
using secondary criteria such as interval width or validation
performance. \citet{zaffran2022adaptive}, for example, consider adaptive
selection based on historical performance, while
\citet{gibbs2024conformal} aggregate over several candidate step sizes.

The guarantee in Theorem~\ref{thm:empirical_coverage} is marginal over
the sequence of realized scenarios. It does not imply correct coverage
within every part of the scenario space. In particular, a procedure
could attain the target average coverage while performing poorly in rare
or extreme scenarios, which are often precisely the scenarios of
greatest interest in risk management. One possible refinement is to
partition the scenario space into a small number of groups and maintain
a separate ACSA update within each group. We discuss this
group-balanced variant in E-Companion~\ref{subapx:alg_for_gbacsa}. Section~\ref{sec:KSA} introduces KSA, which instead targets
scenario-conditional coverage under stronger assumptions. We therefore
view ACSA and KSA as complementary: ACSA provides a simple and robust
online calibration mechanism, while KSA uses additional structure to
obtain more scenario-specific uncertainty estimates.

\subsection{The 3-Factor Log-Linear Model Revisited}\label{sec:ACSAMotivating_rewrite}

% \zz{For each time step $t$, the underlying index level $S_t$ may change, and we reinitialize the option configuration at each $t$ with a matching strike price $K_t$ such that the moneyness $K_t/S_t$ remains fixed at $1.05$, and the time to maturity is kept constant at $\tau = 21/252$.}

In order to illustrate the advantages of ACSA, we now revisit the three-factor log-linear model introduced in Section \ref{sec:MotivatingEG} where our portfolio consists of a single short position in a European call option on the index with a time to maturity of one month and a fixed moneyness $m=1.05$ so that the call option is 5\% out-of-the-money.

\paragraph{Parameter Misspecification} We assume the true data-generating process is as specified in Section~\ref{sec:MotivatingEG}. We assume all model parameters are known by the risk manager except for $\sigma_\epsilon$, the standard deviation of the random noise component. Recall the true value of $\sigma_\epsilon$ is $0.18$ but here the risk manager assumes (incorrectly) its value is $\widehat{\sigma}_\epsilon \not = \sigma_\epsilon$. Consequently, several quantities that depend on $\sigma_\epsilon$ are also incorrectly evaluated for each scenario $\bz$, including $\sigma_{\text{tot}}$, $m_{\bz}$ and $v_{\text{\tiny Cred}}$. Their estimates are obtained by replacing $\sigma_\epsilon$ with $\widehat{\sigma}_\epsilon$ in expressions \eqref{eq:expression_of_sig_total} and \eqref{eq:define_of_ms}. The target coverage probability of the prediction intervals is set to $1-\alpha = 50\%$. For a given risk scenario $\bz$, the benchmark quantile predictor $g_t$ uses $\widehat{\sigma}_\epsilon$ to construct a prediction interval based on the estimated $\alpha/2$ and $1-\alpha/2$ quantiles of the portfolio gain. In contrast, the ACSA approach instead uses the $\tilde{\alpha}_t/2$ and $1-\tilde{\alpha}_t/2$ quantiles computed by Algorithm \ref{alg:ACSA} using as input the benchmark $g_t$. We denote the two corresponding intervals by $\mathcal{C}_{\text{Mis}}(\bz, \widehat{\sigma}_\epsilon)$ (constructed using the benchmark $g_t$) and $\mathcal{C}_{\text{ACSA}}(\bz, \widehat{\sigma}_\epsilon, \tilde{\alpha}_t)$ (constructed using the ACSA algorithm). For reference, we let $\mathcal{C}_{\text{True}}(\bz)$ denote the oracle prediction interval based on the true $\alpha/2$ and $1-\alpha/2$ quantiles. The explicit forms of $\mathcal{C}_{\text{Mis}}(\bz, \widehat{\sigma}_\epsilon)$, $\mathcal{C}_{\text{ACSA}}(\bz, \widehat{\sigma}_\epsilon,\tilde{\alpha}_t)$, and $\mathcal{C}_{\text{True}}(\bz)$ are provided in E-Companion~\ref{subapx:addi_details_for_ACSA_motivating}.

\paragraph{Scenario Analysis Table}

In Table \ref{tab:ConvexIntervals_wACSA_wKSA_Sec5_2}, we present\footnote{Table \ref{tab:ConvexIntervals_wACSA_wKSA_Sec5_2} also presents the prediction interval $\mathcal{C}_{\text{KSA}}(\bz)$ produced by the KSA algorithm. We defer a discussion of this until Section \ref{sec:KSA}.} the three aforementioned prediction intervals. Due to space considerations we omit the most extreme scenarios where the rate factor is shocked by $\pm 30$ bp's and the oil factor is shocked by $\pm 20 \%$. In each cell, the first four rows respectively display: (1) the oracle interval $\mathcal{C}_\text{True}(\bz)$; (2) the misspecified interval $\mathcal{C}_\text{Mis}(\bz, \widehat{\sigma}_\epsilon)$ based on the benchmark $g_t$ with $\widehat{\sigma}_\epsilon = 0.1$; (3) the ACSA interval $\mathcal{C}_\text{ACSA}(\bz, \widehat{\sigma}_\epsilon, \bar{\tilde{\alpha}})$ with\footnote{The value $\bar{\tilde{\alpha}} = .03$ was the mean value of $\tilde{\alpha}_t$ across $2000$ steps of ACSA on a simulated dataset and hence is a ``typical'' value of $\tilde{\alpha}_t$. See Figure \ref{fig:beta_t_series} below for an illustration of the dynamics of $\tilde{\alpha}_t$.}  $\bar{\tilde{\alpha}} = 0.03$ and using $g_t$ with $\widehat{\sigma}_\epsilon = 0.1$ as the quantile predictor input and (4) the interval $\mathcal{C}_\text{KSA}(\bz)$ generated by the KSA algorithm of Section \ref{sec:KSAMotivtaingEG}. The remaining two rows contain the corresponding scenario point estimates from Table \ref{tab:Convex} and the SSA/CSA entries in Table \ref{tab:SSA}. The true coverage rate of each interval is shown as a percentage beside it, except for the oracle intervals whose true coverage is always $50\%$. We use red, blue and green asterisks as superscripts on the Mis, ACSA and KSA intervals,  respectively, when those intervals fail to contain $\mathbb{E}[G_{t+1} \mid \bz ]$, the true conditional expected gain in that scenario. To be clear,
this is only a diagnostic for interval location relative to the conditional mean $\mathbb{E}[G_{t+1} \mid \bz ]$. In particular, exclusion of $\mathbb{E}[G_{t+1} \mid \bz ]$ is not a validity failure, since the intervals target realized $G_{t+1} \mid \bz$ and not $\mathbb{E}[G_{t+1} \mid \bz ]$.

The following observations are in order. The true coverage rate of $\mathcal{C}_{\text{Mis}}(\bz, \widehat{\sigma}_\epsilon)$ can fall significantly below the target $50\%$, and is even less than $10\%$ in some extreme scenarios. But ACSA mitigates this and its prediction intervals have coverage rates that are much closer to the target coverage rate of $50\%$. This is especially true for relatively ``common'' scenarios where $z_{\text{\tiny Oil}}$ and $z_{\text{\tiny Rate}}$ are in the ranges $\{-6\%, 0, 6\%\}$ and $\{$-10 bp, 0, 10 bp$\}$, respectively. 
This illustrates how ACSA can compensate for an overly narrow quantile predictor by reducing the adjusted miscoverage level and
widening the resulting intervals. Theorem~\ref{thm:empirical_coverage} concerns coverage averaged over the realized sequence under the online update; it does not imply
that the ACSA intervals in this table, constructed using the fixed value $\bar{\tilde{\alpha}}=0.03$, have marginal coverage exactly
equal to $50\%$. The table exhibits both undercoverage and overcoverage across scenarios.

Additionally, we see that the true coverage rates of both $\mathcal{C}_{\text{Mis}}(\bz,\allowbreak \widehat{\sigma}_\epsilon)$ and $\mathcal{C}_{\text{ACSA}}(\bz,\allowbreak \widehat{\sigma}_\epsilon,\allowbreak\bar{\tilde{\alpha}})$ decrease from left to right within each row, and increase from top to bottom within each column. This pattern mirrors the variation of the conditional expected log-return of the underlying asset (i.e., the $m_{\bz}$ term in (\ref{eq:define_of_ms})) across the different scenarios. In particular, (\ref{eq:logret_cond_dist}) and (\ref{eq:define_of_ms}) yield
$\log\left(S_{t+1}/S_t\right)\Bigl|\allowbreak\, \left(\mathbf F_s=\bz \right) \;\sim\allowbreak\;
\mathcal{N}\bigl(m_{\bz},\,v_{\text{\tiny Cred}}\bigr)$
where
\begin{equation}\label{eq:explicit_form_of_ms}
m_{\bz} \approx 1.19\times 10^{-4} -0.16 z_{\text{\tiny Oil}} + 4.33 z_{\text{\tiny Rate}} -\dfrac{\sigma_\epsilon^2}{504}
\end{equation}

and $v_{\text{\tiny Cred}}$ does not depend on the scenario $\bz$. The key observation is that $m_{\bz}$ decreases with $z_{\text{\tiny Oil}}$ and increases with $z_{\text{\tiny Rate}}$ and this observation ultimately explains the monotonic behavior of the coverage rates of $\mathcal{C}_{\text{Mis}}(\bz, \widehat{\sigma}_\epsilon)$ and $\mathcal{C}_{\text{ACSA}}(\bz, \widehat{\sigma}_\epsilon,\bar{\tilde{\alpha}})$ as we decrease from left to right within each row, and increase from top to bottom within each column of Table \ref{tab:ConvexIntervals_wACSA_wKSA_Sec5_2}. Further details are provided in E-Companion \ref{subapx:addi_details_for_ACSA_motivating}.

\begin{table}[!htbp]
\centering
\scriptsize
\renewcommand{\arraystretch}{1.0}
\setlength{\tabcolsep}{2pt}
\resizebox{\textwidth}{!}{%
\begin{tabularx}{1.1\textwidth}{@{}R{12mm} R{18mm} C C C C C @{}}
\toprule
Rate$\backslash$Oil &  & $-12\%$ & $-6\%$ & $0$ & $+6\%$ & $+12\%$ \\
\midrule
\multirow{6}{*}{\centering $-20$ bp}
& \makecell{$\mathcal{C}_{\text{True}}$ \\ $\mathcal{C}_{\text{Mis}}$ \\ $\mathcal{C}_{\text{ACSA}}$ \\ $\mathcal{C}_{\text{KSA}}$ \\ $\mathbb{E}[G_{t+1} \mid \bz ]$ \\ SSA/CSA}
& \makecell{
(-27.2, -0.2)\\
(7.3, 21.8)$^{\textcolor{red}{*}}$,12\%\\
(-16.1, 32.1),57\%\\
(-20.6, 1.8),44\%\\
-15.10\\
-35.0, -12.5
}
& \makecell{
(-11.1, 10.9)\\
(19.1, 29.6)$^{\textcolor{red}{*}}$,8\%\\
(1.3, 36.7)$^{\textcolor{blue}{*}}$,49\%\\
(-10.1, 10.7),48\%\\
-1.35\\
-14.3, 1.0
}
& \makecell{
(2.1, 19.7)\\
(27.7, 35.1)$^{\textcolor{red}{*}}$,5\%\\
(14.6, 39.8)$^{\textcolor{blue}{*}}$,41\%\\
(-3.2, 17.6),53\%\\
9.80\\
2.1, 11.9
}
& \makecell{
(12.8, 26.5)\\
(33.8, 38.7)$^{\textcolor{red}{*}}$,3\%\\
(24.5, 41.6)$^{\textcolor{blue}{*}}$,33\%\\
(1.3, 22.6),52\%\\
18.66\\
14.7, 20.4
}
& \makecell{
(21.1, 31.6)\\
(37.8, 41.0)$^{\textcolor{red}{*}}$,2\%\\
(31.6, 42.8)$^{\textcolor{blue}{*}}$,25\%\\
(3.2, 26.8),47\%\\
25.58\\
24.1, 27.0
} \\
\midrule

\multirow{6}{*}{\centering $-10$ bp}
& \makecell{$\mathcal{C}_{\text{True}}$ \\ $\mathcal{C}_{\text{Mis}}$ \\ $\mathcal{C}_{\text{ACSA}}$ \\ $\mathcal{C}_{\text{KSA}}$ \\ $\mathbb{E}[G_{t+1} \mid \bz ]$ \\ SSA/CSA}
& \makecell{
(-35.2, -5.8)\\
(0.9, 17.4)$^{\textcolor{red}{*}}$,13\%\\
(-25.1, 29.5),60\%\\
(-26.3, -3.2),42\%\\
-22.02\\
-34.8, -19.3
}
& \makecell{
(-17.8, 6.4)\\
(14.4, 26.5)$^{\textcolor{red}{*}}$,10\%\\
(-5.8, 34.9)$^{\textcolor{blue}{*}}$,53\%\\
(-14.4, 7.3),47\%\\
-7.04\\
-14.0, -4.6
}
& \makecell{
(-3.3, 16.1)\\
(24.3, 32.9)$^{\textcolor{red}{*}}$,7\%\\
(9.3, 38.6)$^{\textcolor{blue}{*}}$,45\%\\
(-6.1, 15.1),52\%\\
5.21\\
2.3, 7.4
}
& \makecell{
(8.4, 23.7)\\
(31.4, 37.3)$^{\textcolor{red}{*}}$,4\%\\
(20.6, 40.9)$^{\textcolor{blue}{*}}$,37\%\\
(-0.3, 20.5),52\%\\
15.04\\
14.8, 16.9
}
& \makecell{
(17.7, 29.6)\\
(36.3, 40.1)$^{\textcolor{red}{*}}$,3\%\\
(28.8, 42.4)$^{\textcolor{blue}{*}}$,29\%\\
(3.2, 24.5),47\%\\
22.77\\
24.2, 24.4
} \\
\midrule

\multirow{6}{*}{\centering $0$}
& \makecell{$\mathcal{C}_{\text{True}}$ \\ $\mathcal{C}_{\text{Mis}}$ \\ $\mathcal{C}_{\text{ACSA}}$ \\ $\mathcal{C}_{\text{KSA}}$ \\ $\mathbb{E}[G_{t+1} \mid \bz ]$ \\ SSA/CSA}
& \makecell{
(-43.9, -12.0)\\
(-6.2, 12.4)$^{\textcolor{red}{*}}$,15\%\\
(-35.1, 26.3),63\%\\
(-34.0, -7.5),44\%\\
-29.52\\
-34.5, -26.7
}
& \makecell{
(-25.0, 1.4)\\
(8.9, 22.9)$^{\textcolor{red}{*}}$,11\%\\
(-13.7, 32.8),56\%\\
(-19.5, 3.4),46\%\\
-13.25\\
-13.8, -10.6
}
& \makecell{
(-9.3, 12.1)\\
(20.3, 30.4)$^{\textcolor{red}{*}}$,8\%\\
(3.2, 37.2)$^{\textcolor{blue}{*}}$,48\%\\
(-9.1, 11.7),49\%\\
0.16\\
2.5, 2.5
}
& \makecell{
(3.6, 20.6)\\
(28.6, 35.6)$^{\textcolor{red}{*}}$,5\%\\
(16.0, 40.0)$^{\textcolor{blue}{*}}$,40\%\\
(-2.5, 18.3),52\%\\
11.01\\
15.0, 13.0
}
& \makecell{
(13.9, 27.2)\\
(34.4, 39.0)$^{\textcolor{red}{*}}$,3\%\\
(25.5, 41.8)$^{\textcolor{blue}{*}}$,32\%\\
(2.0, 22.7),50\%\\
19.62\\
24.3, 21.3
} \\
\midrule

\multirow{6}{*}{\centering $10$ bp}
& \makecell{$\mathcal{C}_{\text{True}}$ \\ $\mathcal{C}_{\text{Mis}}$ \\ $\mathcal{C}_{\text{ACSA}}$ \\ $\mathcal{C}_{\text{KSA}}$ \\ $\mathbb{E}[G_{t+1} \mid \bz ]$\\ SSA/CSA}
& \makecell{
(-53.3, -18.8)\\
(-14.2, 6.7)$^{\textcolor{red}{*}}$,17\%\\
(-45.9, 22.7),66\%\\
(-40.8, -14.3),41\%\\
-37.62\\
-34.2, -34.7
}
& \makecell{
(-32.9, -4.2)\\
(2.8, 18.7)$^{\textcolor{red}{*}}$,13\%\\
(-22.5, 30.2),59\%\\
(-25.3, -1.0),45\%\\
-20.01\\
-13.6, -17.3
}
& \makecell{
(-15.8, 7.7)\\
(15.8, 27.4)$^{\textcolor{red}{*}}$,9\%\\
(-3.7, 35.5)$^{\textcolor{blue}{*}}$,52\%\\
(-12.8, 7.9),46\%\\
-5.38\\
2.7, -2.9
}
& \makecell{
(-1.7, 17.2)\\
(25.3, 33.6)$^{\textcolor{red}{*}}$,6\%\\
(10.9, 38.9)$^{\textcolor{blue}{*}}$,44\%\\
(-5.1, 15.7),52\%\\
6.55\\
15.1, 8.7
}
& \makecell{
(9.7, 24.5)\\
(32.1, 37.7)$^{\textcolor{red}{*}}$,4\%\\
(21.8, 41.1)$^{\textcolor{blue}{*}}$,36\%\\
(0.3, 21.0),52\%\\
16.10\\
24.4, 17.9
} \\
\midrule

\multirow{6}{*}{\centering $20$ bp}
& \makecell{$\mathcal{C}_{\text{True}}$ \\ $\mathcal{C}_{\text{Mis}}$ \\ $\mathcal{C}_{\text{ACSA}}$ \\ $\mathcal{C}_{\text{KSA}}$ \\ $\mathbb{E}[G_{t+1} \mid \bz ]$\\ SSA/CSA}
& \makecell{
(-63.2, -26.2)\\
(-23.0, 0.3)$^{\textcolor{red}{*}}$,19\%\\
(-57.7, 18.4),69\%\\
(-49.1, -21.2),40\%\\
-46.34\\
-33.9, -43.3
}
& \makecell{
(-41.4, -10.2)\\
(-4.1, 13.9)$^{\textcolor{red}{*}}$,14\%\\
(-32.2, 27.3),62\%\\
(-31.6, -6.3),43\%\\
-27.34\\
-13.4, -24.5
}
& \makecell{
(-22.9, 2.8)\\
(10.5, 24.0)$^{\textcolor{red}{*}}$,11\%\\
(-11.4, 33.4)$^{\textcolor{blue}{*}}$,55\%\\
(-17.1, 3.7),42\%\\
-11.44\\
2.8, -8.9
}
& \makecell{
(-7.5, 13.3)\\
(21.5, 31.2)$^{\textcolor{red}{*}}$,7\%\\
(5.0, 37.6)$^{\textcolor{blue}{*}}$,47\%\\
(-7.9, 12.4),48\%\\
1.64\\
15.2, 3.9
}
& \makecell{
(5.0, 21.5)\\
(29.4, 36.1)$^{\textcolor{red}{*}}$,5\%\\
(17.4, 40.3)$^{\textcolor{blue}{*}}$,39\%\\
(-1.8, 18.9),52\%\\
12.19\\
24.5, 14.2
} \\
\bottomrule
\end{tabularx}}
\caption{Portfolio gains across the $(\text{rate}, \text{oil})$ grid.
The first row of each cell reports the true 50\% prediction interval
$\mathcal{C}_{\text{True}}$. Rows two and three report the estimated
50\% prediction intervals $\mathcal{C}_{\text{Mis}}$ and
$\mathcal{C}_{\text{ACSA}}$, respectively. The fourth row reports the KSA interval,
$\mathcal{C}_{\text{KSA}}$. The fifth row reports the true
conditional expected scenario gain $\mathbb{E}[G_{t+1}\mid\bz]$, and
the sixth row reports the SSA and CSA point estimates.}
\label{tab:ConvexIntervals_wACSA_wKSA_Sec5_2}
\end{table}

\paragraph{The Dynamics of $\tilde{\alpha}_t$}
The evolution of ACSA’s $\tilde{\alpha}_t$ parameter is illustrated in Figure~\ref{fig:beta_t_series}. We test several values of the misspecified noise level $\widehat{\sigma}_\epsilon$, ranging from $0.05$ (an underestimate) to $0.25$ (an overestimate). Across all cases, $\tilde{\alpha}_t$ starts at $\tilde{\alpha}_1=0.5$ and then fluctuates around a stable level $\bar{\tilde{\alpha}}$ that depends on $\widehat{\sigma}_\epsilon$. This reflects ACSA’s process of learning the appropriate adjusted miscoverage level such that the realized prediction intervals achieve an empirical coverage rate of approximately $50\%$.

An interesting observation is that $\bar{\tilde{\alpha}}$ does not vary monotonically with $\widehat{\sigma}_\epsilon$: when $\widehat{\sigma}_\epsilon = 0.05$, the stabilized $\bar{\tilde{\alpha}}$ is $0.006$, suggesting that using the nominal quantile level $\alpha$ would yield an empirical coverage below $50\%$; as $\widehat{\sigma}_\epsilon$ approaches the true value $\sigma_\epsilon = 0.18$ from below, $\bar{\tilde{\alpha}}$ moves toward $0.5$ as expected. When $\widehat{\sigma}_\epsilon$ increases beyond  $\sigma_\epsilon$ to $0.25$, $\bar{\tilde{\alpha}}$ drops again, to around $0.18$. Intuitively, this \textit{non-monotonic} pattern arises from an underlying \textit{monotonic} dynamic behavior of the benchmark interval $\mathcal{C}_\text{Mis}(\bz, \widehat{\sigma}_\epsilon)$. As $\widehat{\sigma}_\epsilon$ increases, the estimated price (and its quantiles) for the call option increase, so the estimated portfolio gain decreases because the portfolio is short a call. Since the two endpoints of $\mathcal{C}_\text{Mis}(\bz,\widehat{\sigma}_\epsilon)$ are estimated quantiles of the portfolio gain, both endpoints move monotonically in the opposite direction of $\widehat{\sigma}_\epsilon$. The resulting true coverage rate is nevertheless non-monotonic in
$\widehat{\sigma}_\epsilon$. As shown using the explicit coverage function
derived in E-Companion \ref{subapx:addi_details_for_ACSA_motivating} and illustrated in Figure \ref{fig:non_mono_in_sigma}, the coverage first
increases and then decreases as $\widehat{\sigma}_\epsilon$ varies over the
range considered in our experiment. This non-monotonicity in turn explains the
non-monotonic dependence of the stabilized value $\bar{\widetilde{\alpha}}$
on $\widehat{\sigma}_\epsilon$.

\begin{figure}
    \centering
    \includegraphics[width=\linewidth]{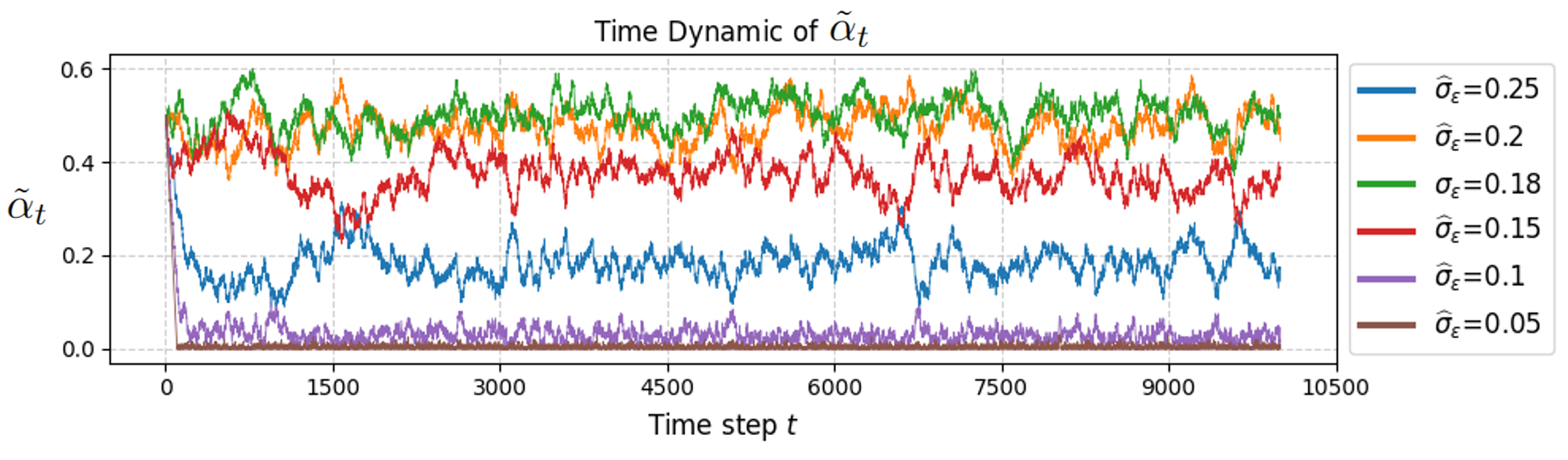}
    \caption{Dynamics of ACSA's adjusted miscoverage level $\tilde{\alpha}_t$ under different $\widehat{\sigma}_\epsilon$. For each estimation, the $\tilde{\alpha}_t$ value fluctuates around distinct values. Specifically, for $\widehat{\sigma}_\epsilon = \{0.25, 0.2, 0.18, 0.15, 0.1, 0.05\}$, the stabilized values of $\tilde{\alpha}_t$ are approximately $\{0.18, 0.47, 0.5, 0.37, 0.03, 0.006\}$, respectively. These stabilized values summarize the dynamic ACSA trajectories targeting $50\%$ empirical coverage.
   }
    \label{fig:beta_t_series}
\end{figure}

\section{Kernel Scenario Analysis}
\label{sec:KSA}

ACSA provides an online calibration mechanism without prescribing how scenario-wise quantiles should be estimated. We now introduce kernel scenario analysis (KSA), which constructs these quantiles by combining a predictor of the scenario gain with a kernel estimate of the remaining uncertainty. Historical observations receive greater weight when their stressed-factor moves and market conditions resemble those of the scenario being evaluated. In this way, KSA uses factor-based features to distinguish between scenarios with different gain distributions, without specifying a parametric model for the joint factor dynamics.

We first introduce the loss function, kernel, and feature vectors used by KSA. Section~\ref{sec:KSA_Alg} presents the core algorithm with a fixed centering function and then discusses adaptive centering and changing portfolios as practical extensions. Section~\ref{sec:ksa-theory} establishes the conditional-coverage guarantee for the core procedure under its stated assumptions, and Section~\ref{sec:KSAMotivtaingEG} revisits the three-factor running example.

\paragraph{Pinball Loss Function}
The pinball loss provides a convenient way to estimate quantiles through optimization. KSA uses a kernel-weighted version of this loss; it can also be used to train the quantile predictor $g_t$ in ACSA.

\begin{dfn}\label{def:Pinball}
The pinball loss function for $y,q\in\mathbb{R}$ and a quantile level $u\in[0,1]$ is
\begin{equation}\label{eq:pinball_loss_def}
\ell_u(y,q)\defeq u(y-q)^+ + (1-u)(y-q)^-,
\end{equation}
where $y^+\defeq\max\{y,0\}$ and $y^-\defeq-\min\{y,0\}$.
\end{dfn}

The pinball loss is piecewise linear and convex in $q$, and is a proper scoring rule for quantiles \citep{Gneiting_Raftery}. For a real-valued random variable $Y$ with CDF $F_Y$, define its $u$-quantile by $Q_Y(u)\defeq\inf\{y:F_Y(y)\geq u\}$ for $u\in(0,1)$. If $Y$ is integrable, then
\begin{equation}\label{eq:PBL}
Q_Y(u)\in\argmin_{q\in\mathbb{R}}\mathbb{E}[\ell_u(Y,q)].
\end{equation}
Replacing the expectation in \eqref{eq:PBL} by a sample average gives an empirical quantile estimator. KSA extends this construction by weighting observations according to their relevance to the target scenario.
We measure similarity between feature vectors using the following class of kernels.
\begin{dfn}\label{def:kernel}
A function $\kappa:\mathbb{R}^d\times\mathbb{R}^d\to\mathbb{R}$ is a \textbf{similarity radial kernel} if there exists a nonnegative, nonincreasing function $\widetilde{\kappa}:[0,\infty)\to[0,\infty)$ such that $\kappa(\bw,\bw')=\widetilde{\kappa}(\|\bw-\bw'\|)$ for all $\bw,\bw'\in\mathbb{R}^d$.
\end{dfn}

Thus, observations closer to the target feature vector receive at least as much weight as observations farther away. The kernel determines how broadly KSA pools historical observations around a given target; its bandwidth is made explicit in Section~\ref{sec:ksa-theory}.

\paragraph{Historical Context Vector}
The historical context vector $\bH_t$ collects the information available at time $t$ that we use to predict the next-period portfolio gain. For example, we may take $\bH_t=(\Delta\bff_t,\Delta\bff_{t-1},\ldots,\Delta\bff_{t-L_f+1})$ when $L_f$ lags of common-factor returns capture the relevant history. Risk-factor changes, current factor levels, or other market-state information can be included when they provide additional predictive information. Conversely, redundant components can be omitted: if the relevant risk-factor changes are known functions of the included common-factor returns, they need not be included separately.

A Markov structure can further reduce the required history. When a current state contains all historical information relevant to the next-period gain distribution, $\bH_t$ can be taken as that state rather than a longer vector of lagged observations. The important requirement is predictive sufficiency for the portfolio gain, not merely the Markov property of one component of the market model. The precise condition used by our analysis is stated in Assumption~\ref{assump:individual_cali_assump}.

\paragraph{Feature Vector}
We combine the historical context with the stress scenario to form the \textbf{feature vector} $\bW_t(\bz)\defeq(\bz,\bH_{t-1})$. At forecast date $t$, the target feature vector is therefore $\bW_{t+1}(\bz)=(\bz,\bH_t)$, where $\bz$ is the hypothesized stressed-factor move on the next day. For a historical date $s\leq t$, the scenario $\bz_s$ has already been observed, and we write $\bW_s\defeq\bW_s(\bz_s)=(\bz_s,\bH_{s-1})$. These definitions put historical observations and hypothetical scenarios in the same feature space. If $\bH_t$ retains all relevant historical information, conditioning on $\cF_t$ and $\bz_{t+1}=\bz$ gives the same gain distribution as conditioning on $\bH_t$ and $\bz_{t+1}=\bz$. KSA then estimates this distribution locally by comparing $\bW_{t+1}(\bz)$ with the observed feature vectors $\bW_s$.

\subsection{The Kernel Scenario Analysis Algorithm}
\label{sec:KSA_Alg}

The key idea of the Kernel Scenario Analysis (KSA) algorithm is to assign an \textit{importance weight} to each historical date $s\leq t$, based on the similarity between the historical feature vector $\bW_s$ and the scenario feature vector $\bW_{t+1}(\bz)$. Specifically, we use a similarity radial kernel $\kappa(\bw,\bw')$ to encode similarities between feature vectors $\bw$ and $\bw'$. These importance weights are then used to form a reweighted empirical pinball loss, whose minimizer yields an estimate of the conditional quantile of the portfolio gain under a given scenario $\bz$. The prediction interval is then constructed using the estimated quantiles as its endpoints. The core KSA procedure is summarized in Algorithm~\ref{alg:KSA}.

\begin{algorithm}[ht!]
    \caption{Kernel Scenario Analysis}
    \label{alg:KSA}
    \begin{algorithmic}[1]
    \Require Target coverage rate $1-\alpha$, kernel $\kappa(\cdot,\cdot)$, predictor of expected portfolio gain $\varphi(\cdot)$
    \State Set $\bW_1\gets(\bz_1,\Ht{0})$ and
$Y_1\gets G_1-\varphi(\bW_1)$ \label{line:NSA_initial}
    \For{$t=2,\ldots,T$}
        \State Observe portfolio gain $G_t$ and realized stress scenario $\bz_t$
        \State Set $\bW_t\gets(\bz_t,\Ht{t-1})$ and $Y_t\gets G_t-\varphi(\bW_t)$\label{step:construct_Y}
        \State For any $u\in[0,1]$, $q\in\mathbb{R}$ and $\bw$, define the reweighted residual pinball loss $\widehat{\ell}_t$
        \begin{equation}\label{eq:reweightedquantilegain}
            \widehat{\ell}_t(u;q,\bw)
            \defeq
            \dfrac{\sum_{s=1}^{t}\kappa(\bW_s,\bw)\cdot\ell_u(Y_s,q)}
            {\sum_{s=1}^{t}\kappa(\bW_s,\bw)}
        \end{equation}
        \label{line:reweight_gain}
        \State \textbf{Output:} For each scenario $\bz$, set $\bW_{t+1}(\bz)\gets(\bz,\Ht{t})$ and
        \begin{equation}\label{eq:individual_interval}
            \Ct_t(\bz)
            \gets
            \left[
            \varphi(\bW_{t+1}(\bz))
            +\widehat{Q}_t(\alpha/2;\bW_{t+1}(\bz)),
            \,
            \varphi(\bW_{t+1}(\bz))
            +\widehat{Q}_t(1-\alpha/2;\bW_{t+1}(\bz))
            \right]
        \end{equation}
        where
        \begin{equation}\label{eq:quantile_of_U}
            \widehat{Q}_t(u;\bW_{t+1}(\bz))
            \defeq
            \argmin_{q\in\mathbb{R}}
            \widehat{\ell}_t(u;q,\bW_{t+1}(\bz)),
            \qquad
            u\in\{\alpha/2,1-\alpha/2\}.
        \end{equation}
        \label{line:produce_interval}
    \EndFor
    \end{algorithmic}
\end{algorithm}

Algorithm~\ref{alg:KSA}  uses a predictor $\varphi(\bw)$ to represent the predictable component of the portfolio gain given the stress scenario and historical context vector. Its role is to separate the portfolio gain into a predictable part and a residual. In Step~\ref{step:construct_Y}, the residual term $Y_s$ is calculated for $s\leq t$ by subtracting $\varphi(\bW_s)$ from $G_s$. The $Y_s$'s are then used in place of the $G_s$'s as the target variable whose conditional quantiles are estimated.

The reason that we work with the $Y_s$'s instead of the $G_s$'s is that if $\varphi(\bW_s)$ predicts $\E[G_s\mid\bW_s]$ well, then $Y_s$ should contain less predictable variation and it should therefore be easier to estimate its conditional quantiles than those of the original $G_s$, which may be strongly covariate-dependent. Importantly, however, $\varphi$ need not equal the true conditional expectation for the construction to be valid. For any fixed deterministic $\varphi$, the conditional quantile of $G_{t+1}\mid\bW_{t+1}(\bz)$ can be recovered from the conditional quantile of $Y_{t+1}\mid\bW_{t+1}(\bz)$ via
\begin{equation}\label{eq:Y_and_U}
    Q_{G_{t+1}}(u\mid\bW_{t+1}(\bz))
    =
    \varphi(\bW_{t+1}(\bz))
    +
    Q_{Y_{t+1}}(u\mid\bW_{t+1}(\bz)).
\end{equation}
Thus, a more accurate choice of $\varphi$ is useful primarily because it can make the residual quantile surface easier to estimate, rather than because the validity of the quantile-shift identity requires $\varphi$ to be the true conditional mean. We demonstrate this ``smoother landscape'' phenomenon in the following example.

\begin{example}[Centering Removes Predictable Variation]\label{example:cite_an_example}\sffamily
{\sffamily
Suppose $G=W+\epsilon$, where $\epsilon\sim\mathcal{N}(0,1)$ denotes a noise term independent of the random variable $W\in\mathbb{R}$. The conditional $u$-quantile of $G\mid W$ is $Q_G(u\mid W)$ and takes the form
$Q_G(u\mid W=w)=w+\Phi^{-1}(u)$,
where $\Phi$ is the standard normal CDF. If we have a perfect expectation predictor $\varphi$ of the conditional expectation, i.e.,
$\varphi(w)=\E[G\mid W=w]=w$,
then the constructed residual $Y\defeq G-\varphi(w)$ satisfies $Y=\epsilon$. The conditional $u$-quantile of $Y\mid W$ is then
$Q_Y(u\mid W=w)=\Phi^{-1}(u)$,
which is independent of $w$.

The implication of a ``smoother landscape'' is that the distributions, and consequently the quantiles, of $Y\mid W=w$ may vary less with $w$. In this particular example, $\Phi^{-1}(u)$ and therefore the conditional quantile of $Y\mid W=w$ does not depend on $w$ at all. This property makes estimating $Q_Y(u\mid W)$ easier than estimating $Q_G(u\mid W)$.}
\end{example}\medskip

Returning to Algorithm~\ref{alg:KSA}, in Steps~\ref{line:reweight_gain} and~\ref{line:produce_interval}, the conditional quantile of $Y_{t+1}\mid\bW_{t+1}(\bz)$ is estimated by minimizing a reweighted version of the residual pinball loss function. To see what this minimizer gives us, observe that $\widehat{F}_t(y\mid\bW_{t+1}(\bz))
    \defeq
    \frac{
    \sum_{s=1}^{t}
    \kappa(\bW_s,\bW_{t+1}(\bz))
    \cdot
    \bbone\{y\geq Y_s\}
    }{
    \sum_{s=1}^{t}
    \kappa(\bW_s,\bW_{t+1}(\bz))
    }$ is an estimator of the CDF of $Y_{t+1}\mid\bW_{t+1}(\bz)$. Use $\widehat{\mathcal{P}}_t$ to denote the probability distribution induced by this CDF. The $u$-quantile of $\widehat{\mathcal{P}}_t$ is given by
\begin{equation}
    \begin{aligned}
    Q_{\widehat{\mathcal{P}}_t}(u)
    &=
    \argmin_{q\in\mathbb{R}}
    \E_{Y\sim\widehat{\mathcal{P}}_t}
    [\ell_u(Y,q)] \\
    &=
    \argmin_{q\in\mathbb{R}}
    \dfrac{
    \sum_{s=1}^{t}
    \kappa(\bW_s,\bW_{t+1}(\bz))
    \cdot
    \ell_u(Y_s,q)
    }{
    \sum_{s=1}^{t}
    \kappa(\bW_s,\bW_{t+1}(\bz))
    }.
    \end{aligned}
\end{equation}
The estimator $\widehat{Q}_t(u;\bW_{t+1}(\bz))$ in \eqref{eq:quantile_of_U} is therefore exactly the $u$-quantile of $\widehat{\mathcal{P}}_t$. The minimization of $\widehat{\ell}_t(u;q,\bW_{t+1}(\bz))$ in Step~\ref{line:produce_interval} can be solved in $O(t)$ iterations. First, permute the indices with a permutation operator $\Pi:\{1,\ldots,t\}\to\{1,\ldots,t\}$ so that
$Y_{\Pi(1)}\leq Y_{\Pi(2)}\leq\cdots\leq Y_{\Pi(t)}$.
Then,
\begin{equation*}
\widehat{Q}_t(u;\bW_{t+1}(\bz))
=
Y_{\Pi(j)},
\qquad
j
=
\min
\left\{
i\in\{1,\ldots,t\}
\;\middle|\;
\dfrac{
\sum_{s=1}^{i}
\kappa(\bW_{\Pi(s)},\bW_{t+1}(\bz))
}{
\sum_{s=1}^{t}
\kappa(\bW_{\Pi(s)},\bW_{t+1}(\bz))
}
\geq u
\right\}.
\end{equation*}

A brief comparison between the input requirements for ACSA and KSA is useful. For each time $t$, the ACSA algorithm requires a quantile predictor $g_t(\bz)$ of the conditional portfolio gain, whereas the core KSA algorithm requires a centering predictor $\varphi(\cdot)$. It is not immediately clear which requirement is more demanding. For example, the median of a distribution is less sensitive to outliers and may be easier to estimate than the mean, whereas extreme quantiles, such as 1\% quantiles, are generally harder to estimate because far fewer observations are available in the tails. If one has an excellent quantile predictor $g_t$ at hand, then a downside of Algorithm~\ref{alg:KSA} is that it does not directly use $g_t$, but instead constructs its own quantile predictor through \eqref{eq:reweightedquantilegain} and \eqref{eq:quantile_of_U}. Quantile predictors, however, are often not readily available, so this aspect of KSA can be viewed as either a bug or a feature depending on the application.

In practice, it can be useful to replace the fixed centering function $\varphi$ by a time-varying predictor $\varphi_t$ that is updated as new data arrive. This adaptive-centering version follows the same kernel reweighting principle, but introduces additional dependence because the residuals themselves now depend on an estimated predictor. To avoid using the same observations both to update $\varphi_t$ and to calibrate the residual distribution, one can maintain separate fitting and calibration samples; in particular, once an observation is used to update the predictor, it is excluded from subsequent residual calibration.

A practical advantage of KSA over ACSA is the convenience of estimating conditional expectations under frequently changing portfolios. Because conditional expectations aggregate linearly across positions, one can estimate security-level expectation predictors and then combine them using the current portfolio weights to obtain a portfolio-level expectation predictor. The scenario-wise quantiles are then obtained through the kernel reweighting step, without training a new portfolio-level quantile model. By contrast, ACSA trains a portfolio-level quantile predictor, and quantiles do not in general aggregate linearly across securities, so changes in portfolio composition typically require retraining that predictor.

\subsection{Analysis of Algorithm 2}
\label{sec:ksa-theory}

This section establishes a conditional-coverage guarantee for the core KSA procedure in Algorithm~\ref{alg:KSA}. The main distinction from Theorem~\ref{thm:empirical_coverage} is the conditioning information: rather than controlling coverage averaged over the realized sequence of scenarios, we study next-period coverage conditional on the information currently available and on a specified scenario. Throughout this section, the portfolio is fixed. We retain the definitions
$\bW_t=(\bz_t,\bH_{t-1})$ and $\bW_{t+1}(\bz)=(\bz,\bH_t)$, and write $d_{\bW}$ for the dimension of the feature vector. To allow for temporal dependence in the historical feature--residual pairs, we use the notion of $\beta$-mixing. More details are provided in \citet{bradley2005basic}, and a brief definition is given below.

\begin{dfn}[$\beta$-mixing]
\label{def:beta-mixing-revised}
For a stationary process $\{X_t\}_{t\geq 0}$, let $\sigma(\mathcal X)$ denote the $\sigma$-algebra generated by a collection of random variables $\mathcal X$. Define
\begin{equation*}
\beta_X(k)
=
\sup_{t\geq 0}
\mathbb E\left[
\sup_{B\in\sigma(\{X_s:s\ge t+k\})}
\left|
\mathbb P\bigl(B\mid\sigma(\{X_s:s\le t\})\bigr)-\mathbb P(B)
\right|
\right].
\end{equation*}
The process is $\beta$-mixing if $\beta_X(k)\to0$ as $k\to\infty$. It is geometrically $\beta$-mixing if
$\beta_X(k)\le\beta_0\rho^k$ for some $\beta_0<\infty$ and
$0<\rho<1$.
\end{dfn}

Intuitively, $\beta_X(k)$ measures the dependence between the past of the process and its future after a gap of length $k$. A small value of $\beta_X(k)$ means that sufficiently separated observations are nearly independent, while geometric $\beta$-mixing requires this dependence to decay exponentially fast. In particular, $\beta_X(k)=0$ for all $k\geq1$ corresponds to the independent case.

We now specify the regularity conditions under which KSA achieves a conditional-coverage guarantee. Let $\mathcal D\subseteq\mathbb R^{d_{\bW}}$ be a compact region containing the feature vectors $\bW_{t+1}(\bz)$ corresponding to the scenarios and historical contexts of interest. For example, a scenario table typically considers stress scenarios $\bz$ within a prescribed range together with historical contexts $\bH_t$ that are not excessively extreme. We restrict the theoretical guarantee to targets satisfying $\bW_{t+1}(\bz)\in\mathcal D$. For some fixed $h_0>0$, define the enlarged region
$\mathcal D^+
=
\{\bw:\operatorname{dist}(\bw,\mathcal D)\leq h_0\}$,
where
$\operatorname{dist}(\bw,\mathcal D)
:=
\inf_{\bw'\in\mathcal D}\|\bw-\bw'\|$
denotes the Euclidean distance\footnote{Throughout this work, an unsubscripted $\|\cdot\|$ denotes the Euclidean norm.}
from $\bw$ to $\mathcal D$. Since kernel estimation at a target in $\mathcal D$ uses observations in its neighborhood, the regularity conditions below are imposed on $\mathcal D^+$.

\begin{assum}
\label{ass:ksa-revised}
\label{assump:individual_cali_assump}
The following conditions hold.
\begin{enumerate}[label=(\alph*)]

\item\label{itm:separate_stationary_assump}
The process $\{(\bW_t,Y_t)\}_{t\geq0}$ is stationary and geometrically $\beta$-mixing, with
$\beta(k)\leq\beta_0\rho^k$
for some $\beta_0<\infty$ and $0<\rho<1$.

\item\label{itm:independence_n_homogeneity}
The feature vector $\bW_{t+1}(\bz)$ contains all information in
$\mathcal F_t$ that is relevant to the conditional distribution of
$Y_{t+1}$ under scenario $\bz$. Specifically,
\begin{equation}
\mathbb P\bigl(
Y_{t+1}\leq y
\mid
\mathcal F_t,\bz_{t+1}=\bz
\bigr)
=
F_{Y\mid\bW}\bigl(
y\mid\bW_{t+1}(\bz)
\bigr)
\label{eq:predictive-sufficiency}
\end{equation}
almost surely, simultaneously for every $y\in\mathbb R$ and every $\bz$
such that $\bW_{t+1}(\bz)\in\mathcal D$.

\item\label{itm:density_bound}
The marginal distribution of $\bW_t$ admits a density $\pi_{\bW}$ that is bounded above and bounded away from zero on $\mathcal D^+$.

\item\label{itm:cond_dist_cont}
The conditional distribution of $Y_t$ given $\bW_t=\bw$ is continuous for every $\bw\in\mathcal D^+$. Moreover, the conditional CDF varies uniformly continuously with $\bw$:
\begin{equation*}
\lim_{\epsilon\downarrow0}
\sup_{\substack{\bw,\bw'\in\mathcal D^+\\
                 \|\bw-\bw'\|\leq\epsilon}}
\sup_{y\in\mathbb R}
\left|
F_{Y\mid\bW}(y\mid\bw)
-
F_{Y\mid\bW}(y\mid\bw')
\right|
=0.
\end{equation*}

\item\label{itm:kernel_regularity}
The kernel radial function $\widetilde{\kappa}$ in
Definition~\ref{def:kernel} is bounded and has support either
$[0,R_\kappa]$ for some $0<R_\kappa<\infty$ or $[0,\infty)$.
It is Lipschitz continuous on the interior of its support. Moreover,
$0<\int_{\mathbb R^{d_{\bW}}}
\widetilde{\kappa}(\|\mathbf v\|)\td\mathbf v
<\infty$
and
$\int_{\mathbb R^{d_{\bW}}}
\|\mathbf v\|^2
\widetilde{\kappa}(\|\mathbf v\|)\td\mathbf v
<\infty$.

\end{enumerate}
\end{assum}

Assumption~\ref{assump:individual_cali_assump}\ref{itm:separate_stationary_assump} provides the stationarity and weak dependence needed to control the kernel-weighted empirical averages. Stationarity is a working assumption for the KSA analysis; unlike the pathwise guarantee for ACSA, the result below relies on distributional regularity over time. Assumption~\ref{assump:individual_cali_assump}\ref{itm:independence_n_homogeneity} is a predictive-sufficiency condition: once $\bW_{t+1}(\bz)$ is specified, earlier information in $\mathcal F_t$ provides no additional information about the conditional distribution of the next residual. Assumption~\ref{assump:individual_cali_assump}\ref{itm:density_bound} guarantees adequate local support for the feature distribution over the operational region. Assumption~\ref{assump:individual_cali_assump}\ref{itm:cond_dist_cont} ensures that nearby feature vectors have similar conditional residual distributions, thereby controlling the approximation error introduced by local smoothing. Finally, Assumption~\ref{assump:individual_cali_assump}\ref{itm:kernel_regularity} imposes standard regularity conditions on the kernel. The RBF ($\widetilde{\kappa}(r)=e^{-r^2/2}$), Laplace ($\widetilde{\kappa}(r)=e^{-r}$), and naive kernels ($\widetilde{\kappa}(r)=\bbone\{r\leq1\}$) all satisfy these conditions.

To make the dependence on the bandwidth explicit, for $h>0$, write
$\kappa_h(\bw,\bw')
=
\widetilde{\kappa}\left(
\frac{\|\bw-\bw'\|}{h}
\right)$.
Algorithm~\ref{alg:KSA} takes $\kappa_h$ as a fixed input. The following theorem shows that, for a sufficiently small fixed bandwidth and a sufficiently large historical sample, the prediction intervals produced by KSA achieve scenario-wise conditional coverage close to the target level $1-\alpha$.

\begin{thm}
\label{thm:individual_cali_thm}
Suppose Assumption~\ref{assump:individual_cali_assump} holds.
For every $\epsilon>0$, there exists $h_\epsilon>0$ such that, for every fixed $h\in(0,h_\epsilon]$,
\begin{equation}
\sup_{\bz:\,\bW_{t+1}(\bz)\in\mathcal D}
\left|
\mathbb P\left(
G_{t+1}\notin\Ct_t(\bz)
\mid \mathcal F_t,\bz_{t+1}=\bz
\right)
-\alpha
\right|
\leq
\epsilon
+
O_p\left(\sqrt{\frac{\log t}{t}}\right).
\label{eq:ksa-conditional-coverage}
\end{equation}
Here, the $O_p$ term is with respect to the randomness of the historical data in $\mathcal F_t$, and its implicit constant may depend on the fixed bandwidth $h$.
\end{thm}

We use the notation $O_p(a_t)$ to denote a stochastic term whose magnitude, after division by $a_t$, is bounded in probability. An explicit finite-sample version of Theorem~\ref{thm:individual_cali_thm} is given as Theorem~\ref{apxthm:explicit_statement_individual} in E-Companion~\ref{subapx:proof_of_individual_cali}, together with its proof.

Unlike Theorem~\ref{thm:empirical_coverage}, which controls coverage averaged over the realized sequence of scenarios, Theorem~\ref{thm:individual_cali_thm} controls next-period coverage conditional on the current information and the specified scenario. The two terms in \eqref{eq:ksa-conditional-coverage} have a natural interpretation. The tolerance $\epsilon$ represents the approximation error from pooling observations whose feature vectors are close, but not identical, to $\bW_{t+1}(\bz)$. This error can be made arbitrarily small by choosing a sufficiently small bandwidth. The stochastic term represents finite-sample estimation error and vanishes as $t$ increases. A smaller bandwidth therefore reduces local approximation error but also places substantial weight on fewer nearby observations, so more historical data may be needed before the asymptotic behavior becomes apparent.

The guarantee applies only when the target feature vector lies in $\mathcal D$. This restriction is natural for a local-smoothing procedure: the target must have adequate nearby historical support, and the conditional residual distribution must vary regularly over its neighborhood. These conditions may fail for scenarios far outside the historically supported region. In addition, $\mathcal D$ cannot be taken to be all of $\mathbb R^{d_{\bW}}$ under Assumption~\ref{assump:individual_cali_assump}\ref{itm:density_bound}, since a probability density cannot be bounded away from zero over an unbounded space. In applications, $\mathcal D$ should therefore be chosen to cover the scenarios and historical contexts within the intended domain of use.

Under Assumption~\ref{assump:individual_cali_assump} alone, continuity of the conditional CDF controls the kernel-smoothing bias only qualitatively. The following stronger smoothness conditions give explicit orders in the bandwidth.

\begin{cor}
\label{cor:ksa-smoothness}
Suppose Assumption~\ref{assump:individual_cali_assump} holds.

\begin{enumerate}[label=(\roman*)]

\item\label{itm:lipschitz}
If, for some $L<\infty$,
\begin{equation*}
\sup_{y\in\mathbb R}
\left|
F_{Y\mid\bW}(y\mid\bw)
-
F_{Y\mid\bW}(y\mid\bw')
\right|
\leq
L\|\bw-\bw'\|
\end{equation*}
for all $\bw,\bw'\in\mathcal D^+$, then for every sufficiently small $h$,
\begin{equation*}
\sup_{\bz:\,\bW_{t+1}(\bz)\in\mathcal D}
\left|
\mathbb P\left(
G_{t+1}\notin\Ct_t(\bz)
\mid \mathcal F_t,\bz_{t+1}=\bz
\right)
-\alpha
\right|
=
O(h)
+
O_p\left(\sqrt{\frac{\log t}{t}}\right).
\end{equation*}

\item\label{itm:twice_diff}
If $\pi_{\bW}$ is twice continuously differentiable with Hessian uniformly bounded over $\bw\in\mathcal D^+$, and for every $y\in\mathbb R$,
$\pi_{\bW}(\cdot)F_{Y\mid\bW}(y\mid\cdot)$
is twice continuously differentiable on an open neighborhood of $\mathcal D^+$, with Hessian uniformly bounded over $\bw\in\mathcal D^+$ and $y\in\mathbb R$, then for every sufficiently small $h$,
\begin{equation*}
\sup_{\bz:\,\bW_{t+1}(\bz)\in\mathcal D}
\left|
\mathbb P\left(
G_{t+1}\notin\Ct_t(\bz)
\mid \mathcal F_t,\bz_{t+1}=\bz
\right)
-\alpha
\right|
=
O(h^2)
+
O_p\left(\sqrt{\frac{\log t}{t}}\right).
\end{equation*}

\end{enumerate}

The deterministic orders are understood as $h\downarrow0$, with their constants uniform over all sufficiently small $h$. For each fixed $h$, the stochastic orders are understood as $t\to\infty$, and their implicit constants may depend on $h$.
\end{cor}

An explicit finite-sample version of Corollary~\ref{cor:ksa-smoothness} is given as Corollary~\ref{apxcor:explicit_ksa_smoothness} in E-Companion~\ref{subapx:proof_of_individual_cali}, together with its proof. Corollary~\ref{cor:ksa-smoothness} makes the kernel-approximation term in Theorem~\ref{thm:individual_cali_thm} explicit. Under the Lipschitz condition in part~\ref{itm:lipschitz}, the approximation error is of order $h$. Under the stronger second-order condition in part~\ref{itm:twice_diff}, radial symmetry of the kernel cancels the first-order contribution and improves this error to order $h^2$. In both cases, the remaining stochastic term reflects estimation from a finite historical sample.

Although Algorithm~\ref{alg:KSA} is formulated with a fixed bandwidth, the theoretical construction can also use a deterministic sequence of bandwidths $h_t$. A fixed bandwidth leaves a non-vanishing kernel-approximation error as $t\to\infty$, whereas a suitably shrinking bandwidth can drive both approximation and estimation errors to zero. The bandwidth-dependent bounds in Corollary~\ref{cor:ksa-smoothness} and its explicit restatement, Corollary~\ref{apxcor:explicit_ksa_smoothness}, provide guidance for this choice.

Under the Lipschitz condition in part~\ref{itm:lipschitz}, taking
$h_t\asymp(\log t/t)^{1/(2d_{\bW}+2)}$
gives a coverage-error rate of
$O_p((\log t/t)^{1/(2d_{\bW}+2)})$.
Under the stronger second-order condition in part~\ref{itm:twice_diff}, taking
$h_t\asymp(\log t/t)^{1/(2d_{\bW}+4)}$
gives the faster rate
$O_p((\log t/t)^{1/(d_{\bW}+2)})$.
The dependence of these rates on $d_{\bW}$ reflects the usual curse of dimensionality in nonparametric estimation: higher-dimensional feature vectors require more data to obtain the same degree of local accuracy. Corollary~\ref{apxcor:optimized_ksa_bandwidth} in E-Companion~\ref{subapx:proof_of_individual_cali} gives the explicit bandwidths that minimize the corresponding finite-sample bounds. These statements concern deterministic bandwidth choices; data-dependent bandwidth selection is a separate practical extension.

\subsection{Applying KSA to the 3-Factor Log-Linear Model}
\label{sec:KSAMotivtaingEG}

We now revisit our running example, in which the log security price follows a three-factor log-linear model. The portfolio is short a call option with fixed time to maturity $\tau$ and moneyness (defined as the strike price divided by the underlying price) $m$. In the simulated historical dataset, for each date $s<t$, we assume that the option has the same moneyness $m$ and time-to-maturity $\tau$ on day $s$, with the portfolio gain $G_{s+1}$ then computed on day $s+1$. As detailed in E-Companion~\ref{subapx:addi_details_KSA_motive}, KSA is applied to the normalized gain
$G_{s+1}'\defeq \frac{G_{s+1}}{S_s}.$ This normalization removes the changing index-price scale across historical dates. Since $S_s$ is positive and known at time $s$, a prediction interval for $G_{s+1}'$ can be converted into an equivalent prediction interval for $G_{s+1}$ by multiplying both endpoints by $S_s$. We evaluate the KSA algorithm using $2000$ time steps of historical data simulated under the model parameters specified in Section~\ref{sec:MotivatingEG}. We set the feature vector to $\bW_t(\bz)=\bz$, so no historical context $\bH_t$ is needed and the two stressed factor returns are sufficient as the feature vector for KSA. In E-Companion~\ref{subapx:addi_details_KSA_motive}, we show that this choice of $\bW_t(\bz)$ satisfies Assumption~\ref{assump:individual_cali_assump}. Additional implementation details for KSA, including the choice of the kernel $\kappa$ and the expectation predictor $\varphi$, are also deferred to E-Companion~\ref{subapx:addi_details_KSA_motive}.

The KSA prediction intervals for the various stress scenarios are reported in Table~\ref{tab:ConvexIntervals_wACSA_wKSA_Sec5_2}. Among all methods considered, the KSA intervals are closest to the oracle intervals. For relatively ``common'' scenarios with $z_{\text{\tiny Oil}}\in\{-6\%,0,6\%\}$ and $z_{\text{\tiny Rate}}\in\{-10\text{ bp},0,10\text{ bp}\}$, the conditional coverage rates all lie within $\pm 5\%$ of the target $50\%$ level. For the other, more extreme scenarios, KSA still generally attains better coverage than ACSA, with the conditional coverage remaining within $\pm 10\%$ of the target $50\%$ level. The numerical results show that KSA performs exceptionally well in this simple and idealized setting.

The strong performance in the running example reflects several idealizations. The
centering predictor is a fixed oracle conditional mean, evaluated numerically as
described in E-Companion~\ref{subapx:addi_details_KSA_motive}, so the experiment
abstracts from the estimation error that would arise when $\varphi$ must be learned
from data. We also assume that the relevant factor structure is known. Most
importantly, the feature vector is only two-dimensional; in a more realistic Markov
setting with a two-dimensional stress scenario and a $d$-dimensional state vector,
its dimension would typically increase to $2+d$, making nearby historical
observations increasingly sparse and local kernel estimation more difficult.

Thus, the running example is best viewed as an illustration of KSA's mechanics and
potential performance under favorable conditions. In practice, high-dimensional
conditioning is likely to be a key limitation, although good estimates of the mean
predictor and factor structure can mitigate this issue when sufficient historical
data are available. Section~\ref{sec:Numerics} considers a more realistic setting
in which $\varphi$ must be estimated and richer historical context vectors
$\bH_t$ are used.

\section{Further Numerical Experiments}
\label{sec:Numerics}
The 3-factor log-linear model introduced earlier served primarily as a simple running example for demonstrating the ACSA and KSA algorithms. In this section, we turn to a more complex and realistic setting, where we take the perspective of a portfolio manager and stress test the portfolio returns on a daily basis. We again adopt a factor model structure as described in Section \ref{sec:Factor Structure} in order to simulate the market environment. The model parameters are calibrated with real world data to ensure realism and we present a complete pipeline for implementing our proposed methods in practice. We also conduct a systematic analysis of performance across different algorithmic configurations, including the standard scenario analysis as a benchmark.

\subsection{Portfolio, Model and Experimental Setup}
\label{sec:Setup}

% In this subsection we introduce the factor model dynamics and the portfolio setup. In the following sections, we assume the model parameters are known, while details on parameter calibration are deferred to E-Companion \ref{app:mo_calibrate}.

\paragraph{Available Securities}
The portfolio consists of out-of-the-money European call and put options written on the S\&P 500 index, along with a possible position in the underlying index itself. This corresponds to the equity and options setting of Section~\ref{sec:Examples}, except that we now restrict attention to a single underlying security.
We assume there are $n$ tradable securities. The first is the S\&P 500 index, with time $t$ price $S_t$, and the remaining $n-1$ securities are European options of varying strikes and maturities. Each option is characterized by its fixed time-to-maturity $\{\tau_j\}_{j=1}^{n-1}$ and fixed moneyness (the strike price $K$ divided by $S_t$) $\{m_j\}_{j=1}^{n-1}$. The portfolio is rebalanced daily and we assume the option maturities and possible moneyness values remain constant over time. While these characteristics are held fixed, the implied volatility of each option  $\{\sigma_t^{(j)}\}_{j=1}^{n-1}$ evolves over time. For simplicity, we assume the index pays no dividends, and the risk-free continuously compounded interest rate is fixed at $r = 4\%$. The number of units of each of the $n$ securities in the portfolio is denoted by $\{u_j\}_{j=1}^{n}$. 
Although the portfolio is rebalanced over time, at each evaluation
date $t$ we fix the current holdings and recompute historical gains
as if those same holdings had been held at all earlier dates.
This produces a pseudo-history for the current portfolio, and the
procedure is repeated after each rebalance.
Theorem~\ref{thm:individual_cali_thm} concerns a fixed portfolio;
recomputing historical gains does not by itself extend its guarantee
to holdings selected using the observed history. We therefore
evaluate this changing-portfolio implementation empirically.

\begin{comment}
The Black-Scholes formula \eqref{eq:B-S-Option-pricing} gives
$$
C_t^{(j)} = \mbox{BS}(S_t, m_j S_t, \tau_j, \sigma_t^{(j)},\mbox{c}_j)
$$
for $j=1, \ldots , n-1$. The forms of the portfolio value at time $t$ and the portfolio gain at time $t+1$ are given below
\begin{equation}\label{eq:portfolio_value_on_u}
    V_t = u_1\cdot S_t + \sum_{j=1}^{n-1}u_{j+1}\cdot \mbox{BS}(S_t, m_j S_t, \tau_j, \sigma_t^{(j)},\mbox{c}_j)
\end{equation}
$$
G_{t+1} = u_1\cdot (S_{t+1} - S_t) + \sum_{j=1}^{n-1}u_{j+1}\cdot \left(\mbox{BS}(S_{t+1}, m_j S_t, \tau_j, \sigma_{t+1}^{(j)},\mbox{c}_j) - \mbox{BS}(S_{t}, m_j S_t, \tau_j, \sigma_{t}^{(j)},\mbox{c}_j)\right)
$$
where $u_j$ denotes the number of units of the $j^{\mathrm{th}}$ security for $j=1, \ldots , n$ that are held in the portfolio.
\end{comment}

\paragraph{Common Factor Model}
We apply the factor model framework introduced in Section~\ref{sec:Factor Structure} to the specific setting under consideration. We take the risk factors $\bx_t$ to be
\[
\bx_t \coloneqq \Big(\log S_t, \log\sigma_t^{(1)}, \ldots, \log\sigma_t^{(n-1)}\Big).
\]
Because the option parameters \(m_j\), \(\tau_j\), and portfolio positions \(u_j\) are fixed at a given time $t$, there exists a deterministic mapping \(v\) such that the portfolio value
$V_t = v(\mathbf{x}_t)$. We assume a 4-dimensional vector
$\bff_t = \Big(f_{t}^{(1)},   f_{t}^{(2)}, f_{t}^{(3)}, f_{t}^{(4)}\Big)$ of common risk factors.
The first factor $ f_{t}^{(1)}=\log(S_t)$ is the log S\&P 500 index price, and the remaining three factors drive the log-implied volatility surface $\mathbf{I}_t := (\log\sigma_t^{(1)}, \ldots, \log\sigma_t^{(n-1)})$. In particular, following  \citet{cont2002} we assume
\begin{equation}\label{eq:how_f_2_4_drives_I}
    \mathbf{I}_{t+1} = \bI_0 + f_{t+1}^{(2)} \cdot \mathbf{M}_2 + f_{t+1}^{(3)}\cdot \mathbf{M}_3 + f_{t+1}^{(4)} \cdot\mathbf{M}_4
\end{equation}
or equivalently,
\begin{equation}\label{eq:dynamic_I_is_iid}
     \log S_{t+1}/S_t = \Delta f_{t+1}^{(1)}, \quad \Delta\mathbf{I}_{t+1} =\Delta f_{t+1}^{(2)} \cdot \mathbf{M}_2 + \Delta f_{t+1}^{(3)}\cdot \mathbf{M}_3 + \Delta f_{t+1}^{(4)} \cdot\mathbf{M}_4,
\end{equation}
where $\Delta \bI_{t+1} = \bI_{t+1} - \bI_{t}$ and $\mathbf{M}_2$, $\mathbf{M}_3$, $\mathbf{M}_4$ and $\bI_0$ are constant vectors in $\mathbb{R}^{n-1}$. We note that \eqref{eq:dynamic_I_is_iid} agrees with the model structure specified in \eqref{eq:factormodel}, e.g., by setting $\bsym{\mu},\bfxi_t\equiv {\bf 0}$ and
\[
\bB = \begin{bmatrix}
    1 & 0 & 0 & 0 \\
    \mathbf{0} & \bM_2 & \bM_3 & \bM_4
\end{bmatrix}.
\]
Further details are provided in E-Companion \ref{app:mo_calibrate}.

\paragraph{Dynamics of Common Factors} The first factor follows a  Gaussian random walk with drift\footnote{This setting corresponds to the underlying index price $S_t$ following a geometric Brownian motion.}, and following \citet{cont2002}, we assume the remaining three factors each follow an Ornstein-Uhlenbeck process. In particular, we assume
\begin{equation}\label{eq:f_t_i}
\begin{aligned}
    \Delta f_{t+1}^{(1)} &= \lambda_1 + \epsilon_{t+1}^{(1)},\quad \epsilon_{t+1}^{(1)}\sim \mathcal{N}(0, \sigma_1^2) \\
    \Delta f_{t+1}^{(i)} &= \lambda_i + \gamma_i f_{t}^{(i)}+\epsilon_{t+1}^{(i)},\quad \epsilon_{t+1}^{(i)}\sim \mathcal{N}(0, \sigma_i^2)\quad \mbox{for}\,\, i=2,3,4
\end{aligned}
\end{equation}
where $\{\lambda_i\}_{i=1}^4, \{\gamma_i\}_{i=2}^4$ and $\{\sigma_i\}_{i=1}^4$ are model parameters. The random vector of the noise terms $(\epsilon_{t}^{(1)},\epsilon_{t}^{(2)},\epsilon_{t}^{(3)},\epsilon_{t}^{(4)})$ is independent of everything else and across time $t$. We further assume that for each $t$, $\epsilon_{t}^{(1)}$ is correlated with $\epsilon_{t}^{(2)}$, $\epsilon_t^{(3)}$ and $\epsilon_t^{(4)}$ through the following relationship
\begin{equation}\label{eq:epsilon_1}
    \epsilon_t^{(1)} = \eta_2 \epsilon_t^{(2)} + \eta_3 \epsilon_t^{(3)} + \eta_4 \epsilon_t^{(4)} + \epsilon_t^{G},
\end{equation}
where $\epsilon_t^{G}\sim \mathcal{N}(0, \sigma_G^2)$ and is independent of everything else. The variance $\sigma_G^2$ and the constants $\eta_2$, $\eta_3$, $\eta_4$ must satisfy
\begin{equation}\label{eq:sigma_G_constraint}
    \eta_2^2 \sigma_2^2 + \eta_3^2 \sigma_3^2 + \eta_4^2 \sigma_4^2 + \sigma_G^2= \sigma_1^2
\end{equation}
to ensure that the variance of $\epsilon_t^{(1)}$ is equal to $\sigma_1^2$.

\paragraph{Backtesting Process and Evaluation Metrics}

We backtest several algorithms on synthetic data generated from the common factor model. At each
time step $t$, each algorithm constructs a prediction interval for the scenario realized on the
following day. The algorithm's performance is then evaluated based on the portfolio gain realized
on that next day. We report four evaluation metrics: (1) the \textbf{realized coverage rate},
i.e.\ the fraction of time steps the realized portfolio gain falls within the prediction interval;
(2) the \textbf{prediction interval (PI) width}, i.e.\ the width of the prediction intervals
averaged across time steps; (3) the \textbf{interval (I) score}, which is a proper scoring rule
\citep{Gneiting_Raftery} that jointly assesses coverage rate and sharpness (further details are
provided in E-Companion~\ref{app:proper_scoring}); and (4) the \textbf{root mean squared error
(RMSE)}, i.e.\ the square root of the average (over all evaluation dates) of the squared
difference between the realized portfolio gain and their point predictions\footnote{For
ACSA, which does not produce an explicit point prediction for the expected portfolio gain, we use
the midpoint of the prediction interval as its point prediction when computing RMSE.}. Lower PI widths, I
scores, and RMSE values indicate better performance.

\subsection{The Skew Steepener Portfolio}\label{sec:stress_steepener}

We consider a {\bf skew steepener} portfolio that is long out-of-the-money (OTM)
put options and short OTM call options at 1-month, 3-month, and 6-month maturities.
The corresponding moneyness levels are 5\%, 10\%, and 20\% OTM, respectively.
For example, if the underlying index level is \$5000, the put strikes are
\$4750, \$4500, and \$4000, while the call strikes are \$5250, \$5500, and
\$6000. The common option position size $u^{\mathrm{skew}}$ is chosen so that
the combined market value of the long put positions equals \$100. An additional
underlying-index position $u^{\mathrm{hedge}}$ is used to maintain delta
neutrality. The resulting seven-position portfolio is rebalanced daily.

We stress the first two common factor returns, $\Delta\bff_{t+1}^{(1:2)}$.
The experiment generates 6000 observations from the dynamic factor model, with
the first 2000 used as burn-in and the remaining 4000 used for evaluation.
All methods construct 90\% prediction intervals.

Table \ref{tab:BackTest_Straddle} compares ACSA and KSA with two empirical-quantile
benchmarks and a Monte Carlo oracle. Implementation details are deferred to
E-Companion \ref{app:algo_details}, and computational times are reported in
E-Companion \ref{app:runtime}. The ``MC Oracle'' uses Monte Carlo approximations
to the oracle conditional quantiles under the known data-generating process.
The two ``Emp-quantile'' benchmarks instead use historical portfolio gains
recomputed under the current portfolio configuration. The ``Vanilla'' version
uses their empirical quantiles directly, while the ``SSA'' version adds the
current SSA point prediction to empirical quantiles of historical SSA residuals.
Neither benchmark conditions the interval width on the current scenario or
market state.

\begin{table}
\begin{center}
\begin{tabular}{ l c|c|c|c|c}
    \hline
    {\bf Algorithm} & {\bf Variants} & {\bf Coverage} & {\bf PI Width} $\downarrow$ & {\bf I Score} $\downarrow$ & {\bf RMSE} $\downarrow$\\
    \hline
    MC Oracle  & &  90.0\% &  27.2 & 1.72 & 8.59 \\
    \hline
    \multirow{2}{*}{Emp-quantile}
    &Vanilla & 99.8\% & 163.4 & 8.17 & 13.61\\
    &SSA &  92.1\% &  53.4 & 2.98 & 9.25 \\
    \hline
    \multirow{3}{*}{ACSA}
     & Linear-$g_t$ & 89.8\% & 30.7 & 2.27 & 11.05\\
     & NN-$g_t$ & 89.9\% & 28.5 & 1.89 & 9.27 \\
     & MC Oracle-$g_t$ & 89.9\% & 28.1 & 1.78 & 8.72\\
    \hline
    \multirow{3}{*}{KSA NN-$\varphi$}
    &  $\bH_t=\bff_t^{(2:4)}$ &  90.0\% & 28.0 & 1.75 & 8.76 \\
    &  $\bH_t=\bff_t^{(2:3)}$ & 91.7\% & 31.3 & 1.81 & 8.76  \\
    &  $\bH_t=\bff_t^{(2)}$ & 92.6\% & 37.4 &  1.92 & 8.76  \\
    \hline
    \multirow{4}{*}{KSA $\bH_t=\bff_t^{(2:4)}$}
    & no-$\varphi$ & 93.3\% & 50.6 & 3.87 & 16.6\\
    & SSA-$\varphi$ &  92.6\% & 34.6 & 1.94 & 9.25\\
    & NN-$\varphi$  & 90.0\% & 28.0 & 1.75 & 8.76 \\
    & MC Oracle-$\varphi$ & 90.1\% & 28.2 & 1.74 & 8.73 \\
    \hline
    \multirow{3}{*}{ACSA + KSA}
    &  $\bH_t = \bff_t^{(2:4)}$ & 90.5\% & 32.4 & 1.94 & 8.86\\
    &  $\bH_t = \bff_t^{(2:3)}$ & 90.7\% & 32.2 & 2.18 & 8.86\\
    &  $\bH_t = \bff_t^{(2)}$ & 90.3\% & 31.8 & 2.27 & 8.86\\
    \hline
\end{tabular}
\end{center}
\caption{MC Oracle denotes a Monte Carlo approximation to the oracle conditional quantities. Backtesting performance of the prediction intervals generated by the ACSA and KSA algorithms. The evaluation metrics include the realized coverage rate (Coverage), the prediction interval width (PI Width), the interval score (I Score), and the root mean squared error (RMSE). The target coverage for the prediction intervals is 90\%. Smaller values of PI Width, I Score and RMSE all indicate better performance. Each main row represents a major class of algorithms, and each subrow corresponds to a different variant. }\label{tab:BackTest_Straddle}
\end{table}

For ACSA, we consider three choices of quantile predictor $g_t$. ``Linear-$g_t$''
uses linear quantile regression, ``NN-$g_t$'' uses a two-layer neural network
with a 20-dimensional hidden layer, and ``MC Oracle-$g_t$'' uses Monte Carlo
approximations to the oracle conditional quantiles. The learned predictors use
$\Delta\bff_{t+1}^{(1:2)}$ as features and are trained using the pinball loss.
Because the portfolio evolves smoothly over time, the learned predictor is
fine-tuned from the previous day's model rather than retrained from scratch.
All three ACSA variants attain coverage close to the 90\% target. The nonlinear
and oracle predictors also produce interval widths and scores close to the MC
Oracle benchmark, whereas the linear predictor yields noticeably wider intervals.
This illustrates the importance of the quality of the underlying
quantile predictor for interval width and score.

For KSA, we study separately the effects of the historical context $\bH_t$ and
the expectation predictor $\varphi$. In the ``KSA NN-$\varphi$'' block,
$\varphi$ is a two-layer neural network and we compare
$\bH_t=\bff_t^{(2:4)}$, $\bff_t^{(2:3)}$, and $\bff_t^{(2)}$.
The index-price level $f_t^{(1)}$ is excluded because, under our model, the
next-day index return is independent of the current index level and the normalized
option gain depends on the return rather than the price level itself; see
E-Companion \ref{appdix:his_context_of_KSA}. In the
``KSA $\bH_t=\bff_t^{(2:4)}$'' block, we instead fix the full historical context
and compare no centering ($\varphi\equiv0$), the SSA predictor, a neural-network
predictor, and a Monte Carlo oracle predictor.

The results show that KSA improves as either the historical context or the
centering predictor becomes more informative. With
$\bH_t=\bff_t^{(2:4)}$ and either NN-$\varphi$ or MC Oracle-$\varphi$,
coverage is essentially at the 90\% target and the remaining performance
metrics are close to the MC Oracle benchmark. Performance deteriorates
moderately when either $\bH_t$ or $\varphi$ is misspecified, illustrating the
value of both informative conditioning variables and an effective centering
predictor.

Finally, the ``ACSA + KSA'' variants use the KSA conditional-quantile estimator
as the quantile predictor within ACSA, while retaining the online update of
$\tilde{\alpha}_t$. A neural network is used for $\varphi$ under the same three
choices of $\bH_t$. The hybrid achieves coverage close to the nominal level and
performance comparable to standard ACSA. It can also be computationally
attractive when portfolios change frequently, since the learned component is an
expectation predictor rather than a portfolio-level quantile model.

\subsection{An Adversarial Portfolio}\label{sec:adversarial_portfolio}
We now consider a portfolio manager operating under the oversight of a risk management team. The manager forms opinions regarding future market conditions and aims to maximize the expected portfolio returns while simultaneously complying with risk management requirements. We assume the risk management requirements are specified via SSA across a range of predefined stress scenarios.

An \textit{adversarial action} arises when the portfolio manager has access to information that is not available to the risk management team and exploits this informational asymmetry to construct portfolios that appear ``safe'' under the prescribed stress tests, yet contain directional exposure to  risk factors that are not monitored by the risk management team. We assume  the portfolio manager knows the true four-factor model, whereas the risk management team monitors only the first two factors. Under this setup, the portfolio manager can construct a portfolio that is neutral or safe with respect to the monitored factors according to SSA, but effectively represents a bet on upward movements of the third factor. We refer to such a portfolio as an \textbf{adversarial portfolio}. We will show that SSA fails to detect the hidden risks undertaken by the adversarial portfolio, whereas the ACSA and KSA methods are able to reveal them.

The adversarial portfolio consists of the S\&P 500 index together with 12 options. The options have maturities of 1-month, 3-month, and 6-month, with moneyness levels of 5\% and 10\% OTM for both call and put options. The portfolio manager is allowed to take short positions in any of these securities. The total notional value of long positions is constrained to be at most \$100, and the sum of absolute position sizes across all securities is required to remain below \$500. The adversarial portfolio is constructed by solving a linear programming (LP) problem. The LP maximizes a linear approximation of the expected portfolio gain under the betting movement (i.e., an upward movement of the third common factor). The constraints enforce that all SSA evaluations remain within the range $\pm \$3 $, along with additional constraints on the total notional value of long positions and the sum of absolute position sizes. Further details of the portfolio construction are contained in E-Companion \ref{appdix:adversarial_portfolio}. The backtesting results for the adversarial portfolio are reported in Table \ref{tab:BackTest_Adversarial}, which is organized analogously to Table \ref{tab:BackTest_Straddle}. Because the risk management team has access only to the first two common factors, information on $f_t^{(3)}$ and $f_t^{(4)}$ cannot be included in the feature inputs when training the neural networks. This restriction does not affect ACSA, since it relies only on the first two factors (recall that ``Linear-$g_t$'' and ``NN-$g_t$'' correspond to training a linear model and a neural network as the quantile predictor, using $\Delta \bff_{t+1}^{(1:2)}$ as the feature input). The KSA, however, is evaluated only under the configuration $\bH_t = \bff_t^{(2)}$, rather than the full historical context $\bH_t = \bff_t^{(2:4)}$. Accordingly, the expectation predictor $\varphi$ used by KSA is also trained under the restricted information set, with $(\Delta \bff_{t+1}^{(1:2)}, \bff_t^{(2)})$ as the feature input.

We observe that the performance patterns across different algorithms in Table \ref{tab:BackTest_Adversarial} closely mirror those observed in Table \ref{tab:BackTest_Straddle}. The two ``Emp-quantile''-based benchmarks produce exceptionally wide intervals. In contrast, our proposed algorithms demonstrate large and consistent advantages across different portfolio constructions. One caveat is that, due to the absence of a reliable $\varphi$, the overall performance of KSA still shows a noticeable gap relative to the Monte Carlo oracle benchmark.

\begin{table}
\begin{center}
\begin{tabular}{ l c|c|c|c|c}
    \hline
    {\bf Algorithm} & {\bf Variants} & {\bf  Coverage} & {\bf PI Width} $\downarrow$ & {\bf I Score} $\downarrow$ & {\bf RMSE} $\downarrow$\\
    \hline
    MC Oracle  & &  90.0\% & 15.8 & 1.00 & 5.42 \\
    \hline
    \multirow{2}{*}{Emp-quantile}
    &Vanilla & 99.6\% & 2645.4 & 132.3 & 4129.4 \\
    &SSA &  91.7\% & 1113.3 & 55.9 & 5.86\\
    \hline
    \multirow{3}{*}{ACSA}
     & Linear-$g_t$ & 89.7\% & 56.2 & 3.54 & 7.56\\
     & NN-$g_t$ & 89.9\% & 19.3 & 1.69 & 5.67\\
     & MC Oracle-$g_t$ & 90.0\% & 16.5 & 1.06 & 5.44\\
    \hline
    \multirow{4}{*}{KSA $\bH_t=\bff_t^{(2)}$}
    & no-$\varphi$ & 93.8\% & 32.6 & 2.35 & 8.83\\
    & SSA-$\varphi$ & 87.3\% & 19.1 & 1.59 & 5.86\\
    & NN-$\varphi$  & 92.6\% & 20.9 & 1.42 & 5.43 \\
    & MC Oracle-$\varphi$ & 90.2\% & 16.9 & 1.14 & 5.40 \\
    \hline
    \multirow{1}{*}{ACSA + KSA}
    &  $\bH_t = \bff_t^{(2)}$ & 90.2\% & 29.1 & 1.85 & 5.87\\
    \hline
\end{tabular}
\end{center}
\caption{Backtesting performance of the prediction intervals under the adversarial portfolio. The table structure mirrors that of Table \ref{tab:BackTest_Straddle}, but we only report results for $\bH_t = \bff_t^{(2)}$ in the KSA setting, since the risk management team has access only to the first two common factors. For the KSA SSA-$\varphi$ row, RMSE is computed from the SSA point predictor, matching the Emp-quantile SSA row. }\label{tab:BackTest_Adversarial}
\end{table}

The next experiment demonstrates the scenario analysis table at a specific date. We show that, although the adversarial portfolio appears neutral to risk across a range of scenarios under the SSA estimation, it nonetheless produces notably wide prediction intervals using the ACSA and KSA algorithms that we propose. The width of the prediction interval serves as a warning signal for the risk management team. For comparison, we analyze a \textbf{first-order four-factor-neutral benchmark portfolio}.\footnote{We use ``factor-neutral'' here in a relative sense. It does not imply sensitivity to each factor is exactly zero. Rather, relative to the adversarial portfolio, this benchmark is more neutral with respect to the factor directions.}
Its optimization imposes additional constraints requiring the linearized one-day gain to vanish under unit movements in each of the four common-factor directions. Thus, relative to the adversarial construction, the benchmark is more tightly controlled with respect to the modeled factor directions. The construction adds constraints for the third and fourth factor directions to the adversarial-portfolio formulation; see E-Companion~\ref{appdix:adversarial_portfolio}.

% We will see below that, despite the two portfolios exhibiting similar SSA tables, their KSA results differ markedly,  the adversarial portfolio yields substantially wider prediction intervals than the risk-neutral portfolio.

In Table \ref{tab:SSA-Cond-Exp-gain}, we report the SSA estimates and the Monte Carlo approximations to the oracle conditional expectations of portfolio gains under various stress scenarios. By construction, the SSA estimates for both portfolios are tightly controlled, remaining within $\pm \$3$ relative to the initial long position of $\$100$. However, the Monte Carlo approximations to the oracle conditional expectations $\mathbb{E}[G_{t+1} \mid \bW_{t+1}(\bz)]$ under each scenario reveal substantial differences between the two portfolios. In particular, the expected gains of the factor-neutral benchmark are roughly bounded within $\pm \$6.5$, whereas the adversarial portfolio can incur expected losses as large as $-\$20$. This discrepancy highlights the risk of relying solely on SSA-based assessments. Our proposed ACSA and KSA methods address this issue, as shown below.

\begin{table}[ht]
  \centering
\begin{minipage}[t]{0.48\textwidth}
  \centering
  \captionof*{table}{SSA: Adversarial Portfolio}
  \vspace{-0.1cm}
  \resizebox{\linewidth}{2.0cm}{%
    \begin{tabular}{r *{7}{r}}
      \toprule
      \mbox{$\Delta f_{t+1}^{(1)}$ \, }  & \multicolumn{7}{c}{\mbox{$ \Delta f_{t+1}^{(2)} $ ($\times 10^{-1}$)}} \\
      \cmidrule(l){2-8}
      ($\times 10^{-2}$) & $-3$ & $-2$ & $-1$ & $0$ & $1$ & $2$ & $3$ \\
      \midrule
      $-1.5$ & 0.86 & 0.10 & -0.62 & -1.29 & -1.91 & -2.48 & -3.00 \\
      $-1.0$ & 1.57 & 1.09 & 0.65 & 0.25 & -0.10 & -0.41 & -0.68 \\
      $-0.5$ & 1.81 & 1.60 & 1.43 & 1.29 & 1.20 & 1.15 & 1.13 \\
      $0.0$ & 1.52 & 1.58 & 1.67 & 1.80 & 1.96 & 2.15 & 2.37 \\
      $0.5$ & 0.65 & 0.97 & 1.32 & 1.70 & 2.11 & 2.54 & 3.00 \\
      $1.0$ & -0.84 & -0.26 & 0.34 & 0.96 & 1.60 & 2.27 & 2.95 \\
      $1.5$ & -3.00 & -2.18 & -1.34 & -0.49 & 0.39 & 1.28 & 2.18 \\
      \bottomrule
    \end{tabular}
  }
\end{minipage}
\hfill
\begin{minipage}[t]{0.48\textwidth}
  \centering
  \captionof*{table}{$\mathbb{E}[G_{t+1} \mid \bW_{t+1}(\bz)]$: Adversarial Portfolio}
  \vspace{-0.1cm}
  \resizebox{\linewidth}{2.0cm}{%
    \begin{tabular}{r *{7}{r}}
      \toprule
      \mbox{$\Delta f_{t+1}^{(1)}$ \, }& \multicolumn{7}{c}{\mbox{$ \Delta f_{t+1}^{(2)} $ ($\times 10^{-1}$)}} \\
      \cmidrule(l){2-8}
      ($\times 10^{-2}$) & $-3$ & $-2$ & $-1$ & $0$ & $1$ & $2$ & $3$ \\
      \midrule
      $-1.5$ & 2.66 & 3.00 & 3.39 & 3.80 & 4.11 & 4.54 & 4.94 \\
      $-1.0$ & 1.43 & 2.42 & 3.28 & 4.13 & 5.00 & 5.83 & 6.61 \\
      $-0.5$ & -0.77 & 0.78 & 2.17 & 3.56 & 4.92 & 6.16 & 7.39 \\
      $0.0$ & -3.98 & -1.91 & 0.06 & 1.98 & 3.73 & 5.53 & 7.20 \\
      $0.5$ & -8.29 & -5.75 & -3.11 & -0.78 & 1.57 & 3.83 & 5.91 \\
      $1.0$ & -13.87 & -10.61 & -7.54 & -4.58 & -1.68 & 1.00 & 3.65 \\
      $1.5$ & -20.47 & -16.72 & -13.14 & -9.57 & -6.16 & -2.98 & 0.10 \\
      \bottomrule
    \end{tabular}
  }
\end{minipage}

\vspace{0.2cm}
\begin{minipage}[t]{0.48\textwidth}
  \centering
  \captionof*{table}{SSA: Factor-neutral benchmark}
  \vspace{-0.1cm}
  \resizebox{\linewidth}{2.0cm}{%
    \begin{tabular}{r *{7}{r}}
      \toprule
      \mbox{$\Delta f_{t+1}^{(1)}$ \, }& \multicolumn{7}{c}{\mbox{$ \Delta f_{t+1}^{(2)} $ ($\times 10^{-1}$)}} \\
      \cmidrule(l){2-8}
      ($\times 10^{-2}$) & $-3$ & $-2$ & $-1$ & $0$ & $1$ & $2$ & $3$ \\
      \midrule
      $-1.5$ & -0.11 & -0.58 & -1.02 & -1.42 & -1.79 & -2.14 & -2.45 \\
      $-1.0$ & -0.48 & -0.77 & -1.04 & -1.27 & -1.48 & -1.66 & -1.82 \\
      $-0.5$ & -0.73 & -0.85 & -0.94 & -1.01 & -1.05 & -1.07 & -1.06 \\
      $0.0$ & -0.93 & -0.87 & -0.78 & -0.68 & -0.55 & -0.41 & -0.24 \\
      $0.5$ & -1.12 & -0.88 & -0.62 & -0.34 & -0.05 & 0.26 & 0.58 \\
      $1.0$ & -1.36 & -0.94 & -0.51 & -0.06 & 0.39 & 0.87 & 1.35 \\
      $1.5$ & -1.70 & -1.11 & -0.51 & 0.10 & 0.73 & 1.36 & 2.00 \\
      \bottomrule
    \end{tabular}
  }
\end{minipage}
\hfill
\begin{minipage}[t]{0.48\textwidth}
  \centering
  \captionof*{table}{$\mathbb{E}[G_{t+1} \mid \bW_{t+1}(\bz)]$: Factor-neutral benchmark}
  \vspace{-0.1cm}
  \resizebox{\linewidth}{2.0cm}{%
    \begin{tabular}{r *{7}{r}}
      \toprule
      \mbox{$\Delta f_{t+1}^{(1)}$ \, }& \multicolumn{7}{c}{\mbox{$ \Delta f_{t+1}^{(2)} $ ($\times 10^{-1}$)}} \\
      \cmidrule(l){2-8}
      ($\times 10^{-2}$) & $-3$ & $-2$ & $-1$ & $0$ & $1$ & $2$ & $3$ \\
      \midrule
      $-1.5$ & -0.23 & -1.34 & -2.40 & -3.43 & -4.40 & -5.31 & -6.17 \\
      $-1.0$ & 0.01 & -0.72 & -1.44 & -2.12 & -2.77 & -3.38 & -3.94 \\
      $-0.5$ & -0.16 & -0.51 & -0.85 & -1.17 & -1.47 & -1.74 & -2.01 \\
      $0.0$ & -0.84 & -0.78 & -0.72 & -0.66 & -0.59 & -0.51 & -0.42 \\
      $0.5$ & -2.06 & -1.59 & -1.13 & -0.67 & -0.21 & 0.24 & 0.68 \\
      $1.0$ & -3.92 & -3.02 & -2.14 & -1.27 & -0.42 & 0.42 & 1.24 \\
      $1.5$ & -6.43 & -5.11 & -3.82 & -2.53 & -1.27 & -0.04 & 1.16 \\
      \bottomrule
    \end{tabular}
  }
\end{minipage}

\caption{SSA point estimates and the Monte Carlo oracle conditional mean gains (measured in \$) across the scenario grid defined by the first two common factor returns $(\Delta f_{t+1}^{(1)}, \Delta f_{t+1}^{(2)})$. The two tables in the top row correspond to the adversarial portfolio and report the SSA estimates and the Monte Carlo oracle conditional expected scenario gain $\mathbb{E}[G_{t+1} \mid \bW_{t+1}(\bz)]$. The bottom row reports the corresponding quantities for the factor-neutral benchmark. Although the SSA estimates for the two portfolios appear similar (with values contained within $\pm \$3 $), the Monte Carlo approximations to the oracle conditional expectations differ substantially, in particular the adversarial portfolio can experience scenario losses as large as $-\$ 20$.
}
\label{tab:SSA-Cond-Exp-gain}
\end{table}

Table \ref{tab:ksa-pi-adversarial-n-neutral} reports the scenario tables of prediction intervals generated by our proposed algorithms, evaluated at a specific day (additional results for a different day are provided in the E-Companion). In the ACSA algorithm, a two-layer neural network is trained as the quantile predictor $g_t$. In the KSA algorithm, the SSA estimate is used as the expectation predictor $\varphi$, since the tables are produced from the perspective of a risk manager and we assume that this represents the best available estimator in practice. Comparing the prediction intervals for the two portfolios, we observe that, under both algorithms, the intervals for the adversarial portfolio are substantially wider than those for the factor-neutral benchmark. This finding suggests that the width of the prediction intervals can serve as an additional indicator of portfolio risk, beyond relying on a single SSA estimate alone. A clearer illustration of the time dynamics of the prediction interval width (under the realized scenarios) is provided in Figure \ref{fig:ACSA_dynamics_adv_n_neu}, which shows that the prediction intervals of the adversarial portfolio are consistently wider than those of the factor-neutral benchmark.

\begin{figure}
    \centering
    \includegraphics[width=0.9\linewidth]{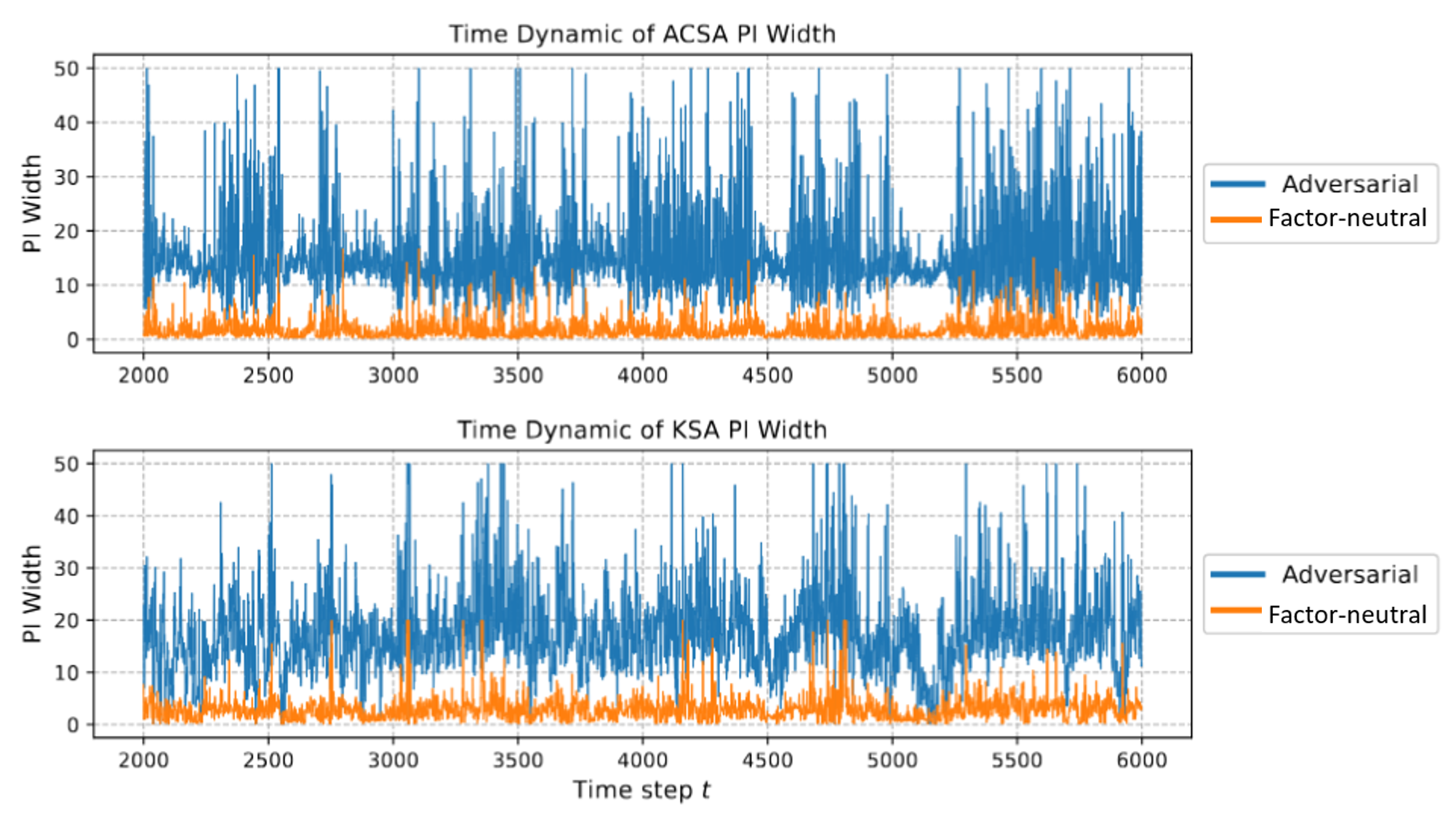}
    \caption{Dynamics of ACSA and KSA's prediction interval (PI) widths, for the adversarial portfolio and the factor-neutral benchmark. On average, the adversarial portfolio has notably wider prediction intervals than the factor-neutral benchmark.
   }\label{fig:ACSA_dynamics_adv_n_neu}
\end{figure}

\begin{table}[p]
\centering
\renewcommand{\arraystretch}{0.95}
\setlength{\parskip}{0pt}
\captionsetup{skip=0pt}
\setlength{\belowcaptionskip}{0pt}
\begin{minipage}[t]{0.90\textwidth}
  \centering
  \captionof*{table}{ACSA 90\% PI ($\tilde{\alpha}_t = 0.08$): Adversarial Portfolio}
  \resizebox{\linewidth}{!}{%
    \begin{tabular}{r *{7}{l}}
      \toprule
      \mbox{$\Delta f_{t+1}^{(1)}$ \, }  & \multicolumn{7}{c}{\mbox{$ \Delta f_{t+1}^{(2)} $ ($\times 10^{-1}$)}} \\
      \cmidrule(l){2-8}
      ($\times 10^{-2}$) & $-3$ & $-2$ & $-1$ & $0$ & $1$ & $2$ & $3$ \\
      \midrule
      $-1.5$ & (-7.1, 12.8) & (-6.0, 12.5) & (-4.9, 12.2) & (-3.8, 11.9) & (-2.8, 11.6) & (-1.8, 11.4) & (-0.9, 11.2) \\
      $-1.0$ & (-8.4, 11.8) & (-6.7, 12.0) & (-5.1, 12.2) & (-3.6, 12.3) & (-2.1, 12.5) & (-0.7, 12.6) & (0.7, 12.8) \\
      $-0.5$ & (-10.9, 9.8) & (-8.6, 10.5) & (-6.4, 11.2) & (-4.4, 11.9) & (-2.4, 12.5) & (-0.5, 13.1) & (1.3, 13.7) \\
      $0.0$ & (-14.4, 6.8) & (-11.6, 8.1) & (-8.9, 9.3) & (-6.2, 10.5) & (-3.7, 11.6) & (-1.4, 12.7) & (0.9, 13.7) \\
      $0.5$ & (-19.1, 2.9) & (-15.7, 4.7) & (-12.4, 6.4) & (-9.3, 8.1) & (-6.2, 9.7) & (-3.4, 11.2) & (-0.6, 12.7) \\
      $1.0$ & (-25.0, -2.3) & (-21.0, 0.1) & (-17.2, 2.4) & (-13.5, 4.6) & (-9.9, 6.7) & (-6.5, 8.8) & (-3.3, 10.7) \\
      $1.5$ & (-32.1, -8.6) & (-27.6, -5.6) & (-23.2, -2.8) & (-18.9, -0.0) & (-14.8, 2.6) & (-10.9, 5.2) & (-7.2, 7.6) \\
      \bottomrule
    \end{tabular}
  }
\end{minipage}
\par\nointerlineskip\vspace{1pt}
\begin{minipage}[t]{0.90\textwidth}
  \centering
  \captionof*{table}{ACSA 90\% PI ($\tilde{\alpha}_t = 0.09$): Factor-neutral benchmark}
  \resizebox{\linewidth}{!}{%
    \begin{tabular}{r *{7}{l}}
      \toprule
      \mbox{$\Delta f_{t+1}^{(1)}$ \, }  & \multicolumn{7}{c}{\mbox{$ \Delta f_{t+1}^{(2)} $ ($\times 10^{-1}$)}} \\
      \cmidrule(l){2-8}
      ($\times 10^{-2}$) & $-3$ & $-2$ & $-1$ & $0$ & $1$ & $2$ & $3$ \\
      \midrule
      $-1.5$ & (-1.3, 1.2) & (-2.3, -0.2) & (-3.5, -1.4) & (-4.7, -2.3) & (-5.9, -3.1) & (-7.1, -3.7) & (-8.2, -4.2) \\
      $-1.0$ & (-0.9, 1.3) & (-1.4, 0.2) & (-2.2, -0.7) & (-3.1, -1.3) & (-4.0, -1.7) & (-4.9, -1.9) & (-5.8, -2.1) \\
      $-0.5$ & (-1.2, 1.2) & (-1.0, 0.4) & (-1.3, -0.3) & (-1.8, -0.7) & (-2.4, -0.6) & (-3.0, -0.5) & (-3.6, -0.3) \\
      $0.0$ & (-2.1, 0.6) & (-1.6, 0.2) & (-1.1, -0.2) & (-0.9, -0.3) & (-1.2, 0.1) & (-1.5, 0.6) & (-1.8, 1.1) \\
      $0.5$ & (-3.7, -0.4) & (-2.8, -0.4) & (-1.9, -0.5) & (-1.1, -0.3) & (-0.6, 0.4) & (-0.5, 1.2) & (-0.5, 2.0) \\
      $1.0$ & (-6.0, -1.9) & (-4.6, -1.5) & (-3.4, -1.1) & (-2.1, -0.6) & (-1.1, 0.3) & (-0.2, 1.4) & (0.2, 2.5) \\
      $1.5$ & (-9.0, -4.0) & (-7.2, -3.2) & (-5.5, -2.3) & (-3.9, -1.4) & (-2.3, -0.3) & (-1.0, 1.0) & (0.1, 2.4) \\
      \bottomrule
    \end{tabular}
  }
\end{minipage}
\par\nointerlineskip\vspace{1pt}
\begin{minipage}[t]{0.90\textwidth}
  \centering
  \captionof*{table}{KSA 90\% PI: Adversarial Portfolio}
  \resizebox{\linewidth}{!}{%
    \begin{tabular}{r *{7}{l}}
      \toprule
      \mbox{$\Delta f_{t+1}^{(1)}$ \, }& \multicolumn{7}{c}{\mbox{$ \Delta f_{t+1}^{(2)} $ ($\times 10^{-1}$)}} \\
      \cmidrule(l){2-8}
      ($\times 10^{-2}$) & $-3$ & $-2$ & $-1$ & $0$ & $1$ & $2$ & $3$ \\
      \midrule
      $-1.5$ & (-7.3, 11.4) & (-5.9, 11.7) & (-5.0, 13.5) & (-4.3, 19.4) & (-2.7, 18.7) & (-2.8, 18.2) & (-2.6, 17.6) \\
      $-1.0$ & (-7.7, 11.8) & (-7.1, 11.3) & (-5.9, 11.2) & (-4.8, 14.4) & (-3.2, 20.5) & (-1.2, 20.2) & (-1.0, 20.0) \\
      $-0.5$ & (-10.8, 11.0) & (-7.7, 11.2) & (-6.8, 11.6) & (-5.5, 11.8) & (-4.3, 13.9) & (-2.0, 21.8) & (0.2, 21.8) \\
      $0.0$ & (-12.3, 9.3) & (-10.5, 9.8) & (-8.2, 11.1) & (-6.8, 11.4) & (-5.0, 12.1) & (-3.8, 14.3) & (-0.7, 21.9) \\
      $0.5$ & (-13.2, 8.5) & (-12.4, 8.8) & (-10.0, 9.2) & (-8.7, 9.5) & (-6.9, 10.4) & (-5.0, 12.1) & (-3.7, 13.5) \\
      $1.0$ & (-11.2, 7.0)$^{\textcolor{red}{*}}$ & (-13.0, 7.6) & (-12.4, 8.2) & (-10.4, 8.6) & (-9.0, 9.2) & (-7.1, 9.9) & (-5.4, 11.2) \\
      $1.5$ & (-13.0, 4.8)$^{\textcolor{red}{*}}$ & (-12.5, 5.6)$^{\textcolor{red}{*}}$ & (-11.7, 6.5)$^{\textcolor{red}{*}}$ & (-11.9, 7.3) & (-11.0, 7.3) & (-9.3, 8.7) & (-8.1, 9.8) \\
      \bottomrule
    \end{tabular}
  }
\end{minipage}
\par\nointerlineskip\vspace{1pt}
\begin{minipage}[t]{0.90\textwidth}
  \centering
  \captionof*{table}{KSA 90\% PI: Factor-neutral benchmark}
  \resizebox{\linewidth}{!}{%
    \begin{tabular}{r *{7}{l}}
      \toprule
      \mbox{$\Delta f_{t+1}^{(1)}$ \, }& \multicolumn{7}{c}{\mbox{$ \Delta f_{t+1}^{(2)} $ ($\times 10^{-1}$)}} \\
      \cmidrule(l){2-8}
      ($\times 10^{-2}$) & $-3$ & $-2$ & $-1$ & $0$ & $1$ & $2$ & $3$ \\
      \midrule
      $-1.5$ & (-3.1, 3.1) & (-3.9, 2.3) & (-4.5, 1.3) & (-5.2, 0.9) & (-5.9, 0.5) & (-6.2, 0.2) & (-6.5, -0.4) \\
      $-1.0$ & (-3.3, 3.0) & (-3.8, 2.4) & (-4.3, 1.8) & (-4.8, 1.3) & (-5.2, 0.8) & (-5.7, 0.7) & (-5.9, 0.5) \\
      $-0.5$ & (-3.7, 2.9) & (-3.7, 2.6) & (-3.9, 2.2) & (-4.1, 1.9) & (-4.3, 1.7) & (-4.7, 1.3) & (-5.1, 1.3) \\
      $0.0$ & (-4.1, 4.6) & (-4.0, 2.6) & (-3.7, 2.6) & (-3.7, 2.5) & (-3.6, 2.3) & (-3.6, 2.4) & (-3.8, 2.2) \\
      $0.5$ & (-4.7, 6.9) & (-4.0, 4.8) & (-3.8, 2.6) & (-3.3, 2.8) & (-3.0, 3.1) & (-2.8, 3.2) & (-2.5, 3.5) \\
      $1.0$ & (-4.5, 6.7) & (-4.1, 7.1) & (-3.7, 4.9) & (-3.2, 3.1) & (-2.8, 3.5) & (-2.2, 3.9) & (-1.7, 4.3) \\
      $1.5$ & (-4.3, 6.4)$^{\textcolor{red}{*}}$ & (-4.3, 6.9)$^{\textcolor{red}{*}}$ & (-3.7, 7.5)$^{\textcolor{red}{*}}$ & (-3.1, 4.0) & (-2.4, 3.7) & (-1.8, 4.3) & (-1.2, 5.0) \\
      \bottomrule
    \end{tabular}
  }
\end{minipage}
\caption{ACSA and KSA 90\% prediction intervals for the adversarial portfolio and the factor-neutral benchmark. In ACSA, the two portfolios follow different trajectories in the time dynamics of $\tilde{\alpha}_t$, and therefore the quantile arguments differ at the evaluated time step for the two tables ($\tilde{\alpha}_t = 0.08$ for the adversarial portfolio and $\tilde{\alpha}_t = 0.09$ for the factor-neutral benchmark). The KSA uses SSA estimation as the mean predictor. A red $^{\textcolor{red}{*}}$ superscript indicates scenarios where the prediction interval \textit{does not contain} the Monte Carlo oracle conditional expected gain $\mathbb{E}[G_{t+1} \mid \bW_{t+1}(\bz)]$ reported in Table~\ref{tab:SSA-Cond-Exp-gain}. As was the case with Table \ref{tab:ConvexIntervals_wACSA_wKSA_Sec5_2},
this is only a diagnostic for interval location relative to the conditional mean $\mathbb{E}[G_{t+1} \mid \bW_{t+1}(\bz)]$. In particular, exclusion of $\mathbb{E}[G_{t+1} \mid \bz ]$ is not a validity failure since the intervals target realized $G_{t+1} \mid \bW_{t+1}(\bz)$ and not $\mathbb{E}[G_{t+1} \mid \bW_{t+1}(\bz)]$.}
%\caption{KSA 90\% prediction intervals for the adversarial and risk-neutral portfolio, using SSA as the mean predictor. A red $^{\textcolor{red}{*}}$ superscript indicates scenarios where the prediction interval \textit{does not contain} the true conditional expected gain (Cond Exp) reported in Table~\ref{tab:SSA-Cond-Exp-gain}. Most intervals successfully cover the Cond Exp. Moreover, the prediction intervals for the adversarial portfolio are significantly wider than those for the risk-neutral portfolio. This indicates higher uncertainty and reveals the presence of hidden risk.}
\label{tab:ksa-pi-adversarial-n-neutral}
\end{table}
\clearpage

\section{Conclusions and Further Research}
\label{sec:conclusions}

Scenario analysis is widely used in financial risk management, but standard
implementations typically report only point estimates of scenario P\&L with
little statistical validation. We instead formulate scenario analysis as a
sequential prediction problem and augment scenario tables with prediction
intervals for realized scenario gains. ACSA provides online marginal calibration,
while KSA uses the stressed scenario and current market state to estimate
scenario-conditional uncertainty.

Our experiments show that these methods can provide more informative assessments
than standard scenario analysis, particularly when dependence, nonlinear
exposures, or unmonitored risk factors make conventional scenario P\&L estimates
misleading. The two approaches are complementary: ACSA offers a broadly
applicable calibration mechanism, whereas KSA provides more scenario-specific
uncertainty estimates when sufficient structure and historical support are
available. KSA can also be computationally attractive for changing portfolios
when conditional expectations are estimated at the security level and aggregated
across positions.

Several directions for future research remain. One is to improve the scalability
of KSA when the conditioning information $H_t$ is high-dimensional. Promising
approaches include structured or additive kernels, sparsity assumptions, and
learned low-dimensional representations $r(H_t)$ designed to preserve the
information relevant for scenario-loss quantiles. More generally, this amounts
to learning an efficient notion of similarity between market states while
retaining useful conditional-coverage behavior.

A second direction is to integrate scenario design more closely with statistical
calibration. Rather than specifying a scenario set first and validating it only
afterwards, one could incorporate plausibility and severity directly into the
learning and calibration process, for example by allocating more resolution or
weight to plausible but severe scenarios. This could lead to scenario tables
that are jointly designed to reflect risk-management objectives and to support
systematic statistical validation.

%{\bf Acknowledgment} The authors thank

\clearpage
\phantomsection\label{part:main-references}
\putbib
\end{bibunit}

\appendix
\startappendixpart{EC}
\setbibliographypart{ec}
\begin{bibunit}
\phantomsection\label{part:ec-start}
\begin{center}
{\Large\bfseries E-Companion}\\[0.5\baselineskip]
{\large Stress Testing Financial Portfolios with Coverage Guarantees}
\end{center}

\section{Additional Supporting Material}

\subsection{Notation Summary}
\label{subsec:notation_summary}
Table~\ref{tab:notation_summary} summarizes the main symbols used throughout the paper.
\renewcommand{\arraystretch}{1.12}
\begin{longtable}{p{0.21\textwidth}p{0.71\textwidth}}
\caption{Notation summary.}\label{tab:notation_summary}\\
\toprule
Symbol & Meaning\\
\midrule
\endfirsthead
\multicolumn{2}{l}{\textit{Table \thetable\ (continued)}}\\
\toprule
Symbol & Meaning\\
\midrule
\endhead
\midrule
\multicolumn{2}{r}{\textit{Continued on next page}}\\
\endfoot
\bottomrule
\endlastfoot
\multicolumn{2}{l}{\textbf{General market and scenario notation}}\\
$t$ & Time index.\\
$V_t$ & Portfolio value at time $t$.\\
$G_{t+1}$ & One-step portfolio gain / P\&L.\\
${\cal F}_t$ & Information available at time $t$.\\
$\bx_t$ & Vector of risk factors.\\
$\bff_t$ & Vector of common factors.\\
$\Delta \bx_{t+1}, \Delta \bff_{t+1}$ & One-step changes in the risk factors and common factors.\\
$\bB$ & Factor loading matrix in the factor model.\\
$\bz$ & Stress scenario applied to selected common factor returns.\\
$\bz_t$ & Realized stress scenario on date $t$.\\
$v(\cdot)$ & Pricing map from risk factors to portfolio value.\\
 \\

\multicolumn{2}{l}{\textbf{ACSA notation}}\\
$\alpha$ & Target miscoverage level.\\
$\tilde{\alpha}_t$ & Time-$t$ adjusted miscoverage level used by ACSA.\\
$\bar{\tilde{\alpha}}$&  Average adjusted miscoverage level used in the ACSA simulations.\\
$\gamma$ & ACSA step size.\\
$g_t(\bz;\tilde{\alpha})$ & Scenario-wise quantile predictor evaluated at quantile level $\tilde{\alpha}$.\\
$\Ct_t(\bz)$ & Prediction interval reported for scenario $\bz$.\\
$\mathrm{err}_t$ & Indicator that the previously reported interval misses the realized gain.\\
$B_1,\ldots,B_m$ & Scenario groups in group-balanced ACSA.\\
$\tilde{\alpha}_{t,k}$ & Group-specific adjusted miscoverage level.\\
 \\
\multicolumn{2}{l}{\textbf{KSA notation}}\\
$u$ & Quantile level arguments used in pinball-loss and quantile expressions.\\
$\ell_u(y,q)$ & Pinball loss at quantile level $u$.\\
$Q_Y(u)$ & $u$-quantile of a random variable $Y$.\\
$\Ht{t}$ & Historical context vector.\\
$\bW_t(\bz)$ & Feature vector $(\bz,\Ht{t-1})$.\\
$\varphi(\bw)$ & Expectation predictor of the scenario gain.\\
$Y_t$ & Residual of the portfolio gain after subtracting the predicted expectation.\\
$\kappa(\bw,\bw')$ & Kernel similarity function.\\
$\widehat{\ell}_t(u;q,\bw)$ & Kernel-reweighted empirical pinball loss.\\
$\widehat{Q}_t(u;\bW_{t+1}(\bz))$ & KSA estimator of the residual $u$-quantile.\\
$\mathcal{S}$ & Set of scenarios in the scenario table.\\
$h$ & Kernel bandwidth.\\
 \\
\multicolumn{2}{l}{\textbf{Running example and experiments}}\\
$\boldsymbol{\beta}$
& Factor loadings in the three-factor log-linear example of Section~4. \\

$\bsym{\Sigma}$
& Annualized covariance matrix of the common factor returns. \\

$\sigma_{\mathrm{Cred}|s}^2$
& Annualized conditional variance of the unstressed credit factor given the stressed factors. \\

$m_z$
& One-day conditional mean of the log return under scenario $z$. \\

$v_{\mathrm{Cred}}$
& One-day conditional variance of the log return under scenario $z$. \\

$\sigma_\epsilon,\widehat{\sigma}_\epsilon$
& True and misspecified annualized idiosyncratic volatilities in the running example. \\

$\sigma_{\rm tot}$
& Total annualized volatility in the running example. \\

$\delta$
& One-trading-day time increment in years, set to $1/252$. \\

$u_j$ & Number of units held in security $j$.\\
$u^{\mathrm{skew}}, u^{\mathrm{hedge}}$ & Common option position size and underlying-index hedge position in the skew-steepener portfolio of Section~\ref{sec:stress_steepener}.\\
$m_j,\tau_j$ & Option moneyness and time to maturity of security $j$.\\
 \\
\multicolumn{2}{l}{\textbf{Dependence and theory}}\\
$\beta(k)$
& $\beta$-mixing coefficient of the feature--residual process
$\{(\bW_t,Y_t)\}$.
\end{longtable}

\subsection{A More Complex Example of ACSA Dynamics}\label{subapx:complex_example_ACSA}
Following on from Examples \ref{eg:Ideal1} and \ref{eg:Ideal2}, we now consider a more complex ACSA example in which ${G_t}$ follows an autoregressive process. The example compares conditional and stationary quantile predictors at a common adjusted level. As in Examples~\ref{eg:Ideal1} and \ref{eg:Ideal2}, we do not condition on any scenario argument in this example.

\begin{example} \label{apx_eg:Ideal3}
\sffamily\upshape
We assume the dynamics of $G_t$ satisfy $G_{t+1} = \mu + \theta G_t + \epsilon_{t+1}$, where $\epsilon_{t+1} \sim \mathcal{N}(0,\sigma^2)$ is an exogenous noise random variable. Then $G_{t+1} \mid G_t \sim \mathcal{N}(\mu + \theta G_t,\sigma^2)$ and the true $\tilde{\alpha}$-quantile of $G_{t+1} \mid G_t$ (denoted by $Q_{G_{t+1}}(\tilde{\alpha}\mid G_t)$) is $Q_{G_{t+1}}(\tilde{\alpha}\mid G_t) = \mu + \theta G_t + \sigma \Phi^{-1}(\tilde{\alpha})$, where $ \Phi^{-1}(\tilde{\alpha})$ is the  $\tilde{\alpha}$-quantile of the standard normal distribution. If we take $g_t(\tilde{\alpha}) \equiv Q_{G_{t+1}}(\tilde{\alpha}\mid G_t)$, then the prediction interval produced in (\ref{eq:OSAinterval}) is $\Ct_t \equiv [\mu + \theta G_t + \sigma \Phi^{-1}(\tilde{\alpha}_t/2), \, \mu + \theta G_t + \sigma \Phi^{-1}(1-\tilde{\alpha}_t/2)].$ This prediction interval depends only on $G_t$ and has width $\mathrm{WI}(\Ct_t)$ given by
\begin{equation} \label{eq:WidthTrue}
\mathrm{WI}(\Ct_t) = \sigma \cdot(\Phi^{-1}(1-\tilde{\alpha}_t/2) - \Phi^{-1}(\tilde{\alpha}_t/2) ).
\end{equation}
As $g_t$ is correctly calibrated, we know from Example \ref{eg:Ideal1} how the $\tilde{\alpha}_t$'s will behave for $0 < \gamma < 2$ and so this tells us (loosely speaking) how $\Ct_t$ and $\mathrm{WI}(\Ct_t)$ behave.

Now consider an alternative quantile predictor. Assuming $-1 < \theta < 1$, then it is easy to see that the stochastic process $\{G_t\}_{t\geq 0}$ has an $\mathcal{N}(\mu/(1-\theta), \, \sigma^2/(1-\theta^2))$ stationary distribution and its $\tilde{\alpha}$-quantile is $Q_{G_\infty}(\tilde{\alpha}) = \mu/(1-\theta) + \sigma / \sqrt{1-\theta^2}\Phi^{-1}(\tilde{\alpha})$.
Suppose we take $g_t(\tilde{\alpha})\equiv Q_{G_\infty}(\tilde{\alpha})$.
In this case the prediction interval produced in (\ref{eq:OSAinterval}) is $\Ct_t \equiv \left[\frac{\mu}{1-\theta} + \frac{\sigma}{\sqrt{1-\theta^2}} \Phi^{-1}(\tilde{\alpha}_t/2), \, \frac{\mu}{1-\theta} + \frac{\sigma}{\sqrt{1-\theta^2}} \Phi^{-1}(1-\tilde{\alpha}_t/2)\right].$ This prediction interval does not depend on $G_t$ and has width
\begin{equation} \label{eq:WidthStat}
\mathrm{WI}(\Ct_t) = \frac{\sigma}{\sqrt{1-\theta^2}}\cdot (\Phi^{-1}(1-\tilde{\alpha}_t/2) - \Phi^{-1}(\tilde{\alpha}_t/2)).
\end{equation}
At a common adjusted level $\widetilde{\alpha}\in(0,1)$, the stationary-quantile
interval has width equal to the conditional-quantile interval width multiplied by
$1/\sqrt{1-\theta^2}$. Thus, whenever $\theta\neq 0$, conditioning on $G_t$
produces a strictly sharper interval at the same quantile level.

The two ACSA procedures generally generate different adaptive level sequences,
however, and the stationary quantiles need not be calibrated conditionally on
$G_t$. Therefore, the common-level comparison above does not imply an ordering
of the realized interval widths along the two adaptive trajectories. Rather, it
isolates the potential gain in sharpness from using a predictor that conditions
on the information relevant to the next-period distribution.
\end{example}\medskip

\subsection{Additional Details for Section \ref{sec:ACSAMotivating_rewrite}}\label{subapx:addi_details_for_ACSA_motivating}

This appendix (i) derives closed form expressions for the intervals $\mathcal{C}_{\text{True}}(\bz)$, $\mathcal{C}_{\text{Mis}}(\bz,\allowbreak \widehat{\sigma}_\epsilon)$ and $\mathcal{C}_{\text{ACSA}}(\bz,\allowbreak \widehat{\sigma}_\epsilon,\allowbreak \tilde{\alpha}_t)$ of Section \ref{sec:ACSAMotivating_rewrite} and (ii) provides further explanations for the behavior we observed in  Table \ref{tab:ConvexIntervals_wACSA_wKSA_Sec5_2} and Figure \ref{fig:beta_t_series}. We first recall that all model parameters are known except for $\sigma_\epsilon$, for which only an estimator $\widehat{\sigma}_\epsilon$ is available. As a result, $m_{\bz}$, $v_{\text{\tiny Cred}}$, and $\sigma_{\text{tot}}$ and their estimated counterparts are
\begin{equation}\label{eq:group_of_quantities_related_to_hatsigma}
\begin{aligned}
\widehat{m}_{\bz} &:=
\delta\,  \left(\mu_\delta-\frac{\widehat{\sigma}_\epsilon^2}{2}\right)
+\beta_{\text{\tiny Oil}}\,z_{\text{\tiny Oil}}
+\beta_{\text{\tiny Rate}}\,z_{\text{\tiny Rate}}
+\beta_{\text{\tiny Cred}}\,\mu_{\text{\tiny Cred}\mid s}\\
\widehat{v}_{\text{\tiny Cred}} &:= \delta\cdot\left(\beta_{\text{\tiny Cred}}^{2} \sigma_{\text{\tiny Cred}|s}^2
+\widehat{\sigma}_\epsilon^{2}\right),\quad \widehat{\sigma}_{\text{tot}}:= \sqrt{\bsym{\beta}^{\mathsf T}\bsym{\Sigma}\boldsymbol{\beta}+ \widehat{\sigma}_{\epsilon}^2}.
\end{aligned}
\end{equation}

\paragraph{Explicit Forms of Prediction Intervals}
In order to provide the explicit forms of the prediction intervals $\mathcal{C}_{\text{True}}(\bz)$, $\mathcal{C}_{\text{Mis}}(\bz, \widehat{\sigma}_\epsilon)$ and $\mathcal{C}_{\text{ACSA}}(\bz, \widehat{\sigma}_\epsilon, \tilde{\alpha}_t)$, recall that $V_0$ and $V_1$ are the $t=0$ and $t=1$ portfolio values, respectively, so that the portfolio gain at $t=1$ is $G_1\defeq V_1 - V_0$. Then recalling that portfolio consists of a single short call position, we have $V_0 = -C(S_0, K, \tau, \sigma_{\text{tot}})$ and $V_1 =  -  C(S_{1}, K, \tau-\delta, \sigma_{\text{tot}})$
where $S_0$ and $S_1$ denote the $t=0$ and $t=1$ prices, respectively, of the S\&P 500, $\delta$ represents 1 day (in years), and we use $C(S, K, \tau, \sigma)$ to denote\footnote{We omit the dependence of the Black-Scholes price on the interest rate and dividend yield.}
the Black-Scholes price of a call option with underlying price $S$, time to maturity $\tau$ and strike $K$.
We write $G_1$ as $G_1(S_1)$ since the randomness of $G_1$ arises entirely from the randomness of $S_1$. The explicit dependency of $G_1$ on $S_1$ is given by:
\begin{equation}\label{eq:explicit_expression_of_G1}
    G_1(S_1) = -V_0 - C(S_{1}, K, \tau-\delta, \sigma_{\text{tot}}).
\end{equation}
Based on \eqref{eq:logret_cond_dist}, for any given scenario $\bz$, the conditional distribution $G_1\mid \bz$ is determined by \eqref{eq:explicit_expression_of_G1} together with $S_1 = S_0\cdot \exp (m_{\bz} + \sqrt{v_{\text{\tiny Cred}}}\cdot \mathcal{Z})$, where $\mathcal{Z}$ is a standard normal random variable. To derive the expression for $\mathcal{C}_{\text{True}}(\bz)$, recall that its endpoints correspond to the $\alpha/2$ and $1-\alpha/2$ conditional quantiles of $G_1$. Since $G_1(S_1)$ is strictly decreasing in $S_1$, for any fixed $\bz$ the conditional quantiles of $G_1$ are obtained by evaluating $G_1(\cdot)$ at the corresponding conditional quantiles of $S_1$. That is,
$Q_{G_1}(\tilde{\alpha}\mid \bz) = G_1(Q_{S_1}(1-\tilde{\alpha}\mid \bz)) = G_1(S_0\cdot \exp (m_{\bz} + \sqrt{v_{\text{\tiny Cred}}}\cdot \Phi^{-1}(1-\tilde{\alpha})))$ where we use $Q_Y(\tilde{\alpha}\mid \bz)$ to denote the $\tilde{\alpha}$ quantile of the random variable $Y$ conditional on $\bz$. The explicit form of $\mathcal{C}_{\text{True}}(\bz)$ is therefore
{\small
\begin{equation}\label{eq:C_correct}
    \mathcal{C}_{\text{True}}(\bz) = \left[G_1(S_0\cdot \exp (m_{\bz} + \sqrt{v_{\text{\tiny Cred}}}\cdot \Phi^{-1}(1-\alpha/2))), G_1(S_0\cdot \exp(m_{\bz} + \sqrt{v_{\text{\tiny Cred}}}\cdot \Phi^{-1}(\alpha/2)))\right].
\end{equation}}

The expressions of $\mathcal{C}_{\text{Mis}}(\bz, \widehat{\sigma}_\epsilon)$ and $\mathcal{C}_{\text{ACSA}}(\bz, \widehat{\sigma}_\epsilon, \tilde{\alpha}_t)$ can be derived similarly. In particular, let $\widehat{G}_1(S_1) = -V_0 - C(S_{1}, K, \tau-\delta, \widehat{\sigma}_{\text{tot}})$ denote the
\textit{estimated} portfolio gain as a function of $S_1$ but with $\widehat{\sigma}_\epsilon$ used in place of $\sigma_\epsilon$, and  $\widehat{m}_{\bz}$, $\widehat{v}_{\text{\tiny Cred}}$ and $\widehat{\sigma}_{\text{tot}}$ determined by $\widehat{\sigma}_\epsilon$ (as well as $\bz$ in the case of $\widehat{m}_{\bz}$); see \eqref{eq:group_of_quantities_related_to_hatsigma}.
For $0<\tilde\alpha_t<1$, the explicit forms of $\mathcal{C}_{\text{Mis}}(\bz, \widehat{\sigma}_\epsilon)$ and $\mathcal{C}_{\text{ACSA}}(\bz, \widehat{\sigma}_\epsilon, \tilde{\alpha}_t)$ are then given by
{\small
\begin{eqnarray*}\nonumber
\mathcal{C}_{\text{Mis}}(\bz,\widehat{\sigma}_\epsilon) &=& \left[\widehat{G}_1(S_0\cdot \exp(\widehat{m}_{\bz} + \sqrt{\widehat{v}_{\text{\tiny Cred}}}\cdot \Phi^{-1}(1-\alpha/2))), \widehat{G}_1(S_0\cdot \exp(\widehat{m}_{\bz} + \sqrt{\widehat{v}_{\text{\tiny Cred}}}\cdot \Phi^{-1}(\alpha/2)))\right] \\
\mathcal{C}_{\text{ACSA}}(\bz, \widehat{\sigma}_\epsilon, \tilde{\alpha}_t) &=& \left[\widehat{G}_1(S_0\cdot \exp(\widehat{m}_{\bz} + \sqrt{\widehat{v}_{\text{\tiny Cred}}}\cdot \Phi^{-1}(1-\tilde{\alpha}_t/2))), \widehat{G}_1(S_0\cdot \exp(\widehat{m}_{\bz} + \sqrt{\widehat{v}_{\text{\tiny Cred}}}\cdot \Phi^{-1}(\tilde{\alpha}_t/2)))\right] \label{eq:C_acsa}
\end{eqnarray*}}where $\tilde{\alpha}_t$ is the time $t$ quantile argument of the ACSA algorithm. We remark that when $\tilde{\alpha}_t = \alpha$, we have exactly $\mathcal{C}_{\text{ACSA}}(\bz, \widehat{\sigma}_\epsilon, \tilde{\alpha}_t) = \mathcal{C}_{\text{Mis}}(\bz,\widehat{\sigma}_\epsilon)$.

\begin{comment}
\paragraph{From Time Step $0$ to Step $t$}
The expressions in \eqref{eq:C_correct}, \eqref{eq:C_incorrect}, and \eqref{eq:C_acsa} are all derived for the portfolio gain $G_1$ between time steps $0$ and $1$. One might think that separate expressions must be explicitly derived for $G_{t+1}$ at each subsequent time step $t$. However, these prediction intervals can in fact be obtained directly by scaling the corresponding intervals for $G_1$. The key observation is that the option’s moneyness is held constant across time steps, implying that its value is always proportional to the underlying price by the same coefficient, i.e., $V_t/S_t = V_0/S_0$. Writing $G_{t+1}$ as a function of $S_{t+1}$, we have
$$
\begin{aligned}
    G_{t+1}(S_{t+1}) &= -V_t - C(S_t, K_t, \tau-\delta, \sigma_\text{tot}) \\
    &= -\dfrac{V_t}{S_t}\cdot S_t - S_t\cdot \widetilde{C}(m, \tau-\delta, \sigma_\text{tot})\\
    &= -\dfrac{V_0}{S_0} \cdot S_0\cdot \dfrac{S_t}{S_0} - \dfrac{S_t}{S_0} \cdot S_0 \cdot \widetilde{C}(m, \tau-\delta, \sigma_\text{tot})\\
    &= \dfrac{S_t}{S_0}\cdot G_1(S_1),
\end{aligned}
$$
where the second equality follows from \eqref{eq:CallPrice}. Consequently, the prediction intervals $\mathcal{C}_{\text{True}}$, $\mathcal{C}_{\text{Mis}}$, and $\mathcal{C}_{\text{ACSA}}$ at time $t$ are obtained simply by scaling those at time $0$ by the factor $S_t/S_0$.
\end{comment}

\paragraph{True Coverage Rates of Various Prediction Intervals}
For any stress scenario $\bz$, we derive an explicit expression for the true coverage rate of the true scenario gain $G_1$ associated with a given prediction interval $[A, B]$. As is clear from \eqref{eq:explicit_expression_of_G1}, the distribution of $G_1$ is determined by the normal distribution of $\log S_1/S_0$. Consequently, by transforming the prediction interval for $G_1$ into an ``equivalent'' prediction interval for $\log S_1/S_0$, we can obtain a closed-form expression for its true coverage rate. In particular, for any interval $[A, B]$, we have
{\small
\begin{equation}
\begin{aligned}
\mathbb{P}(G_1 \in [A, B])
&= \mathbb{P}\Bigg(\log S_1/S_0 \in \left[\log G_1^{-1}(B) - \log S_0, \quad \log G_1^{-1}(A) - \log S_0\right]\Bigg) \\
&= \mathbb{P}\Bigg(\mathcal{Z} \in \left[\dfrac{1}{\sqrt{v_{\text{\tiny Cred}}}}\left(\log G_1^{-1}(B) - \log S_0 - m_{\bz}\right), \quad \dfrac{1}{\sqrt{v_{\text{\tiny Cred}}}}\left(\log G_1^{-1}(A) - \log S_0 - m_{\bz}\right)\right]\Bigg) \\
&= \Phi\Bigg(\dfrac{1}{\sqrt{v_{\text{\tiny Cred}}}}\left(\log G_1^{-1}(A) - \log S_0 - m_{\bz}\right)\Bigg) - \Phi\Bigg(\dfrac{1}{\sqrt{v_{\text{\tiny Cred}}}}\left(\log G_1^{-1}(B) - \log S_0 - m_{\bz}\right)\Bigg) \label{apxprop:equiv_transform}
\end{aligned}
\end{equation}}where $\Phi(\cdot)$ denotes the standard normal CDF. For $0<\tilde\alpha<1$, it follows from \eqref{apxprop:equiv_transform} and our explicit expression for $\mathcal{C}_{\text{ACSA}}(\bz, \widehat{\sigma}_\epsilon, \tilde{\alpha})$  above that
\begin{equation}\label{eq:def_of_psi_in_appdix}
\begin{aligned}
    \psi(m_{\bz}, \widehat{\sigma}_\epsilon, \tilde{\alpha}) &:= \mathbb{P}(G_1 \in \mathcal{C}_\text{ACSA}(\bz, \widehat{\sigma}_\epsilon, \tilde{\alpha})) \\
    &= \Phi\Bigg(\dfrac{1}{\sqrt{v_{\text{\tiny Cred}}}}\left(\log G_1^{-1}(\widehat{G}_1(S_0\cdot \exp(\widehat{m}_{\bz} + \sqrt{\widehat{v}_{\text{\tiny Cred}}}\cdot \Phi^{-1}(1-\tilde{\alpha}/2)))) - \log S_0 - m_{\bz}\right)\Bigg) \\
    &\phantom{=} - \Phi\Bigg(\dfrac{1}{\sqrt{v_{\text{\tiny Cred}}}}\left(\log G_1^{-1}(\widehat{G}_1(S_0\cdot \exp(\widehat{m}_{\bz} + \sqrt{\widehat{v}_{\text{\tiny Cred}}}\cdot \Phi^{-1}(\tilde{\alpha}/2)))) - \log S_0 - m_{\bz}\right)\Bigg)
\end{aligned}
\end{equation}
Outside the interior, define $\psi(m_{\bz},\widehat\sigma_\epsilon,\tilde\alpha)=1$ for $\tilde\alpha\leq0$ and $0$ for $\tilde\alpha\geq1$. The expression above is the probability of the scenario gain $G_1$ being covered by the ACSA interval $\mathcal{C}_\text{ACSA}(\bz, \widehat{\sigma}_\epsilon, \tilde{\alpha})$ as a function of $m_{\bz}$, $\widehat{\sigma}_\epsilon$ and $\tilde{\alpha}$.
Although the function $\psi$ provides the true coverage rate for any stress scenario $\bz$, it depends on $\bz$ only through $m_{\bz}$. In particular, $\widehat{m}_{\bz}$ is determined by $m_{\bz}$ through $\widehat{m}_{\bz} = m_{\bz} + \frac{\delta}{2} (\sigma_\epsilon^2 - \widehat{\sigma}_\epsilon^2)$, while $v_{\text{\tiny Cred}}$ and $\widehat{v}_{\text{\tiny Cred}}$ do not depend on the explicit value of $\bz$.
We also note that \eqref{eq:def_of_psi_in_appdix} may be used to compute
true coverage rate of $\mathcal{C}_\text{Mis}$ since $\mathbb{P}(G_1 \in \mathcal{C}_\text{Mis}(\bz, \widehat{\sigma}_\epsilon)) = \psi(m_{\bz}, \widehat{\sigma}_\epsilon, \alpha)$. We can now use $\psi(m_{\bz}, \widehat{\sigma}_\epsilon, \tilde{\alpha})$ to explain the monotonic pattern of coverage rates that we observed in Table \ref{tab:ConvexIntervals_wACSA_wKSA_Sec5_2} as well as the non-monotonic (in $\widehat{\sigma}_\epsilon$) behavior of $\tilde{\alpha}_t$ we observed in Figure \ref{fig:beta_t_series}.

\paragraph{Monotonic Pattern of Table \ref{tab:ConvexIntervals_wACSA_wKSA_Sec5_2}}
In Table \ref{tab:ConvexIntervals_wACSA_wKSA_Sec5_2}, we observed that for both $\mathcal{C}_{\text{Mis}}(\bz,\allowbreak \widehat{\sigma}_\epsilon)$ and $\mathcal{C}_{\text{ACSA}}(\bz,\allowbreak \widehat{\sigma}_\epsilon,\allowbreak\bar{\tilde{\alpha}})$, the true coverage rates decreased from left to right within each row and increased from top to bottom within each column. Since both $\widehat{\sigma}_\epsilon$ and $\tilde{\alpha}$ are constant across scenarios $\bz$, $m_{\bz}$ is the only argument of $\psi$ that drives differences in the real coverage rate across scenarios.
As demonstrated in Figure \ref{fig:mono_in_ms},  $\psi(m_{\bz}, \widehat{\sigma}_\epsilon, \tilde{\alpha})$ consistently increases with $m_{\bz}$ under all combinations of $\widehat{\sigma}_\epsilon$ and $\tilde{\alpha}$ considered in the experiments of Section \ref{sec:ACSAMotivating_rewrite}. As we recall from \eqref{eq:explicit_form_of_ms}, however, $m_{\bz}$ decreases with $z_{\text{\tiny Oil}}$ and increases with $z_{\text{\tiny Rate}}$ and together these observations therefore explain the monotonicity patterns of Table \ref{tab:ConvexIntervals_wACSA_wKSA_Sec5_2}.

\begin{figure}[htbp]
    \centering
    \begin{minipage}[b]{0.48\textwidth}
        \centering
        \includegraphics[width=\textwidth]{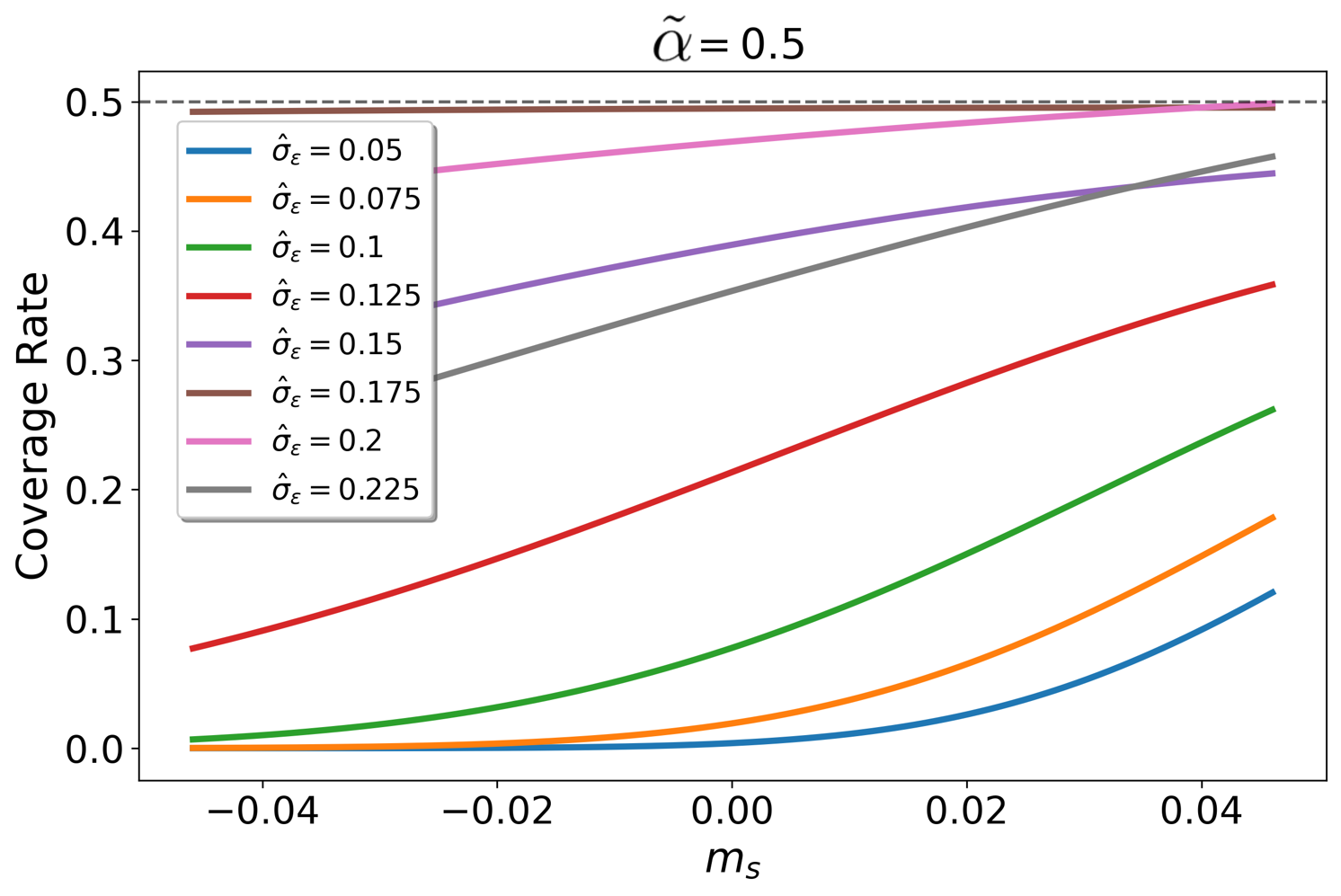}
    \end{minipage}
    \hfill
    \begin{minipage}[b]{0.48\textwidth}
        \centering
        \includegraphics[width=\textwidth]{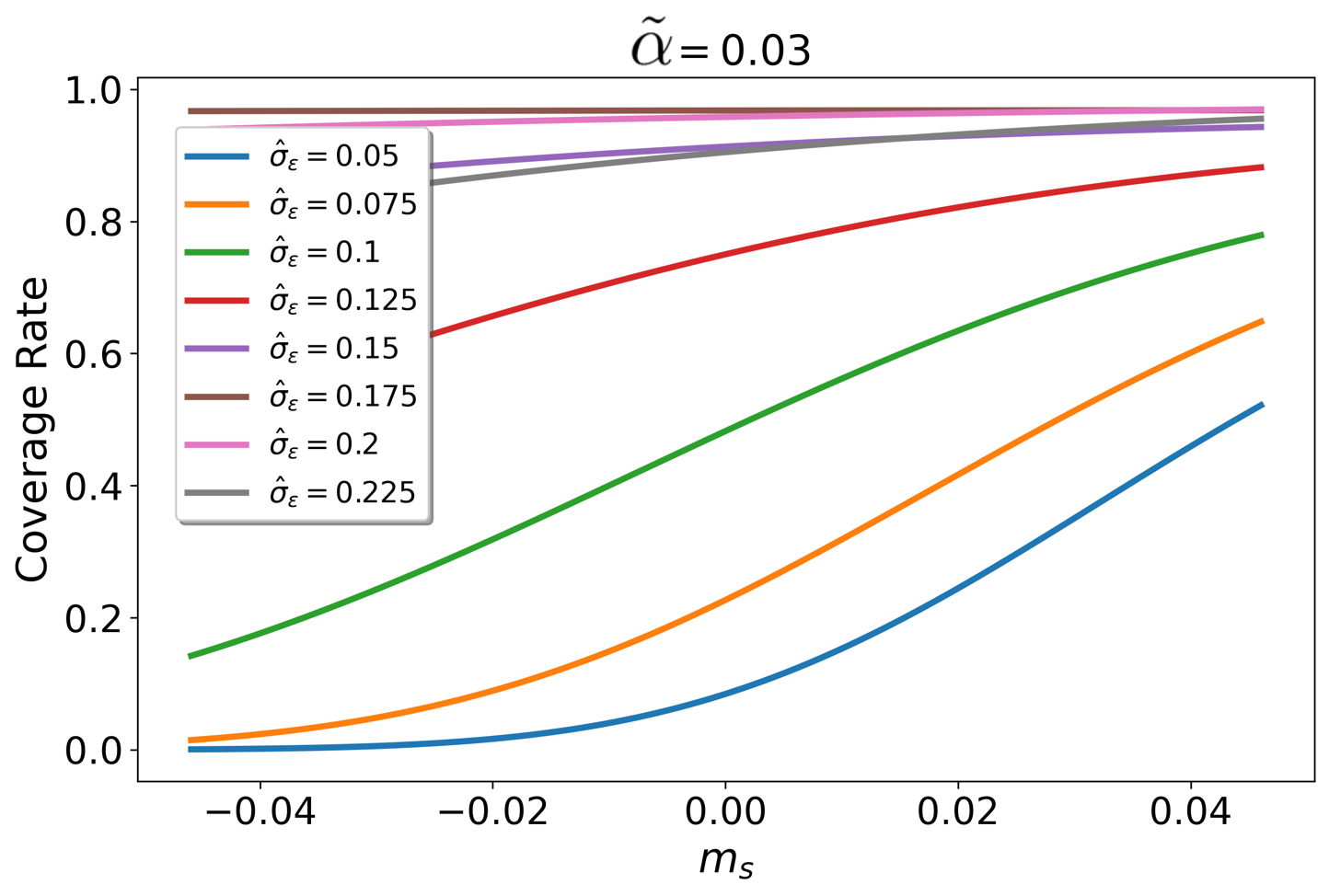}
    \end{minipage}
    \caption{The variation of $\psi(m_{\bz}, \widehat{\sigma}_\epsilon, \tilde{\alpha})$ with respect to $m_{\bz}$. The left figure sets $\tilde{\alpha}=0.5$ and the right figure sets $\tilde{\alpha} = 0.03$. In each figure, different values of $\widehat{\sigma}_\epsilon$ are considered, ranging from $0.05$ to $0.225$. In each of these curves, $\psi(m_{\bz}, \widehat{\sigma}_\epsilon, \tilde{\alpha})$ monotonically increases with $m_{\bz}$.}
    \label{fig:mono_in_ms}
\end{figure}

\paragraph{Non-monotonicity Pattern of Figure \ref{fig:beta_t_series}}
In order to understand the non-monotonicity pattern of Figure \ref{fig:beta_t_series}, we need to consider $\psi(m_{\bz}, \widehat{\sigma}_\epsilon, \tilde{\alpha})$ as a function of $\widehat{\sigma}_\epsilon$. As demonstrated in Figure~\ref{fig:non_mono_in_sigma},  the function $\psi(m_{\bz},\widehat{\sigma}_\epsilon,\tilde{\alpha})$ is non-monotonic in $\widehat{\sigma}_\epsilon$ for all combinations of $m_{\bz}$ and $\tilde{\alpha}$ that we consider. Indeed, it first increases and then decreases as $\widehat{\sigma}_\epsilon$ moves from $0$ to $0.25$\footnote{Specifically, $\psi$ increases for smaller values of $\widehat{\sigma}_\epsilon$ (e.g., $\widehat{\sigma}_\epsilon\in(0,0.15]$) and decreases for larger values (e.g., $\widehat{\sigma}_\epsilon\in[0.20,0.25]$). At the correctly specified value $\widehat{\sigma}_\epsilon=\sigma_\epsilon=0.18$, the coverage equals $1-\tilde{\alpha}$: this is $50\%$ in the left panel ($\tilde{\alpha}=0.5$) and $97\%$ in the right panel ($\tilde{\alpha}=0.03$).}.
But $\psi(m_{\bz}, \widehat{\sigma}_\epsilon, \tilde{\alpha})$ decreases in $\tilde{\alpha}$ since $\mathcal{C}_\text{ACSA}(\bz, \widehat{\sigma}_\epsilon, \tilde{\alpha})$ is the region between the estimated $\tilde{\alpha}/2$-  and $(1-\tilde{\alpha}/2)$- quantiles. Hence, increasing $\tilde{\alpha}$ narrows the prediction interval and thus decreases $\psi(m_{\bz}, \widehat{\sigma}_\epsilon, \tilde{\alpha})$.
Finally, recall that the ACSA algorithm selects $\tilde{\alpha}_t$ to target a $50\%$ empirical coverage rate. Consequently, when $\widehat{\sigma}_\epsilon$ becomes either too small or too large relative to the true value $\sigma_\epsilon = .18$, $\psi(m_{\bz},\widehat{\sigma}_\epsilon,\alpha)$ decreases and achieves a coverage rate that is less than the target rate of 50\% in our experiments. Thus, a smaller $\tilde{\alpha}_t$ is required to bring $\psi(m_{\bz},\widehat{\sigma}_\epsilon,\bar{\tilde{\alpha}})$ back towards $50\%$. This explains the non-monotonic behavior in Figure \ref{fig:beta_t_series} where $\bar{\tilde{\alpha}}$ (the long-term average of $\tilde{\alpha}_t$) decreases as $\widehat{\sigma}_\epsilon$ moves away from its true value $\sigma_\epsilon$.

\begin{figure}[htbp]
    \centering
    \begin{minipage}[b]{0.48\textwidth}
        \centering
    \includegraphics[width=\textwidth]{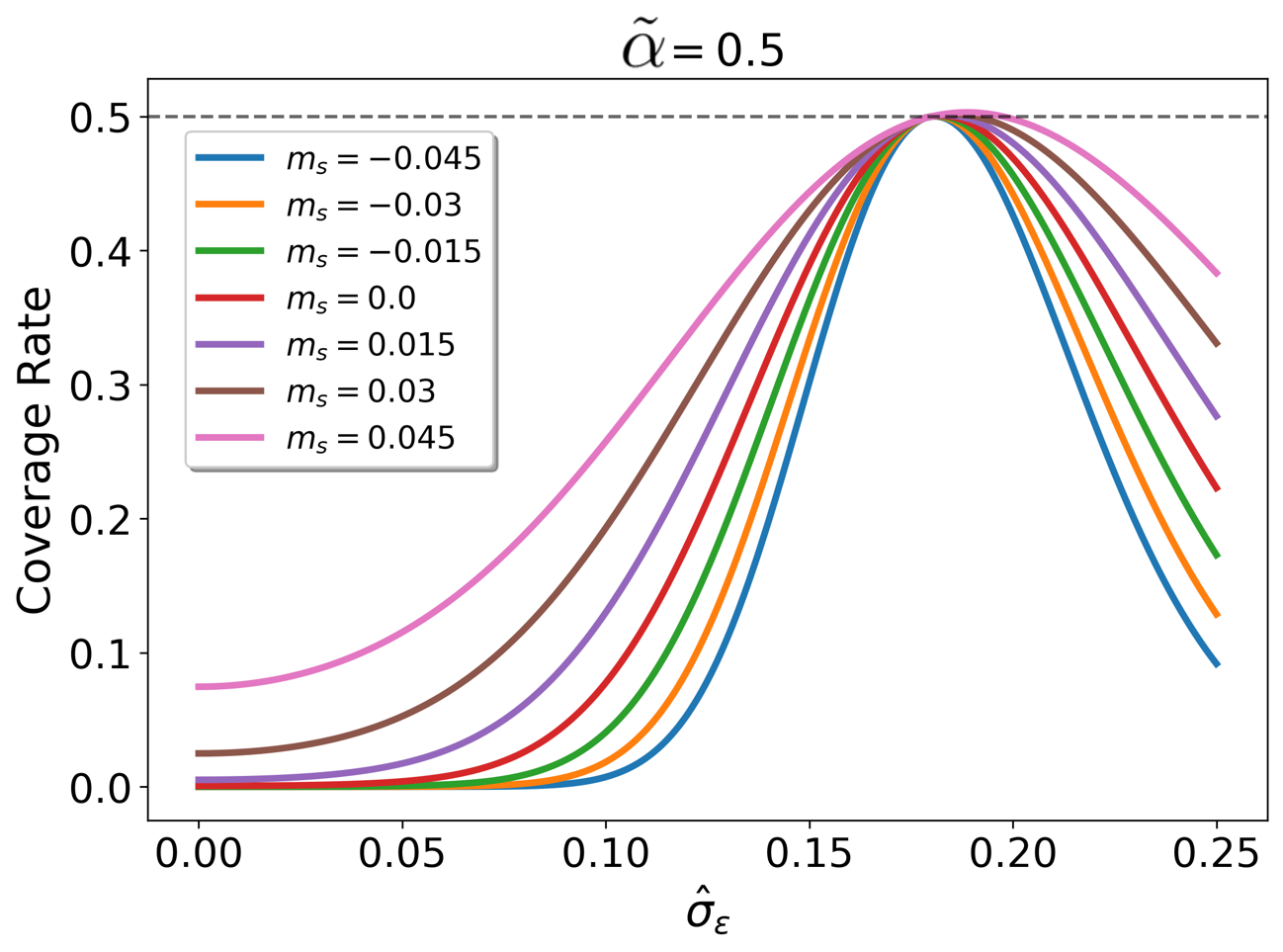}
    \end{minipage}
    \hfill
    \begin{minipage}[b]{0.48\textwidth}
        \centering
    \includegraphics[width=\textwidth]{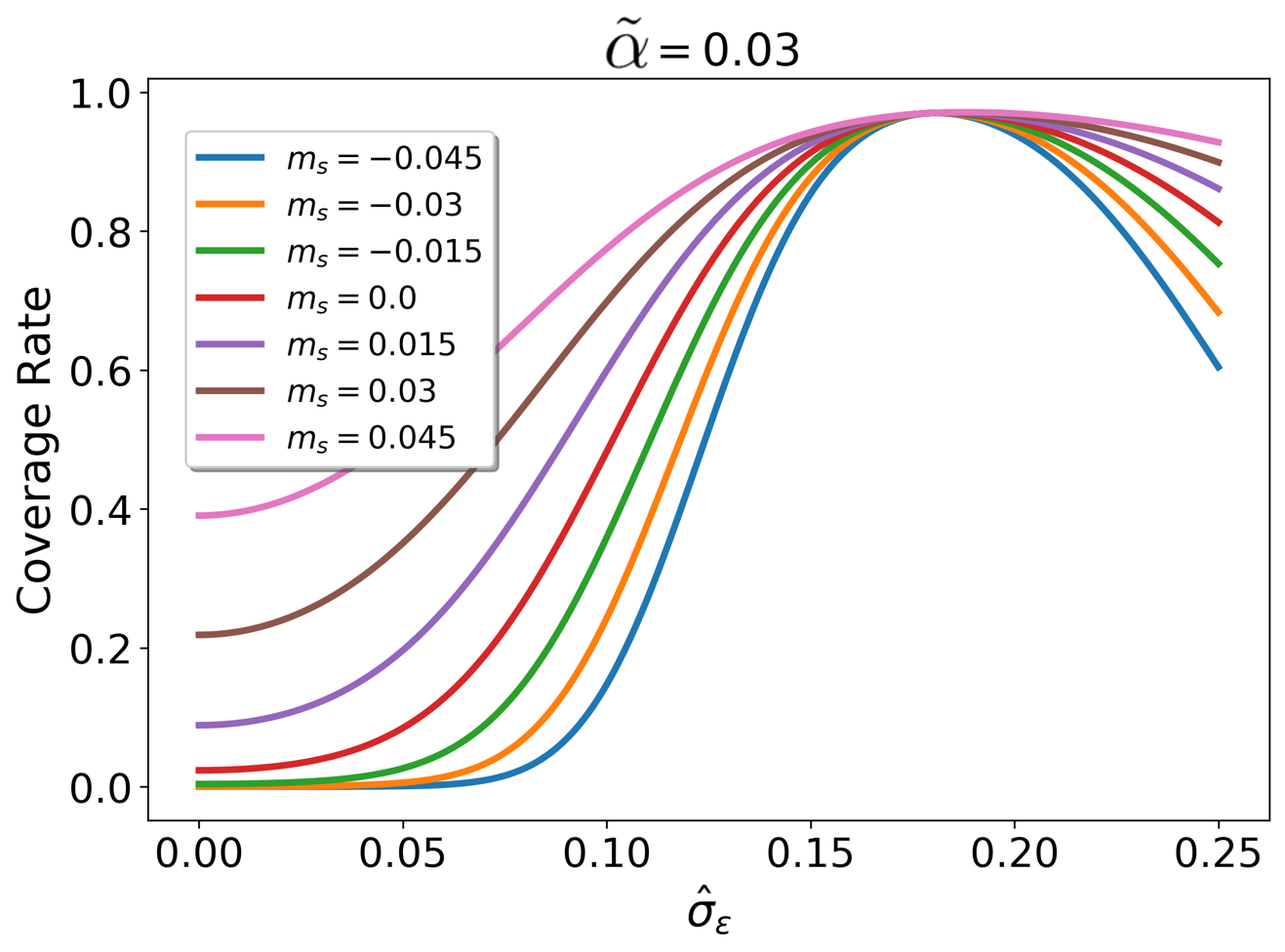}
    \end{minipage}
    \caption{The variation of $\psi(m_{\bz}, \widehat{\sigma}_\epsilon, \tilde{\alpha})$ with respect to $\widehat{\sigma}_\epsilon$. The left panel sets $\tilde{\alpha} = 0.5$ and the right panel sets $\tilde{\alpha} = 0.03$. In each sub-figure, different values of $m_{\bz}$ are considered, ranging from $-0.045$ to $0.045$. Across all curves, $\psi(m_{\bz}, \widehat{\sigma}_\epsilon, \tilde{\alpha})$ exhibits a consistent pattern with respect to $\widehat{\sigma}_\epsilon$: it first increases and then decreases as $\widehat{\sigma}_\epsilon$ passes through some value near the true value $\sigma_\epsilon = 0.18$.
}
    \label{fig:non_mono_in_sigma}
\end{figure}

\begin{comment}
\begin{figure}
    \centering    \includegraphics[width=\linewidth]{fig/sectionA3_psi_surface_enhanced.png}
    \caption{The Landscape of the coverage rate function $\psi(m_{\bz}, \widehat{\sigma}_\epsilon)$. }
    \label{fig:landscape_of_psi_cvg}
\end{figure}
\end{comment}

\begin{comment}

\begin{figure}[htbp]
    \centering
    \begin{minipage}[b]{0.48\textwidth}
        \centering
        \includegraphics[width=\textwidth]{fig/sectionA3_psi_vs_sigmahat_msSlices.png}
    \end{minipage}
    \hfill
    \begin{minipage}[b]{0.48\textwidth}
        \centering
        \includegraphics[width=\textwidth]{fig/sectionA3_psi_vs_ms_sigmaSlices.png}
    \end{minipage}
    \caption{Cross sections of $\psi(m_{\bz}, \widehat{\sigma}_\epsilon)$. The left panel shows $\psi(m_{\bz},\cdot)$ for varying $m_{\bz}$, where each curve first increases and then decreases with respect to $\widehat{\sigma}_\epsilon$. The right panel shows $\psi(\cdot,\widehat{\sigma}_\epsilon)$ for varying $\widehat{\sigma}_\epsilon$, where each curve increases monotonically with $m_{\bz}$.}
    \label{fig:cross_sections}
\end{figure}
\end{comment}

\subsection{The Group-Balanced ACSA Algorithm}\label{subapx:alg_for_gbacsa}

As discussed at the end of Section~\ref{sec:ACSA_Alg}, the ACSA algorithm provides only a marginal coverage guarantee. If we analyze the prediction intervals under each \textit{fixed scenario}, then the coverage rate is often far from the marginal coverage rate. This limitation can be observed in Table~\ref{tab:ConvexIntervals_wACSA_wKSA_Sec5_2}, where the true coverage rates can deviate substantially from the nominal level of $50\%$ in more extreme scenarios. For example, when $(z_{\text{\tiny Oil}}, z_{\text{\tiny Rate}}) = (12\%, -20 \ \text{bp})$ and $(-12\%, 20 \ \text{bp})$, the true coverage rates are $25\%$ and $69\%$, respectively. One of the primary causes of this discrepancy is that ACSA updates its quantile argument $\tilde{\alpha}_t$ iteratively based on \textit{realized} observations. As a result, rare or extreme scenarios are underrepresented\footnote{Loosely speaking, the marginal coverage rate (which we control via the ACSA algorithm) can be viewed as a weighted average of the coverage rates within each scenario with weights given by the probability of each scenario.} in the updates due to their disproportionately low frequency in the observed data. One approach for tackling this problem is to incorporate a group-balancing\footnote{See Section 4.1 of \citet{angelopoulos2021gentle} for detailed discussions on the group-balanced conformal prediction technique.} adjustment into the vanilla ACSA procedure.

The idea behind group-balancing is to partition the space of stress scenarios into several groups, which we denote by $B_1, \ldots, B_m$, and to run ACSA separately within each group. Each group maintains its own quantile parameter, denoted by $\tilde{\alpha}_{t,1}, \ldots, \tilde{\alpha}_{t,m}$. For a concrete example, suppose we want to balance between the extreme and non-extreme scenarios in Table \ref{tab:ConvexIntervals_wACSA_wKSA_Sec5_2}. Then we may partition the stress scenarios into two groups: the ``extreme scenario group'' $B_1$ and the ``non-extreme scenario group'' $B_2$, with
$$
B_1 := \{\bz\ : \ |z_{\text{\tiny Oil}}| > 6\% \text{ or } |z_{\text{\tiny Rate}}| > 10 \ \text{bp}\}, \quad B_2 := \{\bz\ : \ |z_{\text{\tiny Oil}}| \leq 6\% \text{ and } |z_{\text{\tiny Rate}}| \leq 10 \ \text{bp}\}.
$$
Separate quantiles arguments $\tilde{\alpha}_{t,1}$ and $\tilde{\alpha}_{t,2}$ are maintained for these two groups, since the quantile predictor $g_t$ may exhibit different estimation accuracies across them and thus require different adjustments to the quantile argument $\tilde{\alpha}_{t}$. At each time step, when the realized scenario $\bz_t$ is observed, the group-balanced ACSA first identifies its corresponding group ($B_1$ or $B_2$) and then updates only the quantile parameter of that group. For instance, if $\bz_t\in B_k$ for some $k\in \{1,\ldots, m\}$, then only $\tilde{\alpha}_{t,k}$ is updated according to \eqref{eq:OSA_update}, and the quantile arguments for the remaining groups are left unchanged, i.e.
$$
\tilde{\alpha}_{t, k} \gets \tilde{\alpha}_{t-1, k} + \gamma\cdot(\alpha-\mathrm{err}_{t}),\quad \text{with  }\mathrm{err}_{t} \gets \bbone\{G_{t}\notin  \Ct_{t-1}(\bz_t;\tilde\alpha_{t-1,k})\}
$$
$$
\tilde{\alpha}_{t, k'} \gets \tilde{\alpha}_{t-1, k'},\quad \text{for  } k' = 1,\ldots, m \text{ with } k'\neq k
$$
The group-balanced variant of our ACSA algorithm is summarized in Algorithm~\ref{alg:GB-ACSA}, where each group effectively runs Algorithm~\ref{alg:ACSA} independently. The corresponding coverage guarantee for Group-Balanced ACSA is stated in Theorem~\ref{thm:coverage_GroupACSA}, which follows almost directly from the original coverage result in Theorem~\ref{thm:empirical_coverage}.

\begin{thm}\label{thm:coverage_GroupACSA}
Under Assumption~\ref{assump:g_t_basicAssump}, with $\alpha\in(0,1)$, $\gamma>0$, and the boundary construction \eqref{eq:acsa_prediction_set} applied within each group, for each group $B_k$, the empirical coverage of the prediction interval  $\Ct_t(\cdot)$ generated by group-balanced ACSA (Algorithm \ref{alg:GB-ACSA}) satisfies the coverage guarantee
$$\left|\dfrac{1}{T_k}\sum_{t=1}^T \bbone\{G_{t+1}\in  \Ct_t(\bz_{t+1})\} \cdot \bbone\{\bz_{t+1}\in B_k\}  -(1-\alpha)\right|\leq \dfrac{\max\{\alpha, 1-\alpha\}+\gamma}{T_k\gamma}$$
almost surely, with $T_k \defeq \sum_{t=1}^T \bbone\{\bz_{t+1}\in B_k\}$ and is assumed $>0$.
\end{thm}

\begin{algorithm}[ht!]
    \caption{Group-Balanced ACSA}
    \label{alg:GB-ACSA}
    \begin{algorithmic}[1]
    \Require Partition of stress scenarios into groups $B_1, \ldots, B_m$, plus the inputs of Algorithm \ref{alg:ACSA}
    \State Set $\tilde\alpha_{1,k}=\alpha$ and report $\Ct_1(\bz)=\Ct_1(\bz;\alpha)$ for $\bz\in B_k$, $k=1,\ldots,m$
    \For{$t=2,...,T$}
        \State Observe the realized portfolio gain $G_t$ and stress scenario $\bz_{t}$
        \State Identify the group to which $\bz_t$ belongs, for instance $\bz_t \in B_k$, then compute
        $$
            \mathrm{err}_{t} \gets \bbone\{G_{t}\notin  \Ct_{t-1}(\bz_t;\tilde\alpha_{t-1,k})\}
        $$
        \State Keep $\tilde{\alpha}_{t, k'} \gets \tilde{\alpha}_{t-1, k'}$ for all $k'=1,\ldots, m$ with $k'\neq k$, and update
        $$
            \tilde{\alpha}_{t, k} \gets \tilde{\alpha}_{t-1, k} + \gamma\cdot(\alpha-\mathrm{err}_{t})
        $$

        \State Update / retrain quantile predictor $g_t$ if  so desired
        \State \textbf{Output:} For all $k' = 1,\ldots, m$, and any stress scenario $\bz \in B_{k'}$, produce prediction interval
        $$\Ct_t(\bz)\gets\Ct_t(\bz;\tilde\alpha_{t,k'})
        $$

    \EndFor
    \end{algorithmic}
\end{algorithm}

The group-balanced ACSA helps mitigate the lack of theoretical guarantees in extreme scenarios. However, the method remains limited by data scarcity in the more extreme scenario groups. For such scenario groups, Theorem \ref{thm:coverage_GroupACSA} will provide only limited comfort. Hence, there is a resolution vs sample size tradeoff in how finely you partition scenarios.
That said, the group-balanced ACSA can provide theoretical guarantees in moderate scenario groups where there are sufficiently many realizations for calibrating the corresponding $\tilde{\alpha}_t$ values.  Moreover, as may be seen from Table \ref{tab:ConvexIntervals_wACSA_wKSA_Sec5_2}, even the non-group-balanced ACSA algorithm can yield a substantial improvement over a scenario analysis that relies on an inaccurate model.

\subsection{Additional Details for Section \ref{sec:KSAMotivtaingEG}}
\label{subapx:addi_details_KSA_motive}

We provide additional implementation details for the KSA experiments in Section~\ref{sec:KSAMotivtaingEG}.

\paragraph{Feature Vector and Assumption~\ref{assump:individual_cali_assump}}
One of the key preliminary steps before implementing the KSA algorithm is to specify the historical context $\bH_t$ and the feature vector $\bW_t(\bz)$. Intuitively, $\bH_t$ should contain all historical information relevant to the portfolio gain on day $t+1$. In our setting, the portfolio is short a single call option with fixed time to maturity $\tau$ and moneyness (strike price divided by the underlying price) $m$. We show below that \textit{no historical context is required} once the portfolio gain is normalized by the current underlying price. Thus, in the original factor coordinates, the information content of the feature vector can be represented by $\bW_t(\bz)=\bz$. The portfolio gain at time $t+1$ is given by
\begin{equation}\label{eq:form_of_G_before_normalize}
G_{t+1}
=
-\text{BS}(S_{t+1},m\cdot S_t,\tau-\delta,\sigma_{\text{tot}},1)
+\text{BS}(S_t,m\cdot S_t,\tau,\sigma_{\text{tot}},1),
\end{equation}
where $\text{BS}(S_t,K,\tau,\sigma,\mbox{c})$ denotes the Black--Scholes price of an option, with $S_t$, $K$, $\tau$, and $\sigma$ denoting the underlying price, strike, time to maturity, and implied volatility, respectively. $\mbox{c}\in\{0,1\}$ indicating
the option type, with $\mbox{c}=1$ for a call and $\mbox{c}=0$ for a
put. The risk-free rate and dividend yield are treated as fixed and are therefore omitted from the arguments of $\text{BS}$. Since the portfolio maintains a fixed moneyness $m$, its strike price is $m\cdot S_t$, while the remaining time to maturity decreases from $\tau$ to $\tau-\delta$ between times $t$ and $t+1$. After normalizing by $S_t$, the portfolio gain depends on the past only through the return $S_{t+1}/S_t$. Specifically,
\begin{equation}\label{eq:explicit_expression_G_prime}
G_{t+1}'
\defeq
\frac{G_{t+1}}{S_t}
=
-\frac{S_{t+1}}{S_t}
\widetilde{\mbox{BS}}\left(
m\cdot\frac{S_t}{S_{t+1}},
\tau-\delta,
\sigma_{\text{tot}},
1
\right)
+\widetilde{\mbox{BS}}(m,\tau,\sigma_{\text{tot}},1),
\end{equation}

where $\widetilde{\mbox{BS}}(m,\tau,\sigma,\mbox{c})$ denotes the Black--Scholes price normalized by the underlying price:
\begin{equation}\label{eq:BS_tilde_def}
\widetilde{\mbox{BS}}(m,\tau,\sigma,\mbox{c})
:=
\frac{1}{S_t}\text{BS}(S_t,K,\tau,\sigma,\mbox{c}),
\qquad
m=\frac{K}{S_t}.
\end{equation}
The only source of randomness in \eqref{eq:explicit_expression_G_prime} is therefore the return $S_{t+1}/S_t$. Any prediction interval for $G_{t+1}'$ can be mapped to one for $G_{t+1}$ by multiplying both endpoints by the positive, $\mathcal F_t$-measurable quantity $S_t$, which preserves its conditional coverage probability. Hence, when applying the theory in Section~\ref{sec:ksa-theory}, we regard $G_t'$ as the response and set
$Y_t=G_t'-\varphi(\bW_t)$ where $\varphi$ is the fixed centering function specified below.

We now verify parts~\ref{itm:separate_stationary_assump}--\ref{itm:cond_dist_cont} of Assumption~\ref{assump:individual_cali_assump}; the kernel condition in part~\ref{itm:kernel_regularity} is verified in the next paragraph. Under \eqref{eq:StockReturn1}, the factor vectors $\bF_t$ and idiosyncratic noises are independent across time. Since $\bW_t=\bF_t^{(1:2)}$ and $G_t'$ is a fixed function of the contemporaneous factor vector and idiosyncratic noise, $\{(\bW_t,Y_t)\}_{t=1}^{\infty}$ is i.i.d.\ when $\varphi$ is fixed. It is therefore stationary and geometrically $\beta$-mixing, with $\beta(k)=0$ for every $k\geq1$, which verifies part~\ref{itm:separate_stationary_assump}. Conditional on the stressed factors $\bF_{t+1}^{(1:2)}=\bz$, the remaining credit factor has the Gaussian conditional distribution used in \eqref{eq:logret_cond_dist}, and the idiosyncratic noise remains independent of $\mathcal F_t$. Consequently, the conditional distribution of $G_{t+1}'$, and hence that of $Y_{t+1}$, depends on $\mathcal F_t$ only through the specified scenario $\bz=\bW_{t+1}(\bz)$. This gives \eqref{eq:predictive-sufficiency} and verifies part~\ref{itm:independence_n_homogeneity}. Moreover, $\bW_t=\bF_t^{(1:2)}\sim\mathcal N(\mathbf 0,\delta\Sigma_{ss}),$ where $\Sigma_{ss}$ is positive definite. Its density is continuous and strictly positive, and is therefore bounded above and away from zero on every compact enlarged operational region $\mathcal D^+$. Thus, part~\ref{itm:density_bound} holds. Finally, by \eqref{eq:logret_cond_dist}, the conditional log return is Gaussian with strictly positive variance independent of $\bw$, while its mean varies affinely with $\bw$. The normalized option gain in \eqref{eq:explicit_expression_G_prime} and the oracle centering function defined below are smooth functions of their arguments. It follows that the conditional distribution of $Y_t$ is continuous and that its conditional CDF varies uniformly continuously with $\bw$ over the compact region $\mathcal D^+$. This verifies part~\ref{itm:cond_dist_cont}.

\paragraph{Kernel and Bandwidth Tuning}
For the kernel $\kappa$, we use the weighted RBF kernel
\begin{equation}\label{eq:def_gaussian_kernel}
\kappa_h(\bw,\bw')
\defeq
\exp\left(
-\frac{(\bw-\bw')^\top W(\bw-\bw')}{2h^2}
\right),
\end{equation}
where $h$ is the bandwidth and $W=\delta^{-1}\Sigma_{ss}^{-1}$. Since $\bW_t\sim\mathcal N(\mathbf 0,\delta\Sigma_{ss})$, this choice of $W$ standardizes the two stressed factors. Although \eqref{eq:def_gaussian_kernel} is not radial with respect to the Euclidean distance in the original coordinates unless $W$ is proportional to the identity, it becomes an ordinary radial RBF kernel under the fixed invertible transformation $\overline{\bw}=W^{1/2}\bw$, i.e.,
$\kappa_h(\bw,\bw')
=
\exp\left(
-\frac{\|\overline{\bw}-\overline{\bw}'\|^2}{2h^2}
\right).$ Thus, applying KSA to the original features with the weighted kernel is equivalent to applying it to the standardized features with radial function $\widetilde\kappa(r)=e^{-r^2/2}$. This function is bounded and Lipschitz, has support $[0,\infty)$, has a positive finite integral, and has a finite second moment. It therefore satisfies Assumption~\ref{assump:individual_cali_assump}\ref{itm:kernel_regularity}. Hence all parts of Assumption~\ref{assump:individual_cali_assump} hold in the standardized representation, and Theorem~\ref{thm:individual_cali_thm} applies to the corresponding fixed-bandwidth implementation.

For the numerical experiment, we select the bandwidth adaptively during
backtesting. At each time $t$ and for each scenario $z$, we choose a
scenario-specific bandwidth $h_t(z)$ from the candidate set
$\{5,1,0.5,0.1,0.05\}$. For each candidate bandwidth, we evaluate its
historical empirical coverage using observations available prior to time $t$,
excluding the observation being evaluated from the construction of its own
prediction interval, and select the bandwidth whose coverage is closest to the
target level $1-\alpha=50\%$. The bandwidth is reselected at every time step
and may therefore vary across both time and scenarios.

\paragraph{Expectation Predictor $\varphi$}
We use the oracle expectation predictor $\varphi$ defined as the true conditional expectation of $G_{t+1}'$. For any stress scenario $\bz\defeq(z_{\text{\tiny Oil}},z_{\text{\tiny Rate}})$, it takes the form
\begin{equation}\label{eq:sec_6_3_phit}
\varphi(\bz)
=
\widetilde{\mbox{BS}}(m,\tau,\sigma_{\text{tot}},1)-\E_{\mathcal Z}\left[
e^{m_{\bz}+\sqrt{v_{\text{\tiny Cred}}}\cdot\mathcal Z}
\cdot
\widetilde{\mbox{BS}}\left(
m\cdot e^{-m_{\bz}-\sqrt{v_{\text{\tiny Cred}}}\cdot\mathcal Z},
\tau-\delta,
\sigma_{\text{tot}},
1
\right)
\right],
\end{equation}
where $\mathcal Z\sim\mathcal N(0,1)$, and $m_{\bz}$ and $v_{\text{\tiny Cred}}$ are determined by $\bz$ through \eqref{eq:define_of_ms}. We approximate the right-hand side of \eqref{eq:sec_6_3_phit} using Monte Carlo simulation by drawing i.i.d.\ samples of $\mathcal Z$ and replacing $\E_{\mathcal Z}[\cdot]$ with the corresponding empirical average. The theoretical verification above concerns the exact fixed oracle function $\varphi$; the Monte Carlo approximation is an implementation device used to evaluate this function numerically.

\section{Experimental Details for Section \ref{sec:Numerics}}
\label{apx:main_exp_details}

\subsection{Model Calibration}\label{app:mo_calibrate}
The parameters of the common factor model include $\bI_0$, $\lambda_i$ for $i=1,2,3,4$, $\mathbf{M}_i, \gamma_i, \sigma_i, \eta_i$ for $i=2,3,4$, and $\sigma_G$. We calibrate the model on a real financial dataset. We obtain the implied volatility data on the S\&P 500 index from the OptionMetrics IVY database, during a selected period from January 2006 to August 2013. The data is transferred to the moneyness-maturity coordinates. We select $90$ different options to construct the log-implied volatility surface, with their moneyness-maturity pairs from $\{(m_j, \tau_j)\}_{j=1}^{90} = \{(m, \tau):m\in \Omega_m, \tau\in \Omega_\tau\}$, where $\Omega_m =\{0.5, \ 0.6, \ 0.7, \ 0.75, \ 0.8, \ 0.85, \ 0.9, \ 0.95, \ 0.975, \ 1, \ 1.025, \ 1.05, \ 1.1, \ 1.15, \ 1.2, \ 1.25, \ 1.3, \ 1.4\}$ and $\Omega_\tau =\{1/12, \ 1/6, \ 1/4, \, 1/2, \ 1\}.$ The dataset consists of $T=1890$ time steps. Applying PCA to the log-implied volatility records $\{\mathbf{I}_t\}_{t=0}^{T-1}$, we obtain the first three principal component directions, denoted as $\mathbf{M}_2,\mathbf{M}_3,\mathbf{M}_4$, along with the corresponding principal component scores $\{\mathbf{f}_t\}_{t=1}^T$. For $i=2,3,4$, the model parameters $\lambda_i, \gamma_i$ and $\sigma_i$ are estimated by fitting the autoregressive process specified in \eqref{eq:f_t_i} to $\{f_t^{(i)}\}_{t=1}^T$. For the first common factor, $\lambda_1$ and $\sigma_1$ are estimated by fitting the autoregressive process in \eqref{eq:f_t_i} to the log-underlying prices $\{\log S_t\}_{t=0}^{T-1}$. To compute $\eta_2$, $\eta_3$ and $\eta_4$, we use the relation
$\mbox{Cov}(\epsilon_t^{(1)}, \epsilon_t^{(i)}) = \eta_i \sigma_i^2$ for $i=2,3,4$, where $\epsilon_t^{(1)}$ equals the log returns of $S_t$ and $\epsilon_t^{(i)}$ are the residuals from fitting $\{f_t^{(i)}\}_{t=1}^T$ to the autoregressive process. Finally, $\sigma_G$ is computed based on \eqref{eq:sigma_G_constraint} as follows
$\sigma_G = \sqrt{\sigma_1^2 - \eta_2^2 \sigma_2^2 - \eta_3^2 \sigma_3^2 - \eta_4^2 \sigma_4^2}.$ The calibration results are as follows:
$\lambda_1 = 1.9\times 10^{-4}$, $\lambda_2 = 5.7\times 10^{-4}$, $\lambda_3 = 1.9\times 10^{-5}$, $\lambda_4 = -3.3\times 10^{-5}$, $\gamma_2 = -9.6\times 10^{-3}$, $\gamma_3 = -4.2\times 10^{-2}$, $\gamma_4 = -6.1\times 10^{-2}
$, $\eta_2 = 3.8\times 10^{-2}$, $\eta_3 = -1.4\times 10^{-2}$, $\eta_4 = 4.0\times 10^{-2}$, $\sigma_2 = 0.30$, $\sigma_3 = 0.11$, $\sigma_4 = 0.075$, $\sigma_G = 0.0028$, $\bI_0 = [-1.47, \,\,-1.50, \,\,-1.57,\,\, -1.82,\,\, -1.94,\,\,-1.93,\,\,-1.95,\,\,-1.42,\,\,\ldots]^\top$, $\bM_2 = [-0.15,\,\, -0.15,\,\,-0.13,\,\,-0.18,\,\,-0.17,\,\,-0.17,\,\,-0.17,\,\,-0.14,\,\,\ldots]^\top$, $\bM_3 = [0.18,\,\, 0.16,\,\,0.08, \\ \,\,-0.07, \,\, -0.27,\,\,-0.30, \,\,-0.27,\,\,0.19, \,\,\ldots]^\top$, $\bM_4 = [ 0.25,\,\, 0.30,\,\,0.26,\,\,0.28,\,\,0.19,\,\,0.17,\,\,0.13,\,\,0.16,\,\,\ldots]^\top$. For readability, we report only the leading entries of $\mathbf I_0$ and
$\mathbf M_2,\mathbf M_3,\mathbf M_4$ here. The complete calibrated arrays,
together with their exact moneyness--maturity ordering, are provided in the accompanying code repository. %\footnote{Upon publication, we will make the code required to reproduce the numerical experiments publicly available in a GitHub repository.}.

\subsection{Proper Scoring Rules: the CRPS and Interval Score}\label{app:proper_scoring}

We now provide a brief introduction to proper scoring rules, focusing particularly on the CRPS and interval score. For a comprehensive review of proper scoring rules, we refer the interested reader to \citet{Gneiting_Raftery}. A proper scoring rule $S$ is a mechanism that is used to evaluate the quality of a probabilistic forecast or distribution $\widehat{P}$ by assigning a numerical score that is a function  of both $\widehat{P}$ and the realized outcome. The expected score is then given by ${\cal S}(\widehat{P},P) := \int S(\widehat{P},\omega) \, dP(\omega)$
where we use $\omega$ to denote a generic outcome and $P$ is the true data-generating distribution.
Under the loss convention used here, a scoring rule is \emph{proper} if $\widehat P=P$ minimizes the expected score ${\cal S}(\widehat P,P)$, and \emph{strictly proper} if it is the unique minimizer. Smaller CRPS and interval-score values are better. The primary goal of proper scoring rules is to encourage the honest reporting of probabilities. They aid in the assessment and comparison of predictive models within statistical inference and decision theory and have been widely applied in various fields including, for example, weather forecasting and risk management.

The most commonly used score function for the probabilistic prediction of continuous random variables (scenario gains in our context) is the \textit{continuous ranked probability score (CRPS)}. The score is defined as
\begin{equation} \label{eq:CRPS1}
S(\widehat{P},y)\defeq \int_{x\in \mathbb{R}}(\widehat{F}(x)-\bbone\{x\geq y\})^2\mathrm{d} x
\end{equation}
with $\widehat{F}$ denoting the cumulative distribution function of distribution $\widehat{P}$. The CRPS is therefore the integral of the so-called Brier scores for the associated binary probability forecasts at all real-valued thresholds.
The CRPS is proper relative to the class ${\cal P}$ of Borel probability measures on $\mathbb{R}$ and strictly proper relative to the subclass ${\cal P}_1$ of ${\cal P}$ whose distributions have finite first moment. Of particular interest is when $\widehat{P}$ is an empirical distribution. In particular, let $\delta_{y}$ denote a point mass distribution at $y$ and let $\widehat{P}$ take the form $\widehat{P}= \frac{1}{m}\sum_{j=1}^m \delta_{y_j'}$, then $\widehat{F}(y)=\frac{1}{m}\sup\{j:y\geq y_j'\}$, with $-\infty =: y_0' < y_1'\leq \cdots \leq y_m'$. In this case, the CRPS can be simplified to
$S(\widehat{P}, y)=\dfrac{1}{m}\sum_{j=1}^m |y - y_j'| - \dfrac{1}{2m^2}\sum_{i=1}^m\sum_{j=1}^m |y_i' - y_j'|$. In many practical settings (including the focus of this work), we do not estimate an entire predictive distribution $\widehat{P}$. Instead, we report a prediction interval $[L,U]$. We may still want to assess whether this interval is honest. For example, if $[L,U]$ is intended to approximate the $\alpha/2$ and $1-\alpha/2$ predictive quantiles, we would like to verify that these endpoint estimates are accurate. A convenient proper scoring rule in this setting is the \textit{interval score}, defined by
\[
S_{\alpha}([L,U],y)=
\begin{cases}
(L-y)+\dfrac{\alpha}{2}(U-L), & y<L,\\[6pt]
\dfrac{\alpha}{2}(U-L), & L\le y\le U,\\[6pt]
(y-U)+\dfrac{\alpha}{2}(U-L), & y>U.
\end{cases}
\]
This is a scaled version of the original interval score introduced in \citep{Gneiting_Raftery}. Intuitively, the score trades off sharpness and calibration: it penalizes wide intervals through the term proportional to $(U-L)$, and it additionally penalizes observations that fall outside the interval in proportion to their distance from the nearest endpoint.

The CRPS and interval score are connected in the following way. The CRPS score defined in \eqref{eq:CRPS1} can be represented as an average of the pinball losses over all levels:
$S(\hat{P}, y) = 2\int_{u=0}^1 \ell_u(y, \widehat{Q}(u)) \text{d} u$, where $\ell_u$ is the pinball loss defined in \eqref{eq:pinball_loss_def} and $\widehat{Q}(u)$ is the $u$ quantile of $\widehat{P}$. The interval score defined here is the sum of the pinball losses at its two endpoint quantile levels:
$S_\alpha([L,U], y) = \ell_{\alpha/2}(y, L) + \ell_{1-\alpha/2}(y, U)$. This makes the relationship clear that CRPS aggregates quantile performance across all quantile levels, whereas the interval score samples that same quantile-based scoring idea at only two levels $\alpha/2$ and $(1-\alpha/2)$.

\subsection{Experiment and Algorithm Details}\label{app:algo_details}

In this section, we provide detailed descriptions of the algorithms evaluated in our experiments. These implementation details apply to both experimental settings considered in Sections \ref{sec:stress_steepener} and \ref{sec:adversarial_portfolio}. Moreover, because the portfolio composition evolves over time, for each backtesting date $t$, we assume that, for all prior dates $s<t$, the portfolio holds the same asset units as at time $t$, and we recompute all historical statistics accordingly. This approach follows a standard backtesting convention commonly adopted in production settings. All methods use the same normalization convention in the numerical experiments.
Portfolio gains are first expressed relative to the contemporaneous underlying-price scale and are then rescaled to the common reference level $S_0=5000$. The same
normalization and rescaling are applied to the empirical-quantile benchmarks,
ACSA, KSA, and the Monte Carlo oracle. All prediction intervals and all evaluation
metrics reported in Tables \ref{tab:BackTest_Straddle} and \ref{tab:BackTest_Adversarial} are therefore expressed on the same monetary scale.
For notational simplicity, we continue to write $G_s$ for these consistently
rescaled gains in the algorithm descriptions below.

\begin{itemize}

    \item The ``MC Oracle'' algorithm constructs a Monte Carlo approximation to the oracle conditional prediction interval for $G_{t+1}$ under the stress scenario $\bz$ and current market state. It generates synthetic draws from the known conditional data-generating distribution and uses their empirical quantiles as endpoints; these are numerical approximations, rather than exact oracle quantiles.

    \item The ``Emp-quantile'' algorithm uses the empirical quantiles of different quantities to construct the prediction intervals.
    \begin{itemize}
        \item The ``Vanilla'' variant of the algorithm first computes the empirical $\alpha/2$ and $1-\alpha/2$ quantiles of $\{G_s\}_{s\leq t}$, and then uses the two empirical quantiles as the endpoints of the prediction interval for $G_{t+1}$.
        \item The ``SSA'' variant of the algorithm leverages the empirical quantiles of the residuals of $\{G_s\}_{s\leq t}$ subtracted by their \textit{SSA predictions}. Concretely, let $\widehat{G}_t$ denote the SSA prediction of $G_t$ under the realized stress scenario $\bz_t$, then the algorithm computes the empirical $\alpha/2$ and $1-\alpha/2$ quantiles of $\{\widehat{G}_{t+1} + G_s - \widehat{G}_s\}_{s\leq t}$, and then uses the two empirical quantiles as the endpoints of the prediction interval for $G_{t+1}$.
    \end{itemize}

The ``ACSA'' algorithm has $3$ variants that differ in their choices of the quantile predictor $g_t$. The step size $\gamma$ is fixed to $0.005$.
    \begin{itemize}
        \item The ``Linear-$g_t$'' variant implements $g_t(\bz ; \tilde{\alpha})$ as a linear function whose parameters are estimated via quantile regression.
        \item The ``NN-$g_t$'' variant implements $g_t(\bz;\tilde{\alpha})$ as a 2-layer feedforward neural network. The network is trained by minimizing the pinball loss.
        \item The ``MC Oracle-$g_t$'' variant supplies $g_t$ with Monte Carlo approximations to conditional quantiles of $G_{t+1}$ given the stress scenario $\bz$ and current state, evaluated at the interior quantile levels required by the corrected ACSA construction in \eqref{eq:acsa_prediction_set}. We emphasize that this quantile predictor is not attainable in practice, as the true conditional quantile depends not only on the stress scenario but also on historical information, such as the common factors observed on the previous day.
    \end{itemize}

    \item The ``KSA NN-$\varphi$'' row presents three variants of the KSA algorithm. The three variants differ in their choices of the historical context vector $\bH_t$. For all the algorithm variants we implement the expectation predictor $\varphi$ as a 2-layer feedforward neural network. Expectation predictors $\varphi$ are fitted
once using the burn-in observations and are then held fixed throughout
the evaluation period.

    \begin{itemize}
        \item The ``$\bH_t = \bff_t^{(2:4)}$'' variant incorporates the full set of relevant historical information into the historical context vector. As shown in E-Companion \ref{appdix:his_context_of_KSA}, the components $\bff_t^{(2:4)}$ contain all historical information that is relevant for the distribution of the portfolio gain $G_{t+1}$ at time $t+1$.
        \item The ``$\bH_t = \bff_t^{(2:3)}$'' and ``$\bH_t = \bff_t^{(2)}$'' variants both use incomplete historical information.
    \end{itemize}
    For each variant of the KSA algorithm, we use the weighted RBF kernel $\kappa_h$ defined in \eqref{eq:def_gaussian_kernel}. The explicit form of the weight matrix $W$ is given in E-Companion \ref{appdix:his_context_of_KSA}. The bandwidth parameter $h$ is selected as described in E-Companion \ref{subapx:addi_details_KSA_motive}: at each time $t$, we choose $h$ from the candidate set ${5, 1, 0.5, 0.1, 0.05}$ based on empirical coverage evaluated using historical data up to time $t-1$.

    \item The ``KSA $\bH_t = \bff_t^{(2:4)}$'' row presents four variants of the KSA algorithm, which differ in the choices of the expectation predictor $\varphi$. For all the algorithm variants we set $\bH_t = \bff_t^{(2:4)}$ as the historical context vector.
    \begin{itemize}
        \item The ``no-$\varphi$'' variant sets $\varphi \equiv 0$.
        \item The ``SSA-$\varphi$'' variant uses the SSA estimation of the portfolio gain.
        \item The ``NN-$\varphi$'' variant implements $\varphi$ as a 2-layer feedforward neural network.
        \item The ``MC Oracle-$\varphi$ variant replaces $\varphi$ with a Monte Carlo approximation to the oracle conditional expectation.
    \end{itemize}

    \item The ``ACSA + KSA'' row corresponds to a new implementation that integrates the KSA quantile predictor into the ACSA framework. Specifically, the KSA quantile estimator is embedded inside ACSA while the adjusted miscoverage level $\tilde{\alpha}_t$ is updated online. The three variants differ in the choices of the historical context vector $\bH_t$ that the KSA quantile predictor uses.

\end{itemize}
Some benchmark methods and algorithm variants we evaluate assume access to \textit{oracle} knowledge of the relevant conditional distributions. In particular, conditional on a given stress scenario $\Delta \bff_{t+1}^{(1:2)}$, the distribution of the remaining common factor returns $\Delta \bff_{t+1}^{(3:4)}$ admits the following conditional mean:
\[
\E\!\left[\Delta \bff_{t+1}^{(3:4)} \,\middle|\, \Delta \bff_{t+1}^{(1:2)}, \bff_t^{(2:4)}\right]
=
\begin{pmatrix}
\lambda_3 + \gamma_3 f_t^{(3)} \\
\lambda_4 + \gamma_4 f_t^{(4)}
\end{pmatrix}
+
\big(\epsilon_{t+1}^{(1)} - \eta_2 \epsilon_{t+1}^{(2)}\big)
\begin{pmatrix}
\dfrac{\eta_3 \sigma_3^2}{\eta_3^2 \sigma_3^2 + \eta_4^2 \sigma_4^2 + \sigma_G^2} \\
\dfrac{\eta_4 \sigma_4^2}{\eta_3^2 \sigma_3^2 + \eta_4^2 \sigma_4^2 + \sigma_G^2}
\end{pmatrix}.
\]
The corresponding conditional covariance matrix is given by
\[
\mathrm{Cov}\!\left[\Delta \bff_{t+1}^{(3:4)} \,\middle|\, \Delta \bff_{t+1}^{(1:2)}, \bff_t^{(2:4)}\right]
=
\begin{pmatrix}
\dfrac{\sigma_3^2(\eta_4^2 \sigma_4^2 + \sigma_G^2)}{\eta_3^2 \sigma_3^2 + \eta_4^2 \sigma_4^2 + \sigma_G^2}
&
-\dfrac{\eta_3 \eta_4 \sigma_3^2 \sigma_4^2}{\eta_3^2 \sigma_3^2 + \eta_4^2 \sigma_4^2 + \sigma_G^2}
\\[8pt]
-\dfrac{\eta_3 \eta_4 \sigma_3^2 \sigma_4^2}{\eta_3^2 \sigma_3^2 + \eta_4^2 \sigma_4^2 + \sigma_G^2}
&
\dfrac{\sigma_4^2(\eta_3^2 \sigma_3^2 + \sigma_G^2)}{\eta_3^2 \sigma_3^2 + \eta_4^2 \sigma_4^2 + \sigma_G^2}
\end{pmatrix}.
\]
Here, $\epsilon_{t+1}^{(1)} = \Delta f_{t+1}^{(1)} - \lambda_1$ and
$\epsilon_{t+1}^{(2)} = \Delta f_{t+1}^{(2)} - \lambda_2 - \gamma_2 f_t^{(2)}$.
Given this conditional Gaussian specification, we use Monte Carlo simulation to generate a large number of samples of $\Delta \bff_{t+1}^{(3:4)}$, which are then used to compute Monte Carlo approximations to oracle conditional means, quantiles, and other functionals.

%The complete implementation used for the numerical experiments, including the full calibrated parameter arrays, neural-network architectures, optimizers, training hyperparameters and bandwidth-selection configuration will be available in a public GitHub repository upon publication.

\subsection{Computational Efficiency}\label{app:runtime}

To provide a fine-grained assessment of computational efficiency, we decompose runtime into four components: (1) the \textit{initial training time}, i.e., the time required to train the models over the burn-in period (the first 2000 time steps); (2) the average daily \textit{method-update time}, i.e., the time required at each evaluation
date to update portfolio-dependent quantities used by the method; when a learned
component is adaptively updated, this quantity also includes the corresponding
model-update cost; (3) the average daily interval-construction time, i.e., the time required to
construct the prediction interval for the target scenario after the method's
daily portfolio-dependent quantities have been updated; and (4) the \textit{total backtest time} over a 3000-step
benchmark subset of the 4000-step evaluation period. All timings are recorded on an Ubuntu system with an A100 GPU. Computational time is measured for the skew-steepener portfolio experiment reported in Table~\ref{tab:BackTest_Straddle} of Section~\ref{sec:stress_steepener}, and the results are summarized in Table~\ref{tab:runtime_comparison}.

\begin{table}[htbp]
\centering
\small
\renewcommand{\arraystretch}{1.15}
\setlength{\tabcolsep}{5pt}
\begin{tabular}{@{}lcccc@{}}
\toprule
Algorithm
& \shortstack{Initial\\training}
& \shortstack{Avg.\ daily\\method-update}
& \shortstack{Avg.\ daily\\interval construction}
& \shortstack{Total\\backtest} \\
\midrule
Emp-quantile (SSA) & 0 & 0 & $<10^{-3}$ s & $<3$ s \\
\midrule
ACSA Linear-$g_t$ & $4.8$ s & $0.93$ s & $<10^{-3}$ s & $0.78$ h \\
ACSA NN-$g_t$ & $40.7$ s & $16.2$ s & $<10^{-3}$ s & $13.51$ h \\
\midrule
KSA NN-$\varphi$ & $25.3$ s & $3.76$ s & $3\times 10^{-3}$ s & $3.14$ h \\
\midrule
ACSA + KSA & $26.5$ s & $3.79$ s & $3\times 10^{-3}$ s & $3.17$ h \\
\bottomrule
\end{tabular}
\caption{Runtime decomposition for representative methods from Table~\ref{tab:BackTest_Straddle}. Initial training time is measured over the 2000-step burn-in period. The last column is the 3000-step total implied by the reported component times, computed as initial training plus $3000$ times the sum of the average daily method-update and interval-construction times; inequalities in the component times are propagated to the total. The statistical experiments themselves use the 4000-step evaluation period.}
\label{tab:runtime_comparison}
\end{table}

The empirical-quantile baselines are computationally cheapest because they avoid
model fitting altogether. Among the proposed methods, ACSA has the highest daily
update cost, since the quantile predictor $g_t$ must be retrained or fine-tuned as
the portfolio evolves over time. By contrast, KSA fits the expectation predictor
$\varphi$ on the burn-in sample and holds it fixed throughout the evaluation period.
Its daily method-update cost therefore reflects the recomputation of
portfolio-dependent quantities and historical residuals under the current
portfolio, rather than refitting $\varphi$. This is computationally attractive because
expectation estimates can adapt to portfolio changes through weighted linear
combinations across securities, whereas quantiles do not aggregate linearly across
securities. KSA then obtains scenario-wise quantiles through kernel reweighting of
the resulting historical residuals. The hybrid ACSA+KSA implementation inherits
much of this reuse advantage while retaining ACSA's online calibration update.

\subsection{Historical Context and Kernel for KSA}\label{appdix:his_context_of_KSA}
In this section, we discuss the choice of the historical context vector $\bH_t$ and the kernel $\kappa$ used in the KSA algorithm. The structure of this section, as well as several definitions, follow E-Companion~\ref{subapx:addi_details_KSA_motive}. Our discussion serves two main purposes. First, we show that setting $\bH_t = \bff_t^{(2:4)}$ is sufficient to capture all relevant historical information. Second, we provide the explicit form of the kernel that our experiment implements.

\paragraph{Historical Context}
The portfolio gain $G_{t+1}$ has the explicit representation
\[
    G_{t+1}
    =
    u_1 (S_{t+1} - S_t)
    +
    \sum_{j=1}^{n-1} u_{j+1}
    \Big(
        \mbox{BS}(S_{t+1}, m_j S_t, \tau_j-\delta, \sigma_{t+1}^{(j)}, \mbox{c}_j)
        -
        \mbox{BS}(S_t, m_j S_t, \tau_j, \sigma_t^{(j)}, \mbox{c}_j)
    \Big),
\]
where $\mbox{BS}(\cdot)$ denotes the Black--Scholes pricing function, and $\mbox{c}_j$ is an indicator of the option type ($\mbox{c}_j = 1$ for a call and $\mbox{c}_j = 0$ for a put). As discussed in E-Companion~\ref{subapx:addi_details_KSA_motive},
constructing a prediction interval for $G_{t+1}$ is equivalent to constructing one for the normalized portfolio gain $G_{t+1}' := G_{t+1}/S_t$, which admits the explicit form
\begin{equation}\label{eq:G_prime_main_exp}
    G_{t+1}'
    =
    u_1\Bigg(\frac{S_{t+1}}{S_t} - 1\Bigg)
    +
    \sum_{j=1}^{n-1} u_{j+1}
    \Bigg(
        \frac{S_{t+1}}{S_t}
        \widetilde{\mbox{BS}}\!\left(
            \frac{m_j S_t}{S_{t+1}},
            \tau_j-\delta,
            \sigma_{t+1}^{(j)},
            \mbox{c}_j
        \right)
        -
        \widetilde{\mbox{BS}}(m_j, \tau_j, \sigma_t^{(j)}, \mbox{c}_j)
    \Bigg).
\end{equation}
Here, $\widetilde{\mbox{BS}}(\cdot)$ is defined in \eqref{eq:BS_tilde_def}. From \eqref{eq:G_prime_main_exp}, the components that constitute the \textit{historical information}
at time $t$ are $\{\sigma_t^{(j)}\}_{j=1}^{n-1}$, while the quantities that remain random at time $t$ are the underlying return $S_{t+1}/S_t$ and the future volatilities $\{\sigma_{t+1}^{(j)}\}_{j=1}^{n-1}$.
Therefore, for a historical context vector $\bH_t$ to be sufficiently informative, it suffices that $\bH_t$ both determines $\{\sigma_t^{(j)}\}_{j=1}^{n-1}$ and contains all factors that drive the distributions of $S_{t+1}/S_t$ and $\{\sigma_{t+1}^{(j)}\}_{j=1}^{n-1}$. From the model construction, choosing $\bH_t = \bff_t^{(2:4)}$ satisfies all of the above requirements. In particular, from \eqref{eq:how_f_2_4_drives_I}, we have $(\sigma_t^{(1)}, \ldots, \sigma_t^{(n-1)})
    =
    \exp \Big(
        \bI_0
        +
        f_t^{(2)} \cdot \mathbf{M}_2
        +
        f_t^{(3)} \cdot \mathbf{M}_3
        +
        f_t^{(4)} \cdot \mathbf{M}_4
    \Big)$, which shows that $\bff_t^{(2:4)}$ fully determines the current volatilities
$\{\sigma_t^{(j)}\}_{j=1}^{n-1}$. Moreover, by \eqref{eq:dynamic_I_is_iid}, the distributions of
$S_{t+1}/S_t$ and $\{\sigma_{t+1}^{(j)}\}_{j=1}^{n-1}$ are driven by $\Delta \bff_{t+1}$. As further
specified in \eqref{eq:f_t_i}, the distribution of $\Delta \bff_{t+1}^{(1)}$ is history-independent,
while the distribution of $\bff_{t+1}^{(2:4)}$ is entirely determined by $\bff_t^{(2:4)}$.

We can further show that Assumption~\ref{assump:individual_cali_assump} is
satisfied with $\bH_t = \bff_t^{(2:4)}$. As in the analysis of Section~\ref{sec:ksa-theory}, fix a deterministic continuous centering function $\varphi$ and set $Y_t=G_t'-\varphi(\bW_t)$. The discussion above establishes
Part~\ref{itm:independence_n_homogeneity}, because subtracting a function of $\bW_t$ does not change predictive sufficiency. It remains to verify Parts~\ref{itm:separate_stationary_assump}, \ref{itm:density_bound}, and \ref{itm:cond_dist_cont}. To begin, rewrite
$\bW_t$ as follows:
$$
\bW_{t} = (\Delta f_{t}^{(1)}, \Delta f_{t}^{(2)}, f_{t-1}^{(2)}, f_{t-1}^{(3)}, f_{t-1}^{(4)}) = \boldsymbol{\lambda} + \mathbf{A}_1 \bff_t^{(2:4)} + \mathbf{A}_0 \bff_{t-1}^{(2:4)} + \boldsymbol{\epsilon}_t^{G}
$$
where
$$
\boldsymbol{\lambda} = \begin{pmatrix}
    \lambda_1-\eta_2 \lambda_2 -\eta_3\lambda_3 -\eta_4\lambda_4 \\
    0 \\
    0 \\
    0 \\
    0
\end{pmatrix},\,\,
\mathbf{A}_1 = \begin{pmatrix}
    \eta_2 & \eta_3 & \eta_4 \\
    1 & 0 & 0 \\
    0 & 0 & 0 \\
    0 & 0 & 0 \\
    0 & 0 & 0
\end{pmatrix},\,\,\mathbf{A}_0 = \begin{pmatrix}
    -\eta_2(1+\gamma_2) & -\eta_3(1+\gamma_3) & -\eta_4(1+\gamma_4) \\
    -1 & 0 & 0 \\
    1 & 0 & 0 \\
    0 & 1 & 0 \\
    0 & 0 & 1
\end{pmatrix},
$$
and $\boldsymbol{\epsilon}_t^{G} = [\epsilon_t^{G},\ 0, \ 0, \ 0, \ 0]^\top$. For the calibrated parameters, $|1+\gamma_i|<1$ for $i=2,3,4$. Taking the Gaussian OU factors in their stationary regime, $\{\bff_t^{(2:4)}\}_{t=1}^\infty$ is stationary and geometrically $\beta$-mixing. The corresponding properties of $\{\bW_t\}_{t=1}^\infty$ then follow from \citet{bradley2005basic} and the following two facts:
\begin{itemize}
    \item[(1)] If $V_t = f(X_t, X_{t-1})$ for any measurable function $f$, then the $\beta$-mixing coefficients of $\{X_t\}_{t=1}^\infty$ and $\{V_t\}_{t=1}^\infty$, denoted by $\beta_X(k)$ and $\beta_V(k)$, satisfy $\beta_V(k) \leq \beta_X(k-1)$.
    \item[(2)] If $\{\epsilon_t\}_{t=1}^\infty$ is i.i.d. and independent of $\{X_t\}_{t=1}^\infty$, then the augmented process $\{(X_t,\epsilon_t)\}_{t=1}^\infty$ has $\beta$-mixing coefficients no larger than those of $\{X_t\}_{t=1}^\infty$.
\end{itemize}
By \eqref{eq:G_prime_main_exp} and the factor dynamics, $(\bW_t,G_t')$, and hence $(\bW_t,Y_t)$, is a measurable function of $(\bff_t^{(2:4)},\epsilon_t^G)$ and $(\bff_{t-1}^{(2:4)},\epsilon_{t-1}^G)$. The two facts above therefore imply that $\{(\bW_t,Y_t)\}_{t=1}^\infty$ is stationary and geometrically $\beta$-mixing. This verifies Assumption~\ref{assump:individual_cali_assump}\ref{itm:separate_stationary_assump}. The vector $\bW_t$ is a nondegenerate stationary Gaussian vector, so its density is continuous and strictly positive. It is therefore bounded above and bounded away from zero on the compact region $\mathcal D^+$, as required by Assumption~\ref{assump:individual_cali_assump}\ref{itm:density_bound}. Finally, conditional on $\bW_t=\bw$, the remaining factor innovations have a nondegenerate Gaussian distribution whose parameters vary continuously with $\bw$. For the nondegenerate portfolios considered here, \eqref{eq:G_prime_main_exp} is a smooth, nonconstant function of these innovations. Since $\varphi$ is continuous, the conditional distribution of $Y_t$ is therefore continuous and its CDF varies continuously with $\bw$. Compactness of $\mathcal D^+$ makes this continuity uniform, verifying Assumption~\ref{assump:individual_cali_assump}\ref{itm:cond_dist_cont}.

\paragraph{The Kernel} We use the weighted RBF kernel $\kappa_h$ defined in (\ref{eq:def_gaussian_kernel}). The bandwidth $h$ is adaptively tuned,\allowbreak and the weight matrix $W$ standardizes the feature vector $\bW_t = (\Delta f_{t}^{(1)},\allowbreak \Delta f_{t}^{(2)},\allowbreak f_{t-1}^{(2)},\allowbreak f_{t-1}^{(3)},\allowbreak f_{t-1}^{(4)})$. Specifically,\allowbreak the covariance matrix of $\bW_t$ takes the form
$$
\text{Cov}(\bW_t) =
\begin{pmatrix}
\sigma_1^2
    & \eta_2\,\sigma_2^2
    & 0
    & 0
    & 0
\\
\eta_2\,\sigma_2^2
    & \dfrac{2\,\sigma_2^2}{2+\gamma_2}
    & \dfrac{\gamma_2\,\sigma_2^2}{1-(1+\gamma_2)^2}
    & 0
    & 0
\\
0
    & \dfrac{\gamma_2\,\sigma_2^2}{1-(1+\gamma_2)^2}
    & \dfrac{\sigma_2^2}{1-(1+\gamma_2)^2}
    & 0
    & 0
\\
0
    & 0
    & 0
    & \dfrac{\sigma_3^2}{1-(1+\gamma_3)^2}
    & 0
\\
0
    & 0
    & 0
    & 0
    & \dfrac{\sigma_4^2}{1-(1+\gamma_4)^2}
\end{pmatrix}
$$
and we set $W = \text{Cov}(\bW_t)^{-1}$. Equivalently, under the fixed invertible standardization $\bar{\bw}=W^{1/2}\bw$, the kernel becomes
$\kappa_h(\bw,\bw')
=
\exp\left(-\frac{\|\bar{\bw}-\bar{\bw}'\|^2}{2h^2}\right).$ Thus, in the standardized coordinates, its radial function is $\widetilde{\kappa}(r)=e^{-r^2/2}$, which satisfies Assumption~\ref{assump:individual_cali_assump}\ref{itm:kernel_regularity}. The fixed invertible transformation does not affect the preceding stationarity, mixing, density, predictive-sufficiency, or continuity conclusions.

In some experiments, we additionally consider incomplete feature vectors. For instance, we may set $\bH_t = \bff_t^{(2:3)}$ instead of $\bff_t^{(2:4)}$. In such cases, the weight matrix is taken to be the corresponding submatrix of $W$ associated with the truncated feature vector. The verification of Assumption~\ref{assump:individual_cali_assump} above applies to the full context $\bH_t=\bff_t^{(2:4)}$; the incomplete-context specifications are included as deliberately misspecified empirical comparisons and need not satisfy Part~\ref{itm:independence_n_homogeneity}.

\subsection{Additional Experiment Results}

Tables \ref{tab:SSA-Cond-Exp-gain_appdix} and \ref{tab:ksa-pi-adversarial-n-neutral_appdix} present results under the same settings as Tables \ref{tab:SSA-Cond-Exp-gain} and \ref{tab:ksa-pi-adversarial-n-neutral}, but evaluated at a different time step. % (with $t = 3283$).

\begin{table}[htbp]
  \centering
\begin{minipage}[t]{0.48\textwidth}
  \centering
  \captionof*{table}{SSA: Adversarial Portfolio}
  \resizebox{\linewidth}{2.0cm}{%
    \begin{tabular}{r *{7}{r}}
      \toprule
      \mbox{$\Delta f_{t+1}^{(1)}$ \, }  & \multicolumn{7}{c}{\mbox{$ \Delta f_{t+1}^{(2)} $ ($\times 10^{-1}$)}} \\
      \cmidrule(l){2-8}
      ($\times 10^{-2}$) & $-3$ & $-2$ & $-1$ & $0$ & $1$ & $2$ & $3$ \\
      \midrule
      $-1.5$ & 2.93 & 1.43 & 0.17 & -0.89 & -1.76 & -2.45 & -3.00 \\
      $-1.0$ & 3.00 & 2.30 & 1.78 & 1.42 & 1.20 & 1.09 & 1.09 \\
      $-0.5$ & 2.04 & 1.90 & 1.91 & 2.05 & 2.29 & 2.61 & 3.00 \\
      $0.0$ & 0.27 & 0.45 & 0.78 & 1.22 & 1.75 & 2.34 & 2.98 \\
      $0.5$ & -1.75 & -1.50 & -1.10 & -0.56 & 0.07 & 0.78 & 1.54 \\
      $1.0$ & -3.00 & -3.00 & -2.79 & -2.41 & -1.88 & -1.25 & -0.54 \\
      $1.5$ & -1.94 & -2.57 & -2.90 & -2.98 & -2.85 & -2.55 & -2.11 \\
      \bottomrule
    \end{tabular}
  }
\end{minipage}
\hfill
\begin{minipage}[t]{0.48\textwidth}
  \centering
  \captionof*{table}{Cond Exp: Adversarial Portfolio}
  \resizebox{\linewidth}{2.0cm}{%
    \begin{tabular}{r *{7}{r}}
      \toprule
      \mbox{$\Delta f_{t+1}^{(1)}$ \, }  & \multicolumn{7}{c}{\mbox{$ \Delta f_{t+1}^{(2)} $ ($\times 10^{-1}$)}} \\
      \cmidrule(l){2-8}
      ($\times 10^{-2}$) & $-3$ & $-2$ & $-1$ & $0$ & $1$ & $2$ & $3$ \\
      \midrule
      $-1.5$ & 6.64 & 7.07 & 8.03 & 9.37 & 10.86 & 12.64 & 14.50 \\
      $-1.0$ & 4.13 & 5.19 & 6.56 & 8.29 & 10.27 & 12.36 & 14.46 \\
      $-0.5$ & 1.43 & 2.44 & 3.87 & 5.73 & 7.84 & 10.02 & 12.28 \\
      $0.0$ & -0.17 & -0.00 & 0.94 & 2.44 & 4.23 & 6.32 & 8.47 \\
      $0.5$ & 1.28 & -0.30 & -0.68 & -0.31 & 0.75 & 2.24 & 3.91 \\
      $1.0$ & 9.00 & 4.26 & 1.28 & -0.32 & -0.88 & -0.70 & 0.08 \\
      $1.5$ & 27.46 & 17.76 & 10.57 & 5.41 & 1.99 & -0.10 & -1.16 \\
      \bottomrule
    \end{tabular}
  }
\end{minipage}

  \vspace{0.5cm}
\begin{minipage}[t]{0.48\textwidth}
  \centering
  \captionof*{table}{SSA: Factor-neutral benchmark}
  \resizebox{\linewidth}{2.0cm}{%
    \begin{tabular}{r *{7}{r}}
      \toprule
      \mbox{$\Delta f_{t+1}^{(1)}$ \, }  & \multicolumn{7}{c}{\mbox{$ \Delta f_{t+1}^{(2)} $ ($\times 10^{-1}$)}} \\
      \cmidrule(l){2-8}
      ($\times 10^{-2}$) & $-3$ & $-2$ & $-1$ & $0$ & $1$ & $2$ & $3$ \\
      \midrule
      $-1.5$ & 0.83 & 0.46 & 0.13 & -0.15 & -0.39 & -0.59 & -0.75 \\
      $-1.0$ & 0.31 & 0.06 & -0.15 & -0.34 & -0.49 & -0.62 & -0.71 \\
      $-0.5$ & 0.07 & -0.05 & -0.17 & -0.26 & -0.34 & -0.41 & -0.45 \\
      $0.0$ & -0.07 & -0.06 & -0.06 & -0.07 & -0.07 & -0.07 & -0.07 \\
      $0.5$ & -0.34 & -0.17 & -0.03 & 0.09 & 0.18 & 0.26 & 0.32 \\
      $1.0$ & -1.07 & -0.66 & -0.31 & -0.02 & 0.22 & 0.42 & 0.58 \\
      $1.5$ & -2.67 & -1.92 & -1.28 & -0.73 & -0.26 & 0.13 & 0.46 \\
      \bottomrule
    \end{tabular}
  }
\end{minipage}
\hfill
\begin{minipage}[t]{0.48\textwidth}
  \centering
  \captionof*{table}{Cond Exp: Factor-neutral benchmark}
  \resizebox{\linewidth}{2.0cm}{%
    \begin{tabular}{r *{7}{r}}
      \toprule
      \mbox{$\Delta f_{t+1}^{(1)}$ \, }  & \multicolumn{7}{c}{\mbox{$ \Delta f_{t+1}^{(2)} $ ($\times 10^{-1}$)}} \\
      \cmidrule(l){2-8}
      ($\times 10^{-2}$) & $-3$ & $-2$ & $-1$ & $0$ & $1$ & $2$ & $3$ \\
      \midrule
      $-1.5$ & 0.43 & -0.28 & -0.90 & -1.43 & -1.88 & -2.25 & -2.56 \\
      $-1.0$ & 0.37 & -0.14 & -0.60 & -1.00 & -1.35 & -1.65 & -1.88 \\
      $-0.5$ & 0.25 & 0.01 & -0.25 & -0.50 & -0.74 & -0.95 & -1.14 \\
      $0.0$ & -0.29 & -0.12 & -0.07 & -0.09 & -0.17 & -0.26 & -0.37 \\
      $0.5$ & -1.82 & -0.98 & -0.42 & -0.05 & 0.16 & 0.27 & 0.30 \\
      $1.0$ & -5.30 & -3.38 & -1.94 & -0.90 & -0.17 & 0.33 & 0.66 \\
      $1.5$ & -12.01 & -8.46 & -5.67 & -3.49 & -1.84 & -0.64 & 0.23 \\
      \bottomrule
    \end{tabular}
  }
\end{minipage}

\caption{Portfolio gains across the scenario grid corresponding to the first two common factor returns $(\Delta f_{t+1}^{(1)}, \Delta f_{t+1}^{(2)})$. The table entries are computed at a different time step from Table \ref{tab:SSA-Cond-Exp-gain}. Cond Exp denotes Monte Carlo approximations to the oracle conditional mean gains.}
\label{tab:SSA-Cond-Exp-gain_appdix}
\end{table}

\begin{table}[htbp]
  \centering

\begin{minipage}[t]{0.90\textwidth}
  \centering
  \captionof*{table}{ACSA 90\% PI ($\tilde{\alpha}_t = 0.09$): Adversarial Portfolio}
  \vspace{-0.3cm}
  \resizebox{\linewidth}{2.2cm}{%
    \begin{tabular}{r *{7}{l}}
      \toprule
      \mbox{$\Delta f_{t+1}^{(1)}$ \, }  & \multicolumn{7}{c}{\mbox{$ \Delta f_{t+1}^{(2)} $ ($\times 10^{-1}$)}} \\
      \cmidrule(l){2-8}
      ($\times 10^{-2}$) & $-3$ & $-2$ & $-1$ & $0$ & $1$ & $2$ & $3$ \\
      \midrule
      $-1.5$ & (-4.9, 19.4) & (-3.3, 18.8) & (-1.5, 18.6) & (0.7, 19.0) & (2.9, 19.7) & (5.2, 20.7) & (7.6, 21.9) \\
      $-1.0$ & (-5.9, 15.5) & (-4.2, 15.6) & (-2.1, 16.1) & (0.3, 17.1) & (2.8, 18.4) & (5.3, 19.9) & (7.9, 21.5) \\
      $-0.5$ & (-6.9, 11.2) & (-5.6, 11.6) & (-3.8, 12.5) & (-1.6, 13.8) & (0.9, 15.4) & (3.4, 17.1) & (6.0, 19.0) \\
      $0.0$ & (-6.5, 7.5) & (-6.4, 7.6) & (-5.5, 8.3) & (-4.0, 9.6) & (-2.0, 11.1) & (0.2, 13.0) & (2.5, 14.8) \\
      $0.5$ & (-3.0, 6.4) & (-4.7, 5.2) & (-5.4, 5.0) & (-5.3, 5.6) & (-4.5, 6.7) & (-3.1, 8.2) & (-1.5, 9.8) \\
      $1.0$ & (4.1, 14.5) & (0.3, 8.5) & (-2.3, 5.2) & (-3.9, 3.7) & (-4.8, 3.6) & (-4.9, 4.1) & (-4.5, 5.1) \\
      $1.5$ & (17.6, 39.8) & (10.4, 26.8) & (5.0, 17.1) & (1.1, 10.1) & (-1.6, 5.7) & (-3.6, 3.5) & (-4.8, 2.7) \\
      \bottomrule
    \end{tabular}
  }
\end{minipage}

\vspace{0.1cm}

\begin{minipage}[t]{0.90\textwidth}
  \centering
  \captionof*{table}{ACSA 90\% PI ($\tilde{\alpha}_t = 0.1$): Factor-neutral benchmark}
  \vspace{-0.3cm}
  \resizebox{\linewidth}{2.2cm}{%
    \begin{tabular}{r *{7}{l}}
      \toprule
      \mbox{$\Delta f_{t+1}^{(1)}$ \, }  & \multicolumn{7}{c}{\mbox{$ \Delta f_{t+1}^{(2)} $ ($\times 10^{-1}$)}} \\
      \cmidrule(l){2-8}
      ($\times 10^{-2}$) & $-3$ & $-2$ & $-1$ & $0$ & $1$ & $2$ & $3$ \\
      \midrule
      $-1.5$ & (-0.3, 1.4) & (-0.9, 0.5) & (-1.4, -0.3) & (-1.9, -0.9) & (-2.3, -1.4) & (-2.7, -1.8) & (-3.0, -2.1) \\
      $-1.0$ & (-0.1, 1.1) & (-0.5, 0.4) & (-1.0, -0.2) & (-1.4, -0.6) & (-1.8, -0.9) & (-2.1, -1.2) & (-2.4, -1.4) \\
      $-0.5$ & (-0.1, 0.8) & (-0.2, 0.3) & (-0.5, -0.0) & (-0.8, -0.2) & (-1.2, -0.3) & (-1.4, -0.5) & (-1.7, -0.6) \\
      $0.0$ & (-1.0, 0.3) & (-0.5, 0.2) & (-0.2, 0.0) & (-0.3, 0.1) & (-0.5, 0.2) & (-0.7, 0.2) & (-0.9, 0.2) \\
      $0.5$ & (-3.5, -0.6) & (-2.0, -0.2) & (-1.0, -0.0) & (-0.3, 0.2) & (-0.1, 0.4) & (-0.1, 0.7) & (-0.2, 0.8) \\
      $1.0$ & (-8.4, -2.8) & (-5.7, -1.6) & (-3.5, -0.8) & (-2.0, -0.2) & (-0.8, 0.3) & (-0.0, 0.7) & (0.3, 1.0) \\
      $1.5$ & (-17.6, -7.4) & (-12.8, -4.9) & (-8.9, -3.1) & (-5.9, -1.6) & (-3.6, -0.6) & (-1.8, 0.2) & (-0.5, 0.8) \\
      \bottomrule
    \end{tabular}
  }
\end{minipage}

\vspace{0.1cm}

\begin{minipage}[t]{0.90\textwidth}
  \centering
  \captionof*{table}{KSA 90\% PI: Adversarial Portfolio}
  \vspace{-0.3cm}
  \resizebox{\linewidth}{2.2cm}{%
    \begin{tabular}{r *{7}{l}}
      \toprule
      \mbox{$\Delta f_{t+1}^{(1)}$ \, }  & \multicolumn{7}{c}{\mbox{$ \Delta f_{t+1}^{(2)} $ ($\times 10^{-1}$)}} \\
      \cmidrule(l){2-8}
      ($\times 10^{-2}$) & $-3$ & $-2$ & $-1$ & $0$ & $1$ & $2$ & $3$ \\
      \midrule
      $-1.5$ & (-48.4, 37.7) & (-58.7, 32.2) & (-64.3, 27.1) & (-83.7, 26.0) & (-99.9, 25.2) & (-100.6, 24.5) & (-101.1, 23.9) \\
      $-1.0$ & (-35.3, 45.2) & (-42.2, 36.4) & (-51.9, 30.3) & (-59.6, 28.4) & (-64.2, 28.1) & (-97.0, 28.0) & (-97.0, 28.0) \\
      $-0.5$ & (-26.6, 55.6) & (-31.9, 45.4) & (-39.7, 36.8) & (-45.5, 30.5) & (-54.7, 26.1) & (-58.8, 27.3) & (-77.5, 29.9) \\
      $0.0$ & (-20.4, 74.8) & (-24.5, 58.0) & (-29.9, 45.9) & (-35.8, 36.1) & (-41.9, 30.0) & (-49.9, 24.1) & (-56.0, 21.7) \\
      $0.5$ & (-15.9, 94.7) & (-17.4, 77.8) & (-22.3, 61.5) & (-29.0, 48.0) & (-33.8, 36.3) & (-41.3, 29.0) & (-47.1, 21.5) \\
      $1.0$ & (-12.1, 146.6) & (-15.4, 108.2) & (-16.9, 77.8) & (-18.9, 65.9) & (-25.7, 51.0) & (-32.7, 37.0) & (-40.5, 27.8) \\
      $1.5$ & (-6.9, 234.2) & (-8.4, 218.1) & (-12.0, 115.5) & (-13.5, 79.5) & (-17.1, 71.7) & (-23.6, 52.8) & (-31.0, 38.9) \\
      \bottomrule
    \end{tabular}
  }
\end{minipage}

\vspace{0.1cm}

\begin{minipage}[t]{0.90\textwidth}
  \centering
  \captionof*{table}{KSA 90\% PI: Factor-neutral benchmark}
  \vspace{-0.3cm}
  \resizebox{\linewidth}{2.2cm}{%
    \begin{tabular}{r *{7}{l}}
      \toprule
      \mbox{$\Delta f_{t+1}^{(1)}$ \, }  & \multicolumn{7}{c}{\mbox{$ \Delta f_{t+1}^{(2)} $ ($\times 10^{-1}$)}} \\
      \cmidrule(l){2-8}
      ($\times 10^{-2}$) & $-3$ & $-2$ & $-1$ & $0$ & $1$ & $2$ & $3$ \\
      \midrule
      $-1.5$ & (-9.0, 12.2) & (-7.6, 12.4) & (-7.1, 13.1) & (-6.7, 15.0) & (-6.4, 15.8) & (-5.1, 14.6) & (-4.1, 10.6) \\
      $-1.0$ & (-11.2, 11.4) & (-9.8, 11.5) & (-8.0, 11.9) & (-7.2, 13.3) & (-6.5, 14.7) & (-6.3, 15.9) & (-5.2, 14.5) \\
      $-0.5$ & (-15.8, 9.9) & (-12.0, 10.9) & (-10.0, 11.4) & (-8.0, 12.0) & (-6.9, 13.7) & (-6.2, 14.8) & (-5.4, 16.1) \\
      $0.0$ & (-20.4, 7.2) & (-17.9, 8.7) & (-12.5, 10.5) & (-10.4, 11.5) & (-7.8, 12.7) & (-6.1, 14.2) & (-5.0, 15.1) \\
      $0.5$ & (-27.7, 3.3) & (-20.7, 6.1) & (-18.9, 7.7) & (-14.2, 10.1) & (-10.7, 11.7) & (-8.1, 13.6) & (-5.7, 14.8) \\
      $1.0$ & (-37.8, 1.8) & (-30.0, 2.7) & (-21.8, 4.9) & (-19.8, 7.3) & (-15.9, 9.5) & (-11.1, 11.9) & (-8.3, 14.3) \\
      $1.5$ & (-69.7, -0.8) & (-52.9, 0.7) & (-34.9, 2.0) & (-26.4, 3.4) & (-20.7, 6.6) & (-17.9, 8.6) & (-12.4, 11.7) \\
      \bottomrule
    \end{tabular}
  }
\end{minipage}
\caption{ACSA and KSA 90\% prediction intervals. The prediction intervals are computed at a different time step from Table~\ref{tab:ksa-pi-adversarial-n-neutral}.}
\label{tab:ksa-pi-adversarial-n-neutral_appdix}
\end{table}

\subsection{Construction of the Adversarial Portfolio}\label{appdix:adversarial_portfolio}
The portfolio consists of \(n = 13\) securities. The first asset is the S\&P 500 index. The remaining \(n-1\) assets are European options characterized by moneyness \(\{m_j\}_{j=1}^{n-1}\), time to maturity \(\{\tau_j\}_{j=1}^{n-1}\), implied volatility \(\{\sigma_t^{(j)}\}_{j=1}^{n-1}\), and a call-put indicator \(\{\mathrm{c}_j\}_{j=1}^{n-1}\). Let \(\{C_t^{(j)}\}_{j=1}^{n-1}\) denote the corresponding option prices, given by $C_t^{(j)} = \mathrm{BS}\left(S_t,\, m_j S_t,\, \tau_j,\, \sigma_t^{(j)},\, \mathrm{c}_j\right)$, where \(\mathrm{BS}(\cdot)\) denotes the Black-Scholes pricing formula. The option universe is constructed by forming all combinations of moneyness and maturity levels: $\{(m_j, \tau_j)\}_{j=1}^{12} = \{(m, \tau) : m \in \Omega_m,\ \tau \in \Omega_\tau\}$, where $\Omega_m = \{0.9,\, 0.95,\, 1.05,\, 1.1\}$ and $\Omega_\tau = \{1/12,\, 1/4,\, 1/2\}$. For each option, the call-put indicator is \(\mathrm{c}_j = 1\) when \(m_j \geq 1\) (call option) and \(\mathrm{c}_j = 0\) when \(m_j < 1\) (put option). The adversarial portfolio $\boldsymbol{u}$ (represented by the number of units held in each security) is obtained by solving the following optimization problem:
\begin{equation}\label{eq:adversarial_LP_program}
\begin{aligned}
\max_{\boldsymbol{u}^+,\, \boldsymbol{u}^- \ge 0} \quad &
    \widetilde{G}_{t+1}(\boldsymbol{u}^+ - \boldsymbol{u}^-,\, \be_3),\\[3pt]
\text{s.t.}\quad
& |\widehat{G}_{t+1, SSA}(\boldsymbol{u}^+ - \boldsymbol{u}^-,\, \bz^i)| \leq  3,\quad i=1,\ldots, m\\[2pt]
& \bC_t^{\top}\boldsymbol{u}^+ \le 100,\qquad \bC_t^{\top}\boldsymbol{u}^- + \bC_t^{\top}\boldsymbol{u}^+ \le 500,
\end{aligned}
\end{equation}
where
\begin{itemize}
    \item $\widetilde{G}_{t+1}(\boldsymbol{u}, \Delta \bff_{t+1})$: the approximated portfolio gain as a linear function of the portfolio position vector $\boldsymbol{u}$ and the common factor return $\Delta \bff_{t+1}$.
    \item $\mathbf u^+,\mathbf u^-$: nonnegative auxiliary position variables, with the net portfolio position given by
$\mathbf u=\mathbf u^+-\mathbf u^-$.
    \item $\be_i = (0,\ldots,1,\ldots,0)\in \mathbb{R}^4$: the $i$-th canonical basis vector.

    \item $\widehat{G}_{t+1, SSA}(\boldsymbol{u}, \bz)$: the SSA estimate of the portfolio gain as a function of the portfolio position vector $\boldsymbol{u}$ and the stress scenario $\bz$.
    \item $\{\bz^1,\cdots, \bz^m\}$: the collection of candidate stress scenarios used in the scenario analysis table. In this experiment, we set $m=49$ with $\{\bz^i\}_{i=1}^{49}
= \{(z_1,z_2) : z_1 \in \Omega_1,\ z_2 \in \Omega_2\}$, where $\Omega_1 = \{-0.015,\,-0.01,\,-0.005,\,0,\,0.005,\,0.01,\,0.015\}$ and $\Omega_2 = \{-0.3,\,-0.2,\,-0.1,\,0,\,0.1,\,0.2,\,0.3\}$.
    \item $\bC_t =
\big(S_t,\ C_t^{(1)},\ldots, C_t^{(n-1)}\big)$: the vector of current security values.

\end{itemize}
The objective of program \eqref{eq:adversarial_LP_program} maximizes the portfolio gain under an upward movement of the third factor, represented by $\be_3$. The first constraint requires the SSA estimate of the portfolio gain to lie within $\pm 3\%$ of the total capital. The second constraint ensures that long positions are fully financed by the \$100 capital. Finally, the last constraint places a cap on the total notional represented by
the positive and negative auxiliary position variables in the split-position
formulation. The explicit forms of $\widetilde{G}_{t+1}(\boldsymbol{u}, \Delta \bff_{t+1})$ and $\widehat{G}_{t+1,SSA}(\boldsymbol{u}, \bz)$ are presented next.

\paragraph{First Order Approximation of $G_{t+1}$}
A first-order approximation to the portfolio gain gives:
\begin{equation}\label{eq:linear_approx_of_portgain}
    \begin{aligned}
        G_{t+1} &= u_1 \cdot (S_{t+1} - S_t) + \sum_{j=1}^{n-1} u_{j+1}\cdot (C_{t+1}^{(j)} - C_t^{(j)}) \\
        &\approx u_1 \cdot \text{d} S_{t+1} + \sum_{j=1}^{n-1} u_{j+1}\cdot \Big(\delta_{t}^{(j)}\text{d} S_{t+1} + \nu_t^{(j)}\cdot \text{d}\sigma_{t+1}^{(j)} - \frac{1}{252}\Theta_t^{(j)}\Big)
    \end{aligned}
\end{equation}
where the Greek letters are: $
    \delta_t^{(j)} \defeq \dfrac{\partial}{\partial S} \text{BS}(S, m_j S_t, \tau_j, \sigma_t^{(j)}$, $\mbox{c}_j) \Big|_{S = S_t}$, $
    \nu_t^{(j)} \defeq \dfrac{\partial}{\partial \sigma} \text{BS}(S_t, m_j S_t, \tau_j, \sigma, \mbox{c}_j) \Big|_{\sigma = \sigma_t^{(j)}}$, and $
    \Theta_t^{(j)} \defeq \dfrac{\partial}{\partial \tau} \text{BS}(S_t, m_j S_t, \tau, \sigma_t^{(j)}, \mbox{c}_j) \Big|_{\tau = \tau_j}$. Equation~\eqref{eq:linear_approx_of_portgain} decomposes fluctuations in the portfolio value into three components: (i) movements in the underlying index,
(ii) the passage of time, which reduces the remaining maturity, and (iii) shifts in the implied volatility surface. Moreover, \eqref{eq:linear_approx_of_portgain} implies that the portfolio gain can be approximated by a linear function of the common factor returns $\Delta \bff_{t+1}$. To see this, recall from \eqref{eq:dynamic_I_is_iid} that
\begin{equation}\label{eq:pass_of_linear}
\begin{aligned}
    \text{d} S_{t+1} &\approx S_t \cdot \Delta f_{t+1}^{(1)}\\
    (\text{d} \sigma_{t+1}^{(1)},\ldots \text{d} \sigma_{t+1}^{(n-1)}) &\approx (\sigma_t^{(1)},\ldots, \sigma_t^{(n-1)})\odot \Delta \bI_{t+1}\\
    &= (\sigma_t^{(1)},\ldots, \sigma_t^{(n-1)}) \odot \left(\Delta f_{t+1}^{(2)} \cdot \mathbf{M}_2 + \Delta f_{t+1}^{(3)}\cdot \mathbf{M}_3 + \Delta f_{t+1}^{(4)} \cdot\mathbf{M}_4\right)
\end{aligned}
\end{equation}
where $\odot$ denotes the Hadamard (element-wise) product. Define the approximated portfolio gain $\widetilde{G}_{t+1}(\boldsymbol{u}, \Delta \bff_{t+1})$ as a linear function of the common factor returns $\Delta \bff_{t+1}$ and the portfolio position vector $\boldsymbol{u}$. This approximation is obtained by substituting \eqref{eq:pass_of_linear} into the linearized gain expression in \eqref{eq:linear_approx_of_portgain}. The explicit expression of $\widetilde{G}_{t+1}$ is
$$
\begin{aligned}
    \widetilde{G}_{t+1}(\boldsymbol{u}, \Delta \bff_{t+1}) =& u_1 S_t \Delta f_{t+1}^{(1)} + S_t \Delta f_{t+1}^{(1)} \cdot \boldsymbol{\delta}_t^\top \boldsymbol{u}^{(2:n)} \\
    &+ \Bigg(\boldsymbol{\nu}_t \odot \boldsymbol{\sigma}_t \odot \left(\Delta f_{t+1}^{(2)} \cdot \mathbf{M}_2 + \Delta f_{t+1}^{(3)}\cdot \mathbf{M}_3 + \Delta f_{t+1}^{(4)} \cdot\mathbf{M}_4\right)
    - \frac{1}{252} \boldsymbol{\Theta}_t\Bigg)^\top \boldsymbol{u}^{(2:n)}
\end{aligned}
$$

\paragraph{SSA Estimate of $G_{t+1}$}
Under a given stress scenario $\bz = (z_1, z_2)$, the SSA estimate $\widehat{G}_{t+1, SSA}(\boldsymbol{u},\allowbreak \bz)$ is given by
$$
\widehat{G}_{t+1, SSA}(\boldsymbol{u}, \bz)
    =
    u_1 S_t \Big(e^{z_1} - 1\Big) + \sum_{j=1}^{n-1} u_{j+1} S_t
    \Bigg(
        e^{z_1}\widetilde{\mbox{BS}}\!\left(
            m_j e^{-z_1},
            \tau_j-\delta,
            \sigma_{t+1}^{(j)},
            \mbox{c}_j
        \right)
        -
        \widetilde{\mbox{BS}}(m_j, \tau_j, \sigma_t^{(j)}, \mbox{c}_j)
    \Bigg)
$$
where the post-stress implied volatilities satisfy
$(\sigma_{t+1}^{(1)},\ldots, \sigma_{t+1}^{(n-1)}) = (\sigma_t^{(1)},\ldots, \sigma_t^{(n-1)})\odot \exp \left(z_2\cdot \bM_2\right).$
This expression is derived analogously to \eqref{eq:G_prime_main_exp}.

\paragraph{Factor-neutral benchmark.}
In a similar spirit to \eqref{eq:adversarial_LP_program}, we construct a factor-neutral benchmark in a
relative sense. The benchmark uses the same linearized one-day gain approximation
as above, but imposes additional constraints requiring the approximated one-day
gain to vanish under unit movements in each of the four modeled common-factor
directions. Because the approximation also contains the one-day time-decay
contribution, these constraints do not in general imply exactly zero marginal
sensitivity to each factor. Rather, their role is to make the benchmark more
tightly controlled than the adversarial portfolio along all four modeled factor
directions.
\begin{equation}\label{eq:neutral_LP_program}
\begin{aligned}
\max_{\boldsymbol{u}^+,\, \boldsymbol{u}^- \ge 0} \quad &
     \bC_t^{\top}\boldsymbol{u}^- + \bC_t^{\top}\boldsymbol{u}^+,\\[2pt]
\text{s.t.}\quad
& |\widehat{G}_{t+1, SSA}(\boldsymbol{u}^+ - \boldsymbol{u}^-,\, \bz^i)| \leq  3,\quad i=1,\ldots, m\\[2pt]
& \widetilde{G}_{t+1}(\boldsymbol{u}^+ - \boldsymbol{u}^-,\, \be_j) = 0, \quad j = 1,\ldots, 4\\[2pt]
& \bC_t^{\top}\boldsymbol{u}^+ \leq 100,\qquad \bC_t^{\top}\boldsymbol{u}^- + \bC_t^{\top}\boldsymbol{u}^+ \le 500.
\end{aligned}
\end{equation}
The objective in \eqref{eq:neutral_LP_program} serves as a selection criterion within the split-position
LP representation, favoring a large total notional allocation subject to the
scenario, directional-gain, and capital constraints. Since $\mathbf u^+$ and
$\mathbf u^-$ are auxiliary variables and need not coincide with the positive and
negative parts of the net position $\mathbf u=\mathbf u^+-\mathbf u^-$, this
objective should not be interpreted as the gross exposure computed from the
absolute net holdings.

\section{Proofs}
\label{sec:Proofs}

\subsection{Proof of Theorem \ref{thm:empirical_coverage}}
\label{subapx:proofThmEmpiricalCoverage}

We include the proof of Theorem \ref{thm:empirical_coverage} here for the sake of completion. As mentioned in Section \ref{sec:ACSA}, the theorem and proof are due to \citet{gibbs2021adaptive}. Throughout the proof, we use the theoretical prediction sets defined in \eqref{eq:acsa_prediction_set}. We need the following lemma.

\begin{lem}\label{lemma:Bounded_g_t}
Let $\alpha\in(0,1)$, $\gamma>0$, and $\tilde\alpha_1=\alpha$. With the prediction sets in \eqref{eq:acsa_prediction_set} and finite realized gains, the ACSA recursion satisfies $\tilde\alpha_t\in[-\gamma,1+\gamma]$ for every integer $t\geq1$.
\end{lem}

\begin{proof}
The claim holds at $t=1$. Suppose it holds at time $t$. The update is
\begin{equation}\label{eq:beta_t_updaterule}
\tilde\alpha_{t+1}=\tilde\alpha_t+\gamma(\alpha-\mathrm{err}_{t+1}).
\end{equation}
If $\tilde\alpha_t\leq0$, the prediction set is $\mathbb R$, so $\mathrm{err}_{t+1}=0$ and $\tilde\alpha_{t+1}=\tilde\alpha_t+\gamma\alpha\in[-\gamma,1+\gamma]$. If $\tilde\alpha_t\geq1$, the prediction set is $\emptyset$, so $\mathrm{err}_{t+1}=1$ and $\tilde\alpha_{t+1}=\tilde\alpha_t-\gamma(1-\alpha)\in[-\gamma,1+\gamma]$. Finally, if $0<\tilde\alpha_t<1$, either value of the indicator gives $-\gamma<\tilde\alpha_{t+1}<1+\gamma$. Induction proves the result.
\end{proof}

Theorem \ref{thm:empirical_coverage} now follows easily from Lemma \ref{lemma:Bounded_g_t}.

\begin{proof}[{\bf Proof of Theorem \ref{thm:empirical_coverage}}]
    After observing $G_{T+1}$, define $\mathrm{err}_{T+1}$ and $\tilde{\alpha}_{T+1}$ by the same recursion. The update rule \eqref{eq:beta_t_updaterule} of $\tilde{\alpha}_t$, together with Lemma \ref{lemma:Bounded_g_t}, implies
     \begin{equation}\label{eq:acsa-telescoping-identity}
    \begin{aligned}
     \tilde{\alpha}_{T+1} &= \tilde{\alpha}_T + \gamma\cdot (\alpha - \mathrm{err}_{T+1}) = \tilde{\alpha}_{T-1} + \gamma\cdot (\alpha - \mathrm{err}_{T+1}) + \gamma\cdot (\alpha - \mathrm{err}_{T})=\cdots \\
     &=\alpha + \gamma \cdot\sum_{t=2}^{T+1} (\alpha-\mathrm{err}_t) \in [-\gamma, 1+\gamma].
    \end{aligned}
    \end{equation}
    Rearranging yields
    $\left|\dfrac{1}{T}\sum_{t=2}^{T+1} \mathrm{err}_t-\alpha\right| \leq \dfrac{\gamma + \max\{\alpha, 1-\alpha\}}{T\gamma}.$ Since $\mathrm{err}_{t+1}=1-\bbone\{G_{t+1}\in\Ct_t(\bz_{t+1})\}$, this is equivalent to the stated coverage bound.
\end{proof}

\subsection{Proofs for Section \ref{sec:ksa-theory}}
\label{subapx:proof_of_individual_cali}

To simplify the notation, define $d:=d_{\bW}, K(\mathbf v):=\widetilde\kappa(\|\mathbf v\|)$, $
\overline K:=\sup_{\mathbf v\in\mathbb R^d}K(\mathbf v).$  For $\bw\in\mathcal D$ and $y\in\mathbb R$, define $A_{t,h}(y,\bw):=
\frac1t\sum_{s=1}^t
\kappa_h(\bW_s,\bw)\bbone\{Y_s\leq y\}$, $A_h(y,\bw):=
\mathbb E\!\left[
\kappa_h(\bW_1,\bw)\bbone\{Y_1\leq y\}
\right]$, $B_{t,h}(\bw):=
\frac1t\sum_{s=1}^t\kappa_h(\bW_s,\bw)$, $B_h(\bw):=
\mathbb E[\kappa_h(\bW_1,\bw)]$, and $F_h(y\mid\bw):=
\frac{A_h(y,\bw)}{B_h(\bw)}.$ Whenever $B_{t,h}(\bw)>0$, we have $\widehat F_t(y\mid\bw)
=
\frac{A_{t,h}(y,\bw)}{B_{t,h}(\bw)}.$ For $B_{t,h}(\bw)>0$ and $u\in(0,1)$, define the $u$-quantile of $\widehat F_t(y\mid\bw)$ as
\begin{equation}
\widehat Q_t(u;\bw)
:=
\inf\{y\in\mathbb R:\widehat F_t(y\mid\bw)\geq u\}.
\label{eq:ksa-proof-quantile-convention}
\end{equation}
And if $B_{t,h}(\bw)=0$, set
$\widehat F_t(y\mid\bw)=\bbone\{y\geq0\}$ and
$\widehat Q_t(u;\bw)=0$.

The primary goal of this section is to show that $\widehat F_t(y\mid\bw)$ accurately estimates $F_{Y\mid\bW}(y\mid\bw)$ for all  $\bw\in\mathcal D$ and $y\in\mathbb R$. The following decomposition provides the basic structure for the proofs.
\begin{equation}
\sup_{\bw\in\mathcal D}\sup_{y\in\mathbb R}
\left|
\widehat F_t(y\mid\bw)-F_{Y\mid\bW}(y\mid\bw)
\right|\leq
\sup_{\bw\in\mathcal D}\sup_{y\in\mathbb R}
\left|
\widehat F_t(y\mid\bw)-F_h(y\mid\bw)
\right|
+
\sup_{\bw\in\mathcal D}\sup_{y\in\mathbb R}
\left|
F_h(y\mid\bw)-F_{Y\mid\bW}(y\mid\bw)
\right|.
\label{eq:ksa-error-decomposition}
\end{equation}
The first term on the right-hand side is controlled by an empirical-process argument, and the second by a direct kernel-bias calculation. A key difficulty is the dependence among observations. The $\beta$-mixing condition allows sufficiently separated blocks of observations to be treated as ``almost'' independent for concentration purposes. The following lemma formalizes this comparison using a \textit{blocking argument} based on the idea of \citet{yu1994rates}.

\begin{lem}
\label{lem:ksa-separated-comparison}
Let $\{X_s\}_{s\geq1}$ be a $\beta$-mixing process taking values in a
standard Borel space.  For $1\leq j\leq m$, let
$V_j=(X_{a_j},\ldots,X_{b_j})$ for $a_j\leq b_j<a_{j+1},$ and let $V_1^*,\ldots,V_m^*$ be independent, with $V_j^*$ distributed as
$V_j$.  If $\psi$ is measurable and takes values in $[0,M]$, then
$$
\left|
\mathbb E\psi(V_1,\ldots,V_m)
-\mathbb E\psi(V_1^*,\ldots,V_m^*)
\right|
\leq
M\sum_{j=2}^m\beta(a_j-b_{j-1}).
$$
In particular, if $a_j-b_{j-1}\geq q$ for every $j\geq2$, the right-hand
side is at most $M(m-1)\beta(q)$.
\end{lem}

\begin{proof}
For probability measures $\mu$ and $\nu$, use the convention
$\|\mu-\nu\|_{\mathrm{TV}}
=
\sup_{0\leq g\leq1}
\left|
\int g\td\mu-\int g\td\nu
\right|.$ For $j\geq2$, set
$\mathcal G_{j-1}=\sigma(\{V_1,\ldots,V_{j-1}\})$, $\mathcal P_{j-1}=\sigma(\{X_s:s\leq b_{j-1}\})$ and $\mathcal A_j=\sigma(V_j).
$ Then
$\mathcal G_{j-1}\subseteq\mathcal P_{j-1}$ and $\mathcal A_j\subseteq\sigma(\{X_s:s\geq a_j\}).$ For $B\in\mathcal A_j$, the tower property gives
$\mathbb P(B\mid\mathcal G_{j-1})-\mathbb P(B)
=
\mathbb E\left[
\mathbb P(B\mid\mathcal P_{j-1})-\mathbb P(B)
\mid\mathcal G_{j-1}
\right].$ Consequently,
$$
\mathbb E\left[
\sup_{B\in\mathcal A_j}
\left|
\mathbb P(B\mid\mathcal G_{j-1})-\mathbb P(B)
\right|
\right]\leq
\mathbb E\left[
\sup_{B\in\sigma(\{X_s:s\geq a_j\})}
\left|
\mathbb P(B\mid\mathcal P_{j-1})-\mathbb P(B)
\right|
\right]
\leq
\beta(a_j-b_{j-1}).
$$
Because the state space is standard Borel, the left-hand side equals the
total-variation distance between the joint law of
$(V_1,\ldots,V_j)$ and the \textit{product of the law} (denoted by $\otimes$) of
$(V_1,\ldots,V_{j-1})$ and the marginal law of $V_j$.  Thus
$\left\|
\mathcal L(V_1,\ldots,V_j)
-\mathcal L(V_1,\ldots,V_{j-1})\otimes\mathcal L(V_j)
\right\|_{\mathrm{TV}}
\leq
\beta(a_j-b_{j-1}).$ For $1\leq j\leq m$, define the hybrid law
$\nu_j
=
\mathcal L(V_1,\ldots,V_j)
\otimes
\bigotimes_{\ell=j+1}^m\mathcal L(V_\ell).$ Then $\nu_1$ is the product law of
$(V_1^*,\ldots,V_m^*)$, whereas $\nu_m$ is the joint law of
$(V_1,\ldots,V_m)$. Tensoring two probability measures with a common
probability measure does not increase their total variation distance.
Therefore,
$\|\nu_j-\nu_{j-1}\|_{\mathrm{TV}}
\leq
\beta(a_j-b_{j-1})$ for $2\leq j\leq m.$ Applying the triangle inequality to the sequence
$\nu_1,\ldots,\nu_m$ yields
$$
\|\nu_m-\nu_1\|_{\mathrm{TV}}
\leq
\sum_{j=2}^m
\|\nu_j-\nu_{j-1}\|_{\mathrm{TV}}
\leq
\sum_{j=2}^m\beta(a_j-b_{j-1}).
$$
Since $0\leq\psi/M\leq1$, the definition of total variation gives
$$
\begin{aligned}
\left|
\mathbb E\psi(V_1,\ldots,V_m)
-
\mathbb E\psi(V_1^*,\ldots,V_m^*)
\right|& =
\left|
\int\psi\,\td\nu_m-\int\psi\,\td\nu_1
\right|=
M\left|
\int\frac{\psi}{M}\,\td\nu_m
-
\int\frac{\psi}{M}\,\td\nu_1
\right|\\
&\leq
M\|\nu_m-\nu_1\|_{\mathrm{TV}}\leq
M\sum_{j=2}^m\beta(a_j-b_{j-1}).
\end{aligned}
$$
Finally, if $a_j-b_{j-1}\geq q$ for every $j\geq2$, the monotonicity of
the mixing coefficients implies
$\sum_{j=2}^m\beta(a_j-b_{j-1})
\leq
(m-1)\beta(q),$ which proves the final assertion.
\end{proof}

Lemma~\ref{lem:ksa-separated-comparison} provides the bridge
from an empirical-process bound for independent observations to the dependent sample. For any statistic bounded by $M$,
replacing $m$ dependent blocks by independent copies with the same marginal
distributions changes its expectation by at most
$M\sum_{j=2}^m\beta(a_j-b_{j-1}).$
Thus, the effect of temporal dependence is represented by an explicit
additive error that decreases as the separation between the blocks increases.

We now establish a uniform concentration bound between
$A_h(y,\bw), B_h(\bw)$ and their empirical estimates
$A_{t,h}(y,\bw), B_{t,h}(\bw)$. The proof first derives a maximal
inequality for independent samples and then transfers it to the original
dependent sample using Lemma \ref{lem:ksa-separated-comparison}. The independent-sample argument is an empirical-process bound over a class of kernel superlevel sets and threshold events. To state this argument
compactly, we introduce the VC and covering-number notation that will
be used below.

We use standard notation from empirical process theory. For a class
$\mathcal F$ of measurable functions and a probability measure $P$, let
$N(\epsilon,\mathcal F,L_2(P))$ denote the $\epsilon$-covering number of
$\mathcal F$ under the $L_2(P)$ semimetric. We write
$\operatorname{VCdim}(\mathcal A)$ for the VC dimension of a class
$\mathcal A$ of measurable sets. Standard background on covering numbers
and VC classes may be found, for example, in
\citet{VanDerVaartWellner1996KSAAppendix} and
\citet{mohri2018foundations}. Let $\mathcal B_d$ denote the class consisting of the empty set and all
open or closed Euclidean balls in $\mathbb R^d$, and define
$\mathcal E_d
=
\left\{
B\times(-\infty,y]:
B\in\mathcal B_d,\;
y\in\mathbb R\cup\{\infty\}
\right\}.$ The proof leverages the standard facts that
$\operatorname{VCdim}(\mathcal B_d)=d+1$ and that, for any class
$\mathcal A$ with $\operatorname{VCdim}(\mathcal A)\leq v$ and any
probability measure $P$,
\begin{equation}\label{eq:EC-lemma_eq}
N\!\left(
\epsilon,
\{\bbone_A:A\in\mathcal A\},
L_2(P)
\right)
\leq
e(v+1)
\left(\frac{2e}{\epsilon^2}\right)^v,
\qquad 0<\epsilon<1.
\end{equation}
The latter can be easily derived from \citet{Haussler1995KSAAppendix}. We are now ready to state the lemma.

\begin{lem}
\label{lem:ksa-fixed-h-concentration}
Suppose parts~\ref{itm:separate_stationary_assump} and
\ref{itm:kernel_regularity} of
Assumption~\ref{assump:individual_cali_assump} hold. For any
$\delta\in(0,1)$, define
$D_d
:=
6\sqrt{
2\left\{
\log\bigl(18e^2(d+3)^2\bigr)
+2(d+2)\log(32e)
\right\}
}$, $T_\rho
:=
\max\left\{
3,
\left\lceil
\left(1+\frac{2}{\log(1/\rho)}\right)^2
\right\rceil
\right\}$, and $R_\delta
:=
\frac{\overline K}{\delta}
\left[
D_d\sqrt{1+\frac{2}{\log(1/\rho)}}
+\beta_0
\right].$ Then, for every fixed $h>0$ and every integer $t\geq T_\rho$,
\begin{equation}
\mathbb P\left\{
\sup_{\bw\in\mathcal D}\sup_{y\in\mathbb R}
|A_{t,h}(y,\bw)-A_h(y,\bw)|
\vee
\sup_{\bw\in\mathcal D}
|B_{t,h}(\bw)-B_h(\bw)|
\leq
R_\delta\sqrt{\frac{\log t}{t}}
\right\}
\geq1-\delta.
\label{eq:ksa-concentration-probability}
\end{equation}
\end{lem}

\begin{proof}
The two quantities in
\eqref{eq:ksa-concentration-probability} are kernel-weighted empirical
processes indexed by the kernel center $\bw$ and, for the numerator, by
the threshold $y$. The proof proceeds in four steps. First, we identify
a VC class that contains the indicator functions associated with the
kernel superlevel sets. Second, we use its covering numbers to derive a
maximal inequality for the corresponding indicator process under
independent sampling. Third, the layer-cake representation transfers
this bound to the kernel-weighted processes. Finally, we partition the
dependent observations into well-separated subsequences and apply
Lemma~\ref{lem:ksa-separated-comparison} to transfer the
independent-sample bound to the original dependent sample.

\medskip
\noindent\textbf{Step 1:}
Let $\mathcal B_d$ denote the class consisting of the empty set and all open or closed Euclidean balls in $\mathbb R^d$, and define
$\mathcal E_d
=
\left\{
B\times(-\infty,y]:
B\in\mathcal B_d,\quad
y\in\mathbb R\cup\{\infty\}
\right\}.
$ The relevance of this class to the kernel representation follows from radial monotonicity. Specifically, for every $\bw\in\mathcal D$ and $s\in(0,\overline K)$, the superlevel set
$\{\bw':\kappa_h(\bw',\bw)>s\}$ is either empty or an open or closed Euclidean ball. Hence, for every $y\in\mathbb R\cup\{\infty\}$,
$\{(\bw',y'):\kappa_h(\bw',\bw)>s,\ y'\leq y\}
\in\mathcal E_d.$ Setting $y=\infty$ also yields the indicators required for the denominator. Therefore, a uniform empirical-process bound over $\mathcal E_d$ simultaneously controls the kernel-weighted averages appearing in both the numerator and the denominator.

We next show that $\operatorname{VCdim}(\mathcal E_d)\leq d+2.$ Since
$\operatorname{VCdim}(\mathcal B_d)=d+1$, it suffices to show
that imposing the additional lower-half-line restriction can increase the
VC dimension by at most one. Suppose that the points
$x_i=(\bw_i,y_i)$, $1\leq i\leq n$, are shattered by $\mathcal E_d$,
and choose $i^\star$ such that
$y_{i^\star}=\max_{1\leq i\leq n}y_i.$
First, the feature vectors $\bw_1,\ldots,\bw_n$ must be distinct. Indeed,
if $\bw_i=\bw_j$ and $y_i\leq y_j$, then any set of the form
$B\times(-\infty,c]$ that contains $x_j$ must also contain $x_i$,
contradicting the assumption that the points are shattered. Next, fix any
$I\subseteq\{1,\ldots,n\}\setminus\{i^\star\}$. Since the points are
shattered by $\mathcal E_d$, there exist $B_I\in\mathcal B_d$ and
$c_I\in\mathbb R\cup\{\infty\}$ such that the trace of
$B_I\times(-\infty,c_I]$ on $\{x_1,\ldots,x_n\}$ is
$\{x_{i^\star}\}\cup\{x_i:i\in I\}.$ Because this trace contains $x_{i^\star}$, we must have
$c_I\geq y_{i^\star}$. By the choice of $i^\star$, this implies
$c_I\geq y_i$ for every $i$. Hence the lower-half-line restriction no
longer excludes any of the points, and we must have
$B_I\cap\{\bw_i:i\neq i^\star\}
=
\{\bw_i:i\in I\}.$ Since this holds for every
$I\subseteq\{1,\ldots,n\}\setminus\{i^\star\}$, the class
$\mathcal B_d$ shatters the $n-1$ feature vectors
$\{\bw_i:i\neq i^\star\}$. Therefore
$n-1\leq\operatorname{VCdim}(\mathcal B_d)=d+1,$ and hence $n\leq d+2$, which proves the claim. Applying \eqref{eq:EC-lemma_eq} with $v=d+2$ gives
\begin{equation}
N\left(
\eta,
\{\bbone_E:E\in\mathcal E_d\},
L_2(P_0)
\right)
\leq
e(d+3)
\left(\frac{2e}{\eta^2}\right)^{d+2}
\label{eq:ksa-explicit-covering}
\end{equation}
for every
probability measure $P_0$ on $\mathbb R^d\times\mathbb R$ and every
$0<\eta<1$,

\medskip
\noindent\textbf{Step 2:}
Let $P$ denote the common marginal law of $(\bW_s,Y_s)$, which does not
depend on $s$ by stationarity. Fix an integer $m\geq1$, let
$X_i^*=(\bW_i^*,Y_i^*)$, $1\leq i\leq m$ be independent observations
with law $P$, and write
$P_m=\frac1m\sum_{i=1}^m\delta_{X_i^*}$
for their empirical measure. Let $\xi_1,\ldots,\xi_m$ be Rademacher
variables independent of the observations, and let $\mathbb E_\xi$ denote
conditional expectation with respect to these variables given
$X_1^*,\ldots,X_m^*$.

The goal of this step is to establish the independent-sample maximal
inequality that holds over $(\bw, y)\in \mathcal{D}\times \mathbb{R}$. In view of the indicator
representation identified in the previous step, we first prove
\begin{equation}\label{eq:desired_step_2}
\mathbb E\left[
\sup_{E\in\mathcal E_d}
\left|
\frac1m\sum_{i=1}^m
\left\{
\bbone\{X_i^*\in E\}
-\mathbb P(X_1^*\in E)
\right\}
\right|
\right]
\leq
\frac{D_d}{\sqrt m}.
\end{equation}
We obtain this bound by first controlling the associated Rademacher
process through chaining and then applying symmetrization. Start with a finite-class increment bound that will be applied at each
level of the chaining argument. Fix $\lambda>0$ and take
$f,g\in\{\pm\bbone_E:E\in\mathcal E_d\}\cup\{0\}.$ For every $u\in\mathbb R$,
$\mathbb E_\xi\bigl[e^{u\xi_i}\bigr]
=
\frac{e^u+e^{-u}}{2}
=:
\cosh(u).$ Since
$\cosh(u)\leq e^{u^2/2}$, we obtain
$$
\begin{aligned}
\mathbb E_\xi
\exp\left\{
\frac{\lambda}{m}\sum_{i=1}^m
\xi_i\bigl(f(X_i^*)-g(X_i^*)\bigr)
\right\} &=
\prod_{i=1}^m
\cosh\left(
\frac{\lambda\{f(X_i^*)-g(X_i^*)\}}{m}
\right)  \\                                &\leq
\exp\left\{
\frac{\lambda^2}{2m^2}
\sum_{i=1}^m
\{f(X_i^*)-g(X_i^*)\}^2
\right\} =
\exp\left\{
\frac{\lambda^2\|f-g\|_{P_m,2}^2}{2m}
\right\}.
\end{aligned}
$$

Now consider pairs $(f_\ell,g_\ell)$, $1\leq\ell\leq J$, from the
signed-indicator class, that satisfy
$\|f_\ell-g_\ell\|_{P_m,2}\leq r_0.$
Applying the preceding bound to each of the $J$ pairs, summing the resulting inequalities, and then using Jensen's inequality gives:
$$
\begin{aligned}
\mathbb E_\xi
\max_{1\leq\ell\leq J}
\left|
\frac1m\sum_{i=1}^m
\xi_i\{f_\ell(X_i^*)-g_\ell(X_i^*)\}
\right|  &\leq \frac{1}{\lambda}\log \mathbb E_\xi
\exp\left\{
\lambda\max_{1\leq\ell\leq J}
\left|
\frac1m\sum_{i=1}^m
\xi_i\{f_\ell(X_i^*)-g_\ell(X_i^*)\}
\right|
\right\}\\
&\leq
\frac{\log(2J)}{\lambda}
+
\frac{\lambda r_0^2}{2m}.
\end{aligned}
$$
The right-hand side is minimized by
$\lambda=\sqrt{2m\log(2J)}/r_0$, which gives
$$
\mathbb E_\xi
\max_{1\leq\ell\leq J}
\left|
\frac1m\sum_{i=1}^m
\xi_i\{f_\ell(X_i^*)-g_\ell(X_i^*)\}
\right|
\leq
\frac{r_0}{\sqrt m}\sqrt{2\log(2J)}.
$$
We now extend the finite-class increment bound to the full indicator
class by chaining through successively finer covers. Conditional on
$X_1^*,\ldots,X_m^*$, write
$Z_m(f)
=
\frac1m\sum_{i=1}^m\xi_i f(X_i^*)$ and let
$\mathcal F_d
=
\{\pm\bbone_E:E\in\mathcal E_d\}\cup\{0\}.$ The target supremum can be written as
$\sup_{E\in\mathcal E_d}
\left|
Z_m(\bbone_E)
\right|
=
\sup_{f\in\mathcal F_d}
|Z_m(f)|.$ For every integer $j\geq1$, apply
\eqref{eq:ksa-explicit-covering} with $P_0=P_m$ and
$\eta=2^{-j}$. The positive-indicator class therefore admits a
$2^{-j}$-net with at most
$N_j
=
e(d+3)\left(2e\,4^j\right)^{d+2}
$ elements. Negating the elements of this net gives a net of the same
size for the negative indicators. After also including the zero
function, we obtain a $2^{-j}$-net $\mathcal F_j$ for
$\mathcal F_d$ satisfying
$|\mathcal F_j|
\leq
2N_j+1
\leq
3N_j.$ Set $\mathcal F_0=\{0\}$. For every $f\in\mathcal F_d$, choose
$f_j\in\mathcal F_j$ such that
$\|f-f_j\|_{P_m,2}\leq2^{-j}$,
and set $f_0=0$. Since $\|f-f_0\|_{P_m,2}\leq1$, the triangle
inequality gives, for every $j\geq1$,
$$
\|f_j-f_{j-1}\|_{P_m,2}
\leq
\|f_j-f\|_{P_m,2}
+
\|f-f_{j-1}\|_{P_m,2}\leq
2^{-j}+2^{-(j-1)}
=
3\cdot2^{-j}.
$$
At level $j$, let
$\mathcal P_j
=
\{(f_j,f_{j-1}):f\in\mathcal F_d\}$ denote the collection of successive approximation pairs that can
arise. For $j\geq2$,
$|\mathcal P_j|
\leq
|\mathcal F_j|\,|\mathcal F_{j-1}|
\leq
9N_jN_{j-1}
\leq
9N_j^2.$ For $j=1$, since $\mathcal F_0=\{0\}$,
$|\mathcal P_1|
\leq
|\mathcal F_1|
\leq
3N_1
\leq
9N_1^2.$ Thus, $|\mathcal P_j|\leq9N_j^2$ for every $j\geq1$.
By the Cauchy--Schwarz inequality,
$$
\begin{aligned}
|Z_m(f)-Z_m(f_j)|
&=
\left|
\frac1m\sum_{i=1}^m
\xi_i\{f(X_i^*)-f_j(X_i^*)\}
\right| \\
&\leq
\left(\frac1m\sum_{i=1}^m\xi_i^2\right)^{1/2}
\|f-f_j\|_{P_m,2} \,\,\leq
2^{-j}
\longrightarrow0.
\end{aligned}
$$
It follows that
$Z_m(f)
=
\sum_{j=1}^\infty
\{Z_m(f_j)-Z_m(f_{j-1})\}.$ Taking absolute values and then the supremum over $f\in\mathcal F_d$
gives the key chaining inequality
$$
\sup_{f\in\mathcal F_d}|Z_m(f)|
\leq
\sum_{j=1}^\infty
\max_{(u,v)\in\mathcal P_j}
|Z_m(u)-Z_m(v)|.
$$
Every pair in $\mathcal P_j$ has $L_2(P_m)$ distance at most
$3\cdot2^{-j}$, and $|\mathcal P_j|\leq9N_j^2$. The preceding
finite-class inequality therefore gives
$$
\mathbb E_\xi
\max_{(u,v)\in\mathcal P_j}
|Z_m(u)-Z_m(v)|\leq
\frac{3\cdot2^{-j}}{\sqrt m}
\sqrt{2\log\bigl(2|\mathcal P_j|\bigr)} \leq
\frac{3\sqrt2}{\sqrt m}
2^{-j}\sqrt{\log(18N_j^2)}.
$$
Summing this bound over $j$ yields
$$
\mathbb E_\xi
\sup_{E\in\mathcal E_d}
\left|
\frac1m\sum_{i=1}^m
\xi_i\bbone\{X_i^*\in E\}
\right|
\leq
\frac{3\sqrt2}{\sqrt m}
\sum_{j=1}^\infty
2^{-j}\sqrt{\log(18N_j^2)}.
$$
From the definition of $N_j$,
$\log(18N_j^2)
=
\log\bigl(18e^2(d+3)^2\bigr)
+
2(d+2)\log(2e)
+
4j(d+2)\log2.$ Since
$\sum_{j=1}^\infty2^{-j}=1$ and
$\sum_{j=1}^\infty j2^{-j}=2$,
the concavity of the square-root function gives
$$
\sum_{j=1}^\infty
2^{-j}\sqrt{\log(18N_j^2)}
\leq
\sqrt{
\sum_{j=1}^\infty
2^{-j}\log(18N_j^2)
}=
\sqrt{
\log\bigl(18e^2(d+3)^2\bigr)
+
2(d+2)\log(32e)
}.
$$
Recall the definition of $D_d$, thus
\begin{equation}\label{eq:sup_ineq_with_rademacher}
\mathbb E_\xi
\sup_{E\in\mathcal E_d}
\left|
\frac1m\sum_{i=1}^m
\xi_i\bbone\{X_i^*\in E\}
\right|
\leq
\frac{D_d}{2\sqrt m}.
\end{equation}
We
now convert inequality \eqref{eq:sup_ineq_with_rademacher} into the required centered empirical-process bound by
symmetrization. Let $X_1',\ldots,X_m'$ be an independent copy of
$X_1^*,\ldots,X_m^*$. For $E\in\mathcal E_d$, write
$P(E)=\mathbb P(X_1^*\in E).$ Conditional Jensen's inequality replaces $P(E)$ by the empirical average
over the copied sample. The exchangeability of each pair
$(X_i^*,X_i')$ then permits the introduction of the Rademacher signs.
Finally, the triangle inequality and the identical distributions of the
two samples give
$$
\begin{aligned}
\mathbb E\sup_{E\in\mathcal E_d}
\left|
\frac1m\sum_{i=1}^m
\{\bbone_E(X_i^*)-P(E)\}
\right|&\leq
\mathbb E\sup_{E\in\mathcal E_d}
\left|
\frac1m\sum_{i=1}^m
\{\bbone_E(X_i^*)-\bbone_E(X_i')\}
\right|=
\mathbb E\sup_{E\in\mathcal E_d}
\left|
\frac1m\sum_{i=1}^m
\xi_i\{\bbone_E(X_i^*)-\bbone_E(X_i')\}
\right|\\
&\leq
2\mathbb E\sup_{E\in\mathcal E_d}
\left|
\frac1m\sum_{i=1}^m
\xi_i\bbone_E(X_i^*)
\right|\leq
\frac{D_d}{\sqrt m}.
\end{aligned}
$$
This proves the desired inequality \eqref{eq:desired_step_2}.

\medskip
\noindent\textbf{Step 3:} We transfer the bound \eqref{eq:desired_step_2} from indicators to kernel-weighted
averages. For every $(\bw',y')\in\mathbb R^d\times\mathbb R$, the
``layer-cake'' identities give
$\kappa_h(\bw',\bw)=
\int_0^{\overline K}
\bbone\{\kappa_h(\bw',\bw)>s\}\td s$ and $
\kappa_h(\bw',\bw)\bbone\{y'\leq y\}=
\int_0^{\overline K}
\bbone\{\kappa_h(\bw',\bw)>s,\ y'\leq y\}\td s.$ Since every indicator appearing under these integrals belongs to
$\mathcal E_d$, and
$
\left|
\int_0^{\overline K}g(s)\td s
\right|
\leq
\int_0^{\overline K}|g(s)|\td s,$
we may apply the indicator bound \eqref{eq:desired_step_2} at each level $s$ and then
integrate over $(0,\overline K)$. This gives
\begin{equation}
\begin{aligned}
\mathbb E\Bigg[&
\sup_{\bw\in\mathcal D}\sup_{y\in\mathbb R}
\left|
\frac1m\sum_{i=1}^m
\left\{
\kappa_h(\bW_i^*,\bw)\bbone\{Y_i^*\leq y\}
-\mathbb E\bigl[
\kappa_h(\bW_1,\bw)\bbone\{Y_1\leq y\}
\bigr]
\right\}
\right|                                                     \\
&\quad\vee
\sup_{\bw\in\mathcal D}
\left|
\frac1m\sum_{i=1}^m
\left\{
\kappa_h(\bW_i^*,\bw)
-\mathbb E\kappa_h(\bW_1,\bw)
\right\}
\right|
\Bigg]
\leq
\frac{\overline K D_d}{\sqrt m}.
\end{aligned}
\label{eq:ksa-iid-maximal}
\end{equation}

\medskip
\noindent\textbf{Step 4:}
The final step is to transfer the independent-sample bound
in \eqref{eq:ksa-iid-maximal} to the original $\beta$-mixing observations
and thereby prove \eqref{eq:ksa-concentration-probability}. Let
$\Delta_t$ denote the maximum of the two empirical deviations appearing
in the lemma:
$$
\begin{aligned}
\Delta_t
=\;&
\sup_{\bw\in\mathcal D}\sup_{y\in\mathbb R}
\left|
\frac1t\sum_{s=1}^t
\left\{
\kappa_h(\bW_s,\bw)\bbone\{Y_s\leq y\}
-\mathbb E\bigl[
\kappa_h(\bW_1,\bw)\bbone\{Y_1\leq y\}
\bigr]
\right\}
\right|                                                     \\
&\vee
\sup_{\bw\in\mathcal D}
\left|
\frac1t\sum_{s=1}^t
\left\{
\kappa_h(\bW_s,\bw)
-\mathbb E\kappa_h(\bW_1,\bw)
\right\}
\right| \\
&=\sup_{\bw\in\mathcal D}\sup_{y\in\mathbb R}
|A_{t,h}(y,\bw)-A_h(y,\bw)|
\vee
\sup_{\bw\in\mathcal D}
|B_{t,h}(\bw)-B_h(\bw)|.
\end{aligned}
$$
It therefore suffices to show that
$\mathbb E\Delta_t
\leq
\delta R_\delta\sqrt{\frac{\log t}{t}}$
because the required probability bound will then follow immediately
from Markov's inequality. The bound in \eqref{eq:ksa-iid-maximal} cannot be applied directly,
because the observations are dependent.

We first derive a comparison
bound for an arbitrary separation length. Fix any integer
$q\in\{1,\ldots,t\}$ and partition $\{1,\ldots,t\}$ into the residue
classes
$$
I_\ell
=
\{\ell,\ell+q,\ell+2q,\ldots\}\cap\{1,\ldots,t\},
\qquad
n_\ell=|I_\ell| = \left\lfloor\dfrac{t - \ell}{q}\right\rfloor + 1,
\qquad
1\leq\ell\leq q.
$$
These sets form a partition of $\{1,\ldots,t\}$, and consecutive
observations within each $I_\ell$ are separated by $q$ time points.
For each $\ell\in\{1,\ldots,q\}$, define the corresponding
subsequence deviation by
$$
\begin{aligned}
\Delta_\ell
:=\;&
\sup_{\bw\in\mathcal D}\sup_{y\in\mathbb R}
\left|
\frac1{n_\ell}\sum_{s\in I_\ell}
\left\{
\kappa_h(\bW_s,\bw)\bbone\{Y_s\leq y\}
-
\mathbb E\left[
\kappa_h(\bW_1,\bw)\bbone\{Y_1\leq y\}
\right]
\right\}
\right|                                                     \\
&\vee
\sup_{\bw\in\mathcal D}
\left|
\frac1{n_\ell}\sum_{s\in I_\ell}
\left\{
\kappa_h(\bW_s,\bw)
-
\mathbb E\kappa_h(\bW_1,\bw)
\right\}
\right|.
\end{aligned}
$$
Because the sets $I_1,\ldots,I_q$ partition $\{1,\ldots,t\}$ and
$\sum_{\ell=1}^q n_\ell=t$, each centered full-sample average is the
weighted sum of its corresponding subsequence averages, with weights
$n_\ell/t$. Applying the triangle inequality and then taking the
corresponding suprema shows that each of the two terms defining
$\Delta_t$ is bounded by
$\sum_{\ell=1}^q(n_\ell/t)\Delta_\ell$. Therefore,
$\Delta_t
\leq
\sum_{\ell=1}^q\frac{n_\ell}{t}\Delta_\ell.$

We next compare each separated subsequence with an independent sample.
For a fixed $\ell$, let
$(\bW_{\ell,1}^*,Y_{\ell,1}^*),\ldots,
(\bW_{\ell,n_\ell}^*,Y_{\ell,n_\ell}^*)$ be independent observations with the common marginal law $P$, and let
$\Delta_\ell^*$ be obtained from $\Delta_\ell$ by replacing the
observations indexed by $I_\ell$ with these independent copies. Applying
\eqref{eq:ksa-iid-maximal} with $m=n_\ell$ gives
$\mathbb E\Delta_\ell^*
\leq
\frac{\overline K D_d}{\sqrt{n_\ell}}.$ Both $\Delta_\ell$ and $\Delta_\ell^*$ take values in
$[0,\overline K]$, since they are maxima of absolute differences between
averages of variables in $[0,\overline K]$ and their expectations.
Moreover, consecutive observations in $I_\ell$ are separated by $q$.
Lemma~\ref{lem:ksa-separated-comparison}, applied to the corresponding
singleton blocks, therefore yields
$$
\mathbb E\Delta_\ell\leq
\mathbb E\Delta_\ell^*
+\overline K(n_\ell-1)\beta(q)\leq
\frac{\overline K D_d}{\sqrt{n_\ell}}
+\overline K(n_\ell-1)\beta(q).
$$
We now aggregate these bounds over the residue classes. Since
$\sum_{\ell=1}^q n_\ell=t$, Cauchy--Schwarz gives
$\sum_{\ell=1}^q\sqrt{n_\ell}
\leq
\sqrt{q\sum_{\ell=1}^q n_\ell}
=
\sqrt{qt}.$ Also, because the indices in each residue class are separated by $q$,
$n_\ell-1\leq\frac{t}{q}.$ Consequently,
$$
\mathbb E\Delta_t
\leq
\sum_{\ell=1}^q\frac{n_\ell}{t}\mathbb E\Delta_\ell \leq
\frac{\overline K D_d}{t}
\sum_{\ell=1}^q\sqrt{n_\ell}
+
\frac{\overline K\beta(q)}{t}
\sum_{\ell=1}^q n_\ell(n_\ell-1)\leq
\overline K D_d\sqrt{\frac{q}{t}}
+
\overline K\beta(q)\frac{t}{q}.
$$
This inequality displays the role of the separation length $q$. Increasing
$q$ reduces the mixing term through $\beta(q)$, but also decreases the
effective size of each subsequence and therefore enlarges the
independent-sample term. Under geometric mixing, a logarithmic separation
is sufficient to control these two effects.

We now make this choice precise. Set
$q_t
=
\left\lceil
\frac{2\log t}{\log(1/\rho)}
\right\rceil.$ Before substituting $q=q_t$, we verify that this is an admissible
separation length. Since $t
\geq
\left(1+\frac{2}{\log(1/\rho)}\right)^2,$ we have $1\leq\log t\leq\sqrt t$ and
$$
q_t\leq
\frac{2\log t}{\log(1/\rho)}+1           \leq
\left(1+\frac{2}{\log(1/\rho)}\right)\log t              \leq
\left(1+\frac{2}{\log(1/\rho)}\right)\sqrt t
\leq t.
$$
Thus $q_t\in\{1,\ldots,t\}$. Moreover, geometric mixing gives
$\beta(q_t)
\leq
\beta_0\rho^{q_t}
\leq
\beta_0t^{-2}.
$ Substituting $q=q_t$ into the preceding general bound therefore yields
$$
\begin{aligned}
\mathbb E\Delta_t
&\leq
\overline K D_d\sqrt{\frac{q_t}{t}}
+
\overline K\beta(q_t)\frac{t}{q_t}\leq
\overline K D_d
\sqrt{
\left(1+\frac{2}{\log(1/\rho)}\right)
\frac{\log t}{t}
}
+
\frac{\overline K\beta_0}{q_t t}                           \\
&\leq
\overline K
\left[
D_d\sqrt{1+\frac{2}{\log(1/\rho)}}
+\beta_0
\right]
\sqrt{\frac{\log t}{t}}=
\delta R_\delta\sqrt{\frac{\log t}{t}}.
\end{aligned}
$$
The final inequality uses $q_t\geq1$ and
$\frac1t\leq\sqrt{\frac{\log t}{t}}$ for $t\geq 3$. Finally, Markov's inequality gives
$\mathbb P\left\{
\Delta_t>
R_\delta\sqrt{\frac{\log t}{t}}
\right\}
\leq
\frac{\mathbb E\Delta_t}
{R_\delta\sqrt{\log t/t}}
\leq
\delta.
$ Since $\Delta_t$ is exactly the maximum of the two deviations in
\eqref{eq:ksa-concentration-probability}, this proves the result.

\end{proof}

We are now ready to give a more explicit restatement of Theorem \ref{thm:individual_cali_thm} and provide its proof. First we introduce the new notations needed. Let
$\mu_0
:=
\int_{\mathbb R^d}K(\mathbf v)\td\mathbf v$, $\mu_2
:=
\int_{\mathbb R^d}\|\mathbf v\|^2K(\mathbf v)\td\mathbf v$, and $v_d
:=
\frac{\pi^{d/2}}{\Gamma(1+d/2)}.$ Assumption~\ref{assump:individual_cali_assump}\ref{itm:kernel_regularity} gives $0<\mu_0<\infty$ and
$0<\mu_2<\infty$. Also write
$\underline\pi
:=
\inf_{\bw\in\mathcal D^+}\pi_{\bW}(\bw)$ and $\overline\pi
:=
\sup_{\bw\in\mathcal D^+}\pi_{\bW}(\bw).$

\begin{thm}[Restatement of Theorem \ref{thm:individual_cali_thm}]\label{apxthm:explicit_statement_individual}
Under Assumption~\ref{assump:individual_cali_assump}, for any $h>0$ and $\delta\in(0,1)$, define
$C_{h,\delta}
:=
\frac{16R_\delta}{\underline\pi\mu_0h^d}$ and $T_{h,\delta}:=
\left\lceil
\max\left\{
T_\rho,
\left(
\frac{4R_\delta}{\underline\pi\mu_0h^d}
\right)^4
\right\}
\right\rceil.$ Then, for every
$\epsilon>0$, there exists $h_\epsilon>0$ such that, for every fixed
$h\in(0,h_\epsilon]$, $\delta\in(0,1)$, and
$t\geq T_{h,\delta}$,
$$
\mathbb P\left(
\sup_{\bz:\,\bW_{t+1}(\bz)\in\mathcal D}
\left|
\mathbb P\left(
G_{t+1}\notin\Ct_t(\bz)
\mid\mathcal F_t,\bz_{t+1}=\bz
\right)-\alpha
\right|\leq
\epsilon
+C_{h,\delta}\sqrt{\frac{\log t}{t}}
\right)
\geq1-\delta.
$$
The outer probability is over the historical data that generate
$\mathcal F_t$.
\end{thm}

\begin{proof}

The proof has three steps. First, we show that the population kernel CDF
$F_h$ uniformly approximates $F_{Y\mid\bW}$ and construct the bandwidth
threshold $h_\epsilon$ needed for this approximation. Second,
Lemma~\ref{lem:ksa-fixed-h-concentration} is used to control the empirical
kernel ratio around $F_h$. Finally, we convert the resulting uniform CDF
error into a coverage error at the two estimated quantiles.

\medskip
\noindent\textbf{Step 1:}
In this step we construct $h_\epsilon>0$ such that
$\sup_{\bw\in\mathcal D}\sup_{y\in\mathbb R}
\left|
F_h(y\mid\bw)-F_{Y\mid\bW}(y\mid\bw)
\right|
\leq\frac{\epsilon}{2}
$ whenever $0<h\leq h_\epsilon$. By stationarity and conditioning on $\bW_1$,
$A_h(y,\bw)=
\int
\kappa_h(\mathbf u,\bw)
F_{Y\mid\bW}(y\mid\mathbf u)
\pi_{\bW}(\mathbf u)\td\mathbf u$ and $B_h(\bw)=
\int
\kappa_h(\mathbf u,\bw)
\pi_{\bW}(\mathbf u)\td\mathbf u.$ Consequently,
\begin{equation}
F_h(y\mid\bw)-F_{Y\mid\bW}(y\mid\bw)=
\frac{1}{B_h(\bw)}
\int
\kappa_h(\mathbf u,\bw)
\left[
F_{Y\mid\bW}(y\mid\mathbf u)
-F_{Y\mid\bW}(y\mid\bw)
\right]
\pi_{\bW}(\mathbf u)\td\mathbf u.
\label{eq:ksa-population-bias-representation}
\end{equation}
Thus, we first need a uniform lower bound for the denominator
$B_h(\bw)$ and then an upper bound for the numerator in
\eqref{eq:ksa-population-bias-representation}. For the denominator, the ball
$\{\mathbf u:\|\mathbf u-\bw\|\leq h_0\}$ is contained in
$\mathcal D^+$ when $\bw\in\mathcal D$. Therefore,
$$
B_h(\bw)\geq
\underline\pi
\int_{\|\mathbf u-\bw\|\leq h_0}
\kappa_h(\mathbf u,\bw)\td\mathbf u=
\underline\pi h^d
\int_{\|\mathbf v\|\leq h_0/h}K(\mathbf v)\td\mathbf v =
\underline\pi h^d
\left[
\mu_0-
\int_{\|\mathbf v\|>h_0/h}K(\mathbf v)\td\mathbf v
\right].
$$
By the definition of $\mu_2$,
$$
\int_{\|\mathbf v\|>h_0/h}K(\mathbf v)\td\mathbf v
\leq
\frac{h^2}{h_0^2}
\int_{\|\mathbf v\|>h_0/h}
\|\mathbf v\|^2K(\mathbf v)\td\mathbf v
\leq
\frac{\mu_2h^2}{h_0^2}.
$$
It follows that, whenever
$0<h\leq
\min\left\{
h_0,
h_0\sqrt{\frac{\mu_0}{2\mu_2}}
\right\},$ we have
\begin{equation}
B_h(\bw)
\geq
\underline\pi h^d
\left(
\mu_0-\frac{\mu_2h^2}{h_0^2}
\right)
\geq
\frac{\underline\pi\mu_0}{2}h^d
\label{eq:ksa-denominator-lower}
\end{equation}
for every $\bw\in\mathcal D$. We now control the numerator in
\eqref{eq:ksa-population-bias-representation}. Fix $\epsilon>0$.
By Assumption~\ref{assump:individual_cali_assump}\ref{itm:cond_dist_cont}, we may choose $a\in(0,h_0)$ such that
$\sup_{\substack{\bw,\bw'\in\mathcal D^+\\
\|\bw-\bw'\|\leq a}}
\sup_{y\in\mathbb R}
\left|
F_{Y\mid\bW}(y\mid\bw)
-F_{Y\mid\bW}(y\mid\bw')
\right|
\leq\frac{\epsilon}{4}.$ To use this local continuity while accounting for observations farther
from $\bw$, split the integral in
\eqref{eq:ksa-population-bias-representation} into the three regions: (1) $
\|\mathbf u-\bw\|\leq a$, (2) $a<\|\mathbf u-\bw\|\leq h_0$, and (3)
$\|\mathbf u-\bw\|>h_0.$

On the first region, $\mathbf u$ and $\bw$ both lie in $\mathcal D^+$,
and the conditional CDFs differ by at most $\epsilon/4$. Hence, after
division by $B_h(\bw)$, this region contributes at most $\epsilon/4$; On the second region, both points still lie in $\mathcal D^+$, so
$\pi_{\bW}(\mathbf u)\leq\overline\pi$. Since the difference between two
CDF values is at most one, the corresponding part of the numerator is
bounded by
$$
\overline\pi
\int_{a<\|\mathbf u-\bw\|\leq h_0}
\kappa_h(\mathbf u,\bw)\td\mathbf u
\leq\overline\pi h^d
\int_{\|\mathbf v\|>a/h}K(\mathbf v)\td\mathbf v \leq
\frac{\overline\pi\mu_2}{a^2}h^{d+2},
$$
where the final inequality again follows from the definition of
$\mu_2$; And on the third region $\|\mathbf u-\bw\|>h_0$, the density
bound on $\mathcal D^+$ is no longer available there, so instead we use
the pointwise decay implied by radial monotonicity. For every $r>0$, by
$$
\frac{v_d(1-2^{-d})}{4}
r^{d+2}\widetilde\kappa(r)
\leq
\int_{r/2<\|\mathbf v\|\leq r}
\|\mathbf v\|^2K(\mathbf v)\td\mathbf v
\leq\mu_2,
$$
we have
\begin{equation}
\widetilde\kappa(r)
\leq
\frac{4\mu_2}
{v_d(1-2^{-d})r^{d+2}}.
\label{eq:ksa-kernel-point-tail}
\end{equation}
Because $\widetilde\kappa$ is nonincreasing and $\pi_{\bW}$ integrates
to one, the contribution of the third region to the numerator is at
most
$$
\int_{\|\mathbf u-\bw\|>h_0}
\kappa_h(\mathbf u,\bw)\pi_{\bW}(\mathbf u)\td\mathbf u
\leq
\widetilde\kappa(h_0/h)\leq
\frac{4\mu_2}
{v_d(1-2^{-d})h_0^{d+2}}h^{d+2}.
$$
Combining the bounds for the three regions and using
\eqref{eq:ksa-denominator-lower}, we obtain
$$
\sup_{\bw\in\mathcal D}\sup_{y\in\mathbb R}
\left|
F_h(y\mid\bw)-F_{Y\mid\bW}(y\mid\bw)
\right| \leq
\frac{\epsilon}{4}
+\frac{2}{\underline\pi\mu_0}
\left(
\frac{\overline\pi\mu_2}{a^2}
+\frac{4\mu_2}
{v_d(1-2^{-d})h_0^{d+2}}
\right)h^2,
$$
where the first term is the local approximation error, while the second
collects the kernel contributions from observations more than $a$ away
from $\bw$. Define
\begin{equation}\label{eq:def_h_epsilon}
h_\epsilon
:=
\min\left\{
h_0,\quad
h_0\sqrt{\frac{\mu_0}{2\mu_2}},\quad
\left[\frac{8}{\epsilon\,\underline\pi\mu_0}\left(
\frac{\overline\pi\mu_2}{a^2}
+\frac{4\mu_2}
{v_d(1-2^{-d})h_0^{d+2}}
\right)
\right]^{-1/2}
\right\}.
\end{equation}
This choice is strictly positive and ensures that, for every
$0<h\leq h_\epsilon$,
\begin{equation}
\sup_{\bw\in\mathcal D}\sup_{y\in\mathbb R}
\left|
F_h(y\mid\bw)-F_{Y\mid\bW}(y\mid\bw)
\right|
\leq\frac{\epsilon}{2}.
\label{eq:ksa-bias-epsilon-half}
\end{equation}

\medskip
\noindent\textbf{Step 2:} We
now control the difference between the empirical CDF $\widehat F_t$ and
its population counterpart $F_h$. Since
$\widehat F_t(y\mid\bw)
=
\frac{A_{t,h}(y,\bw)}{B_{t,h}(\bw)}$ and $F_h(y\mid\bw)
=
\frac{A_h(y,\bw)}{B_h(\bw)},$ the first task is to show that the empirical denominator
$B_{t,h}(\bw)$ remains uniformly away from zero. After that,
the concentration bounds for the empirical numerator and denominator can
be converted into a bound for the ratio. Fix $\delta\in(0,1)$ and a bandwidth satisfying
$0<h\leq
\min\left\{
h_0,
h_0\sqrt{\frac{\mu_0}{2\mu_2}}
\right\}$ (this condition is automatically satisfied
whenever $h\leq h_\epsilon$ from \eqref{eq:def_h_epsilon}). Also fix $t\geq T_{h,\delta}$. Since
$T_{h,\delta}\geq T_\rho$, Lemma~\ref{lem:ksa-fixed-h-concentration}
implies that, with probability at least $1-\delta$,
\begin{equation}
\sup_{\bw\in\mathcal D}\sup_{y\in\mathbb R}
|A_{t,h}(y,\bw)-A_h(y,\bw)|\vee
\sup_{\bw\in\mathcal D}
|B_{t,h}(\bw)-B_h(\bw)|
\leq
R_\delta\sqrt{\frac{\log t}{t}}.
\label{eq:ksa-proof-concentration-event}
\end{equation}
We work on this event for the remainder of the step. The second component in the definition of $T_{h,\delta}$ ensures that
the deviation in \eqref{eq:ksa-proof-concentration-event} is small
relative to the population denominator. Indeed,
$T_{h,\delta}\geq T_\rho\geq3$, so
$\log t\leq\sqrt t$, and therefore
$R_\delta\sqrt{\frac{\log t}{t}}
\leq
R_\delta t^{-1/4}.$ Moreover, the inequality
$t
\geq
\left(
\frac{4R_\delta}{\underline\pi\mu_0h^d}
\right)^4$ implies $
R_\delta t^{-1/4}
\leq
\frac{\underline\pi\mu_0}{4}h^d.$ Combining this bound with
\eqref{eq:ksa-denominator-lower} and
\eqref{eq:ksa-proof-concentration-event} gives
\begin{equation}\label{eq:positive_denominator_step_2}
\begin{aligned}
\inf_{\bw\in\mathcal D}B_{t,h}(\bw)
&\geq
\inf_{\bw\in\mathcal D}B_h(\bw)
-
\sup_{\bw\in\mathcal D}
|B_{t,h}(\bw)-B_h(\bw)|                                      \\
&\geq
\frac{\underline\pi\mu_0}{2}h^d
-
\frac{\underline\pi\mu_0}{4}h^d
=
\frac{\underline\pi\mu_0}{4}h^d
>0.
\end{aligned}
\end{equation}
We can now control the ratio error. For every $\bw\in\mathcal D$ and
$y\in\mathbb R$,
$$
\widehat F_t(y\mid\bw)-F_h(y\mid\bw)=
\frac{A_{t,h}(y,\bw)-A_h(y,\bw)}{B_{t,h}(\bw)}+
F_h(y\mid\bw)
\frac{B_h(\bw)-B_{t,h}(\bw)}{B_{t,h}(\bw)}.
$$
Because $0\leq F_h(y\mid\bw)\leq1$, the two concentration bounds in
\eqref{eq:ksa-proof-concentration-event} and the lower bound for
$B_{t,h}$ give
\begin{equation}
\sup_{\bw\in\mathcal D}\sup_{y\in\mathbb R}
|\widehat F_t(y\mid\bw)-F_h(y\mid\bw)|\leq
\frac{
2R_\delta\sqrt{\log t/t}
}{(\underline\pi\mu_0/4)h^d}=
\frac{8R_\delta}{\underline\pi\mu_0h^d}
\sqrt{\frac{\log t}{t}}.
\label{eq:ksa-ratio-conclusion}
\end{equation}
Finally, for $h\leq h_\epsilon$,
\eqref{eq:ksa-error-decomposition},
\eqref{eq:ksa-bias-epsilon-half}, and
\eqref{eq:ksa-ratio-conclusion} imply
\begin{equation}
\sup_{\bw\in\mathcal D}\sup_{y\in\mathbb R}
|\widehat F_t(y\mid\bw)-F_{Y\mid\bW}(y\mid\bw)| \leq
\frac{\epsilon}{2}
+
\frac{8R_\delta}{\underline\pi\mu_0h^d}
\sqrt{\frac{\log t}{t}} =
\frac{1}{2}
\left(
\epsilon+C_{h,\delta}\sqrt{\frac{\log t}{t}}
\right).
\label{eq:ksa-proof-total-CDF-error}
\end{equation}
The final equality follows from the definition of $C_{h,\delta}$.

\medskip
\noindent\textbf{Step 3:}
The last step is to translate the uniform CDF bound \eqref{eq:ksa-proof-total-CDF-error} into a bound
on conditional coverage. This result follows from two facts: first, uniform
closeness of $\widehat F_t$ and $F_{Y\mid\bW}$ implies that an empirical
quantile has approximately the correct probability under
$F_{Y\mid\bW}$; second, predictive sufficiency of $\bW_{t+1}(\bz)$ on $Y_{t+1}$ identifies these
probabilities with the conditional tail probabilities of $Y_{t+1}$ under
each admissible scenario.

Fix $u\in(0,1)$ and
$\bw\in\mathcal D$, and write
$q=\widehat Q_t(u;\bw).$ By the generalized-inverse convention in
\eqref{eq:ksa-proof-quantile-convention},
$\widehat F_t(q\mid\bw)\geq u$ and $\widehat F_t(x\mid\bw)<u$ for every $x<q.$ The first inequality gives
$$
\begin{aligned}
F_{Y\mid\bW}(q\mid\bw)
&\geq
\widehat F_t(q\mid\bw)
-
\sup_{\substack{\bw'\in\mathcal D\\y\in\mathbb R}}
|\widehat F_t(y\mid\bw')-F_{Y\mid\bW}(y\mid\bw')|            \\
&\geq
u-
\sup_{\substack{\bw'\in\mathcal D\\y\in\mathbb R}}
|\widehat F_t(y\mid\bw')-F_{Y\mid\bW}(y\mid\bw')|.
\end{aligned}
$$
Similarly, for every $x<q$,
$F_{Y\mid\bW}(x\mid\bw)
\leq
u+
\sup_{\substack{\bw'\in\mathcal D\\y\in\mathbb R}}
|\widehat F_t(y\mid\bw')-F_{Y\mid\bW}(y\mid\bw')|.$ Letting $x\uparrow q$ and using the continuity in
Assumption~\ref{assump:individual_cali_assump}\ref{itm:cond_dist_cont} yields
\begin{equation}
\sup_{\bw\in\mathcal D}
\left|
F_{Y\mid\bW}(\widehat Q_t(u;\bw)\mid\bw)-u
\right|
\leq
\sup_{\bw\in\mathcal D}\sup_{y\in\mathbb R}
|\widehat F_t(y\mid\bw)-F_{Y\mid\bW}(y\mid\bw)|.
\label{eq:ksa-quantile-probability-error}
\end{equation}
We next apply this result to the two endpoints of the KSA interval. Fix
an admissible scenario $\bz$, write
$\bw=\bW_{t+1}(\bz)$, $q_-=\widehat Q_t(\alpha/2;\bw)$, and $q_+=\widehat Q_t(1-\alpha/2;\bw).$ The positivity of the empirical denominator established in \eqref{eq:positive_denominator_step_2}
ensures that these quantiles are well defined. They are
$\mathcal F_t$-measurable, and monotonicity of the generalized inverse
gives $q_-\leq q_+$. By Assumption~\ref{assump:individual_cali_assump}
\ref{itm:independence_n_homogeneity}, the predictive-sufficiency
identity holds almost surely, simultaneously for all thresholds
and admissible scenarios. Since $q_-$ and $q_+$ are
$\mathcal F_t$-measurable, we may therefore evaluate this identity
at these two thresholds.
$$
\mathbb P\left(
Y_{t+1}<q_-
\mid\mathcal F_t,\bz_{t+1}=\bz
\right)
=
F_{Y\mid\bW}(q_-\mid\bw), \quad \mathbb P\left(
Y_{t+1}\leq q_+
\mid\mathcal F_t,\bz_{t+1}=\bz
\right)
=
F_{Y\mid\bW}(q_+\mid\bw).
$$
Under the same conditional law,
$Y_{t+1}=G_{t+1}-\varphi(\bw)$. Since the endpoints of
$\Ct_t(\bz)$ are $\varphi(\bw)+q_-$ and
$\varphi(\bw)+q_+$, respectively,
$$
G_{t+1}\notin\Ct_t(\bz)
\quad\Longleftrightarrow\quad
Y_{t+1}<q_-
\quad\text{or}\quad
Y_{t+1}>q_+.
$$
Consequently,
$\mathbb P\left(
G_{t+1}\notin\Ct_t(\bz)
\mid\mathcal F_t,\bz_{t+1}=\bz
\right)=
F_{Y\mid\bW}(q_-\mid\bw)
+1-F_{Y\mid\bW}(q_+\mid\bw).$ Subtracting
$\alpha=\alpha/2+\{1-(1-\alpha/2)\}$ and applying the triangle
inequality gives
$$
\left|
\mathbb P\left(
G_{t+1}\notin\Ct_t(\bz)
\mid\mathcal F_t,\bz_{t+1}=\bz
\right)-\alpha
\right|\leq
\left|
F_{Y\mid\bW}(q_-\mid\bw)-\frac{\alpha}{2}
\right|
+
\left|
F_{Y\mid\bW}(q_+\mid\bw)-\left(1-\frac{\alpha}{2}\right)
\right|.
$$
Applying \eqref{eq:ksa-quantile-probability-error} at the two quantile
levels and taking the supremum over all admissible scenarios yields
\begin{equation}
\sup_{\bz:\,\bW_{t+1}(\bz)\in\mathcal D}
\left|
\mathbb P\left(
G_{t+1}\notin\Ct_t(\bz)
\mid\mathcal F_t,\bz_{t+1}=\bz
\right)-\alpha
\right|\leq
2\sup_{\bw\in\mathcal D}\sup_{y\in\mathbb R}
|\widehat F_t(y\mid\bw)-F_{Y\mid\bW}(y\mid\bw)|.
\label{eq:ksa-coverage-from-CDF}
\end{equation}
The factor $2$ appears because the interval has two endpoints, each of
which contributes at most the uniform CDF error. Finally, combining (\ref{eq:ksa-proof-total-CDF-error}) and
(\ref{eq:ksa-coverage-from-CDF}) gives, on the event from
Lemma~\ref{lem:ksa-fixed-h-concentration},
$$
\sup_{\bz:\,\bW_{t+1}(\bz)\in\mathcal D}
\left|
\mathbb P\left(
G_{t+1}\notin\Ct_t(\bz)
\mid\mathcal F_t,\bz_{t+1}=\bz
\right)-\alpha
\right|
\leq
\epsilon
+C_{h,\delta}\sqrt{\frac{\log t}{t}}.
$$
This event has probability at least $1-\delta$, which proves the
probability-quantified restatement. Since $\delta\in(0,1)$ is arbitrary
and $C_{h,\delta}<\infty$ for every fixed $h>0$, the same result yields
the $O_p(\sqrt{\log t/t})$ conclusion in
\eqref{eq:ksa-conditional-coverage}.
\end{proof}

The next result makes the two bandwidth orders in
Corollary~\ref{cor:ksa-smoothness} explicit.

\begin{cor}[Restatement of Corollary~\ref{cor:ksa-smoothness}]
\label{apxcor:explicit_ksa_smoothness}
Suppose Assumption~\ref{assump:individual_cali_assump} holds. For
$h>0$ and $\delta\in(0,1)$, define $C_{h,\delta}$ and $T_{h,\delta}$ as in Theorem \ref{apxthm:explicit_statement_individual}. Further define
$\mu_1
:=
\int_{\mathbb R^d}\|\mathbf v\|K(\mathbf v)\td\mathbf v$, $\bar h
:=
\min\left\{
h_0,
h_0\sqrt{\frac{\mu_0}{2\mu_2}}
\right\}$, and $R_0
:=
\frac{4\mu_2}
{v_d(1-2^{-d})h_0^{d+2}}.$ Then the following statements hold.
\begin{enumerate}[label=(\roman*)]
\item\label{apxitm:lipschitz} Suppose that, for some $L<\infty$,
$$
|F_{Y\mid\bW}(y\mid\bw)-F_{Y\mid\bW}(y\mid\bw')|
\leq L\|\bw-\bw'\|,
\qquad
\bw,\bw'\in\mathcal D^+,
\quad y\in\mathbb R.
$$
Let $
D_L
:=
\frac{4}{\underline\pi\mu_0}
\left(
\overline\pi L\mu_1+R_0\bar h
\right).$ For every $h\in(0,\bar h]$, $\delta\in(0,1)$, and $t\geq T_{h,\delta}$,
$$
\mathbb P\left(
\sup_{\bz:\,\bW_{t+1}(\bz)\in\mathcal D}
\left|
\mathbb P\left(
G_{t+1}\notin\Ct_t(\bz)
\mid\mathcal F_t,\bz_{t+1}=\bz
\right)-\alpha
\right|
\leq
D_Lh+C_{h,\delta}\sqrt{\frac{\log t}{t}}
\right)
\geq1-\delta.
$$

\item\label{apxitm:twice_diff} Suppose $\pi_{\bW}$ is twice continuously differentiable and, for
every $y\in\mathbb R$,
$\pi_{\bW}(\cdot)F_{Y\mid\bW}(y\mid\cdot)$ is twice continuously
differentiable on an open neighborhood of $\mathcal D^+$. Assume
$M_\pi
:=
\sup_{\bw\in\mathcal D^+}
\|\nabla_{\bw}^2\pi_{\bW}(\bw)\|_{\mathrm{op}}$ and $M_{\pi F}
:=
\sup_{\substack{\bw\in\mathcal D^+\\y\in\mathbb R}}
\left\|
\nabla_{\bw}^2
\left[
\pi_{\bW}(\bw)F_{Y\mid\bW}(y\mid\bw)
\right]
\right\|_{\mathrm{op}}$ are both finite, where $\|\cdot\|_{\mathrm{op}}$ is the operator norm induced by the
Euclidean norm. Define
$D_2
:=
\frac{2(M_{\pi F}+M_\pi)\mu_2+8R_0}
{\underline\pi\mu_0}.$ For every $h\in(0,\bar h]$, $\delta\in(0,1)$, and
$t\geq T_{h,\delta}$,
$$
\mathbb P\left(
\sup_{\bz:\,\bW_{t+1}(\bz)\in\mathcal D}
\left|
\mathbb P\left(
G_{t+1}\notin\Ct_t(\bz)
\mid\mathcal F_t,\bz_{t+1}=\bz
\right)-\alpha
\right|
\leq
D_2h^2+C_{h,\delta}\sqrt{\frac{\log t}{t}}
\right)
\geq1-\delta.
$$
\end{enumerate}
In both statements, the outer probability is over the historical data
that generate $\mathcal F_t$.
\end{cor}

\begin{proof}
The proof refines the population-approximation argument in \textbf{Step~1} of the
proof of Theorem~\ref{apxthm:explicit_statement_individual}. The
empirical-concentration and coverage-conversion arguments remain unchanged. We first derive the two refined population bounds
and then combine them with the remaining arguments.

\medskip
\noindent\textbf{Step 1:}
For every $0<h\leq\bar h$, the denominator bound
\eqref{eq:ksa-denominator-lower} gives
$\inf_{\bw\in\mathcal D}B_h(\bw)
\geq
\frac{\underline\pi\mu_0}{2}h^d.$ We now sharpen the numerator bound under each smoothness condition.

Under part \ref{apxitm:lipschitz}, fix $\bw\in\mathcal D$ and $y\in\mathbb R$.
In the representation \eqref{eq:ksa-population-bias-representation},
split the integral according to whether
$\|\mathbf u-\bw\|\leq h_0$. On this region, both $\mathbf u$ and
$\bw$ lie in $\mathcal D^+$. The Lipschitz condition and the upper bound
on $\pi_{\bW}$ therefore give
$$
\begin{aligned}
&\int_{\|\mathbf u-\bw\|\leq h_0}
\kappa_h(\mathbf u,\bw)
\left|
F_{Y\mid\bW}(y\mid\mathbf u)
-F_{Y\mid\bW}(y\mid\bw)
\right|
\pi_{\bW}(\mathbf u)\td\mathbf u                              \\
&\quad\leq
L\overline\pi
\int_{\|\mathbf u-\bw\|\leq h_0}
\kappa_h(\mathbf u,\bw)\|\mathbf u-\bw\|\td\mathbf u \leq
L\overline\pi\mu_1h^{d+1}.
\end{aligned}
$$
On the complementary region, the difference between the two CDF values
is at most one. Since $\widetilde\kappa$ is nonincreasing and
$\pi_{\bW}$ integrates to one,
$$
\int_{\|\mathbf u-\bw\|>h_0}
\kappa_h(\mathbf u,\bw)
\left|
F_{Y\mid\bW}(y\mid\mathbf u)
-F_{Y\mid\bW}(y\mid\bw)
\right|
\pi_{\bW}(\mathbf u)\td\mathbf u         \leq
\widetilde\kappa(h_0/h)
\leq
R_0h^{d+2},
$$
where the last inequality follows from
\eqref{eq:ksa-kernel-point-tail} and the definition of $R_0$.
Substituting these bounds into
\eqref{eq:ksa-population-bias-representation} and using
\eqref{eq:ksa-denominator-lower}, we obtain
\begin{equation}
\sup_{\bw\in\mathcal D}\sup_{y\in\mathbb R}
\left|F_h(y\mid\bw)-F_{Y\mid\bW}(y\mid\bw)\right|\leq
\frac{2}{\underline\pi\mu_0}
\left(
\overline\pi L\mu_1h+R_0h^2
\right)\leq
\frac{2}{\underline\pi\mu_0}
\left(
\overline\pi L\mu_1+R_0\bar h
\right)h
=
\frac{D_L}{2}h.
\label{eq:ksa-lipschitz-population-bias}
\end{equation}

Under part \ref{apxitm:twice_diff}, define for this part of the proof
$q_y(\mathbf u)
:=
\pi_{\bW}(\mathbf u)F_{Y\mid\bW}(y\mid\mathbf u)$ and $J_h
:=
\int_{\|\mathbf v\|\leq h_0/h}K(\mathbf v)\td\mathbf v.$ Fix $\bw\in\mathcal D$ and $y\in\mathbb R$. If
$\|\mathbf u-\bw\|\leq h_0$, then the entire line segment joining
$\bw$ and $\mathbf u$ lies in $\mathcal D^+$. After the change of
variables $\mathbf u=\bw+h\mathbf v$, Taylor's theorem gives
$$
q_y(\bw+h\mathbf v)
=
q_y(\bw)
+h\nabla_{\bw}q_y(\bw)^\top\mathbf v
+r_y(\mathbf v),
\qquad
|r_y(\mathbf v)|
\leq
\frac{M_{\pi F}}{2}h^2\|\mathbf v\|^2,
\qquad
\|\mathbf v\|\leq\frac{h_0}{h}.
$$
Because $K$ is radial and the ball
$\{\mathbf v:\|\mathbf v\|\leq h_0/h\}$ is symmetric,
$\int_{\|\mathbf v\|\leq h_0/h}
\mathbf vK(\mathbf v)\td\mathbf v
=
\boldsymbol0.$ The linear Taylor term therefore integrates to zero, while the
remainder contributes at most $M_{\pi F}\mu_2h^2/2$. The portion of
$h^{-d}A_h(y,\bw)$ arising from $\|\mathbf u-\bw\|>h_0$ is bounded by
$h^{-d}\widetilde\kappa(h_0/h)
\int_{\|\mathbf u-\bw\|>h_0}q_y(\mathbf u)\td\mathbf u
\leq
R_0h^2,$ because $0\leq q_y\leq\pi_{\bW}$. Consequently,
\begin{equation}
\left|
h^{-d}A_h(y,\bw)-J_hq_y(\bw)
\right|
\leq
\left(
\frac{M_{\pi F}\mu_2}{2}+R_0
\right)h^2.
\label{eq:ksa-second-order-numerator}
\end{equation}
The same argument, with $q_y$ replaced by $\pi_{\bW}$, gives
\begin{equation}
\left|
h^{-d}B_h(\bw)-J_h\pi_{\bW}(\bw)
\right|
\leq
\left(
\frac{M_\pi\mu_2}{2}+R_0
\right)h^2.
\label{eq:ksa-second-order-denominator}
\end{equation}
Moreover,
$q_y(\bw)/\pi_{\bW}(\bw)=F_{Y\mid\bW}(y\mid\bw)\in[0,1]$.
Using $F_h=A_h/B_h$, we may therefore write
$$
\begin{aligned}
\left| F_h(y\mid\bw)-F_{Y\mid\bW}(y\mid\bw)\right|
&=
\Bigg|\frac{h^{-d}A_h(y,\bw)-J_hq_y(\bw)}{h^{-d}B_h(\bw)} +
\frac{q_y(\bw)}{\pi_{\bW}(\bw)}
\frac{J_h\pi_{\bW}(\bw)-h^{-d}B_h(\bw)}
{h^{-d}B_h(\bw)}\Bigg|\\
&\leq \frac{
|h^{-d}A_h(y,\bw)-J_hq_y(\bw)|
+\dfrac{q_y(\bw)}{\pi_{\bW}(\bw)}
 |J_h\pi_{\bW}(\bw)-h^{-d}B_h(\bw)|
}{h^{-d}B_h(\bw)}.
\end{aligned}
$$
Combining \eqref{eq:ksa-denominator-lower}, \eqref{eq:ksa-second-order-numerator} and
\eqref{eq:ksa-second-order-denominator} yields
\begin{equation}
\sup_{\bw\in\mathcal D}\sup_{y\in\mathbb R}
\left|F_h(y\mid\bw)-F_{Y\mid\bW}(y\mid\bw)\right|
\leq
\frac{(M_{\pi F}+M_\pi)\mu_2+4R_0}
{\underline\pi\mu_0}h^2
=
\frac{D_2}{2}h^2.
\label{eq:ksa-second-order-population-bias}
\end{equation}

\medskip
\noindent\textbf{Step 2:} The remaining arguments follow from \textbf{Steps 2--3} in the proof of Theorem \ref{apxthm:explicit_statement_individual}. Fix $\delta\in(0,1)$ and an integer $t\geq T_{h,\delta}$. On the event
in \eqref{eq:ksa-concentration-probability}, the argument leading to
\eqref{eq:positive_denominator_step_2} and
\eqref{eq:ksa-ratio-conclusion} applies because $h\leq\bar h$. Thus,
$$
\sup_{\bw\in\mathcal D}\sup_{y\in\mathbb R}
|\widehat F_t(y\mid\bw)-F_h(y\mid\bw)|
\leq
\frac{8R_\delta}{\underline\pi\mu_0h^d}
\sqrt{\frac{\log t}{t}}
=
\frac{C_{h,\delta}}{2}\sqrt{\frac{\log t}{t}}.
$$
Combining this inequality with
\eqref{eq:ksa-lipschitz-population-bias} gives, under part \ref{apxitm:lipschitz},
\begin{equation}\label{eq:corC4_lipschitz_mid}
\sup_{\bw\in\mathcal D}\sup_{y\in\mathbb R}
|\widehat F_t(y\mid\bw)-F_{Y\mid\bW}(y\mid\bw)|
\leq
\frac12\left(
D_Lh+C_{h,\delta}\sqrt{\frac{\log t}{t}}
\right).
\end{equation}
Similarly, \eqref{eq:ksa-second-order-population-bias} gives, under
part \ref{apxitm:twice_diff},
\begin{equation}\label{eq:corC4_twiceDiff_mid}
\sup_{\bw\in\mathcal D}\sup_{y\in\mathbb R}
|\widehat F_t(y\mid\bw)-F_{Y\mid\bW}(y\mid\bw)|
\leq
\frac12\left(
D_2h^2+C_{h,\delta}\sqrt{\frac{\log t}{t}}
\right).
\end{equation}
By the quantile and predictive-sufficiency argument in \textbf{Step 3} of the
proof of Theorem~\ref{apxthm:explicit_statement_individual},
\eqref{eq:ksa-coverage-from-CDF} gives
$\sup_{\bz:\,\bW_{t+1}(\bz)\in\mathcal D}
\left|
\mathbb P\left(
G_{t+1}\notin\Ct_t(\bz)
\mid\mathcal F_t,\bz_{t+1}=\bz
\right)-\alpha
\right| \leq
2\sup_{\bw\in\mathcal D}\sup_{y\in\mathbb R}
|\widehat F_t(y\mid\bw)-F_{Y\mid\bW}(y\mid\bw)|.$ Substituting \eqref{eq:corC4_lipschitz_mid} yields
$D_Lh+C_{h,\delta}\sqrt{\log t/t}$ under part \ref{apxitm:lipschitz}, and
substituting \eqref{eq:corC4_twiceDiff_mid} yields
$D_2h^2+C_{h,\delta}\sqrt{\log t/t}$ under part \ref{apxitm:twice_diff}.
Finally, Lemma~\ref{lem:ksa-fixed-h-concentration} shows that the event
on which these bounds hold has probability at least $1-\delta$.
\end{proof}

In the next corollary, we allow the bandwidth to vary with $t$ and
choose it to minimize the explicit upper bounds in Corollary \ref{apxcor:explicit_ksa_smoothness}.

\begin{cor}
\label{apxcor:optimized_ksa_bandwidth}
Suppose Assumption~\ref{assump:individual_cali_assump} holds, fix
$\delta\in(0,1)$, and retain the notation from
Theorem~\ref{apxthm:explicit_statement_individual} and
Corollary~\ref{apxcor:explicit_ksa_smoothness}.

\begin{enumerate}[label=(\roman*)]

\item\label{apxitm:final_cor_lip} Suppose the Lipschitz condition in
part~\ref{apxitm:lipschitz} of
Corollary~\ref{apxcor:explicit_ksa_smoothness} holds. At time $t$, use
\begin{equation}
h_t
:=
\left[
\frac{4dR_\delta}
{\displaystyle
\overline\pi L\mu_1
+
\frac{4\mu_2\bar h}
{v_d(1-2^{-d})h_0^{d+2}}}
\sqrt{\frac{\log t}{t}}
\right]^{1/(d+1)}.
\label{eq:ksa-optimal-lipschitz-bandwidth}
\end{equation}
For every integer $t\geq T_\rho$ such that
\begin{equation}
h_t
\leq
\min\left\{
\bar h,
\frac{d\underline\pi\mu_0}
{\displaystyle
\overline\pi L\mu_1
+
\frac{4\mu_2\bar h}
{v_d(1-2^{-d})h_0^{d+2}}}
\right\},
\label{eq:ksa-optimal-lipschitz-admissibility}
\end{equation}
the KSA interval constructed using $\kappa_{h_t}$ satisfies
\begin{equation}
\begin{aligned}
\mathbb P\Bigg(
&\sup_{\bz:\,\bW_{t+1}(\bz)\in\mathcal D}
\left|
\mathbb P\left(
G_{t+1}\notin\Ct_t(\bz)
\mid\mathcal F_t,\bz_{t+1}=\bz
\right)-\alpha
\right|                                                        \\
&\leq
\frac{4(d+1)}{d\underline\pi\mu_0}
\left(
\overline\pi L\mu_1
+
\frac{4\mu_2\bar h}
{v_d(1-2^{-d})h_0^{d+2}}
\right)                                                        \\
&\qquad\times
\left[
\frac{4dR_\delta}
{\displaystyle
\overline\pi L\mu_1
+
\frac{4\mu_2\bar h}
{v_d(1-2^{-d})h_0^{d+2}}}
\sqrt{\frac{\log t}{t}}
\right]^{1/(d+1)}
\Bigg)
\geq1-\delta.
\end{aligned}
\label{eq:ksa-optimal-lipschitz-coverage}
\end{equation}

\item\label{apxitm:final_cor_twicediff} Suppose the second-order condition in
part~\ref{apxitm:twice_diff} of
Corollary~\ref{apxcor:explicit_ksa_smoothness} holds. At time $t$, use
\begin{equation}
h_t
:=
\left[
\frac{4dR_\delta}
{\displaystyle
(M_{\pi F}+M_\pi)\mu_2
+
\frac{16\mu_2}
{v_d(1-2^{-d})h_0^{d+2}}}
\sqrt{\frac{\log t}{t}}
\right]^{1/(d+2)}.
\label{eq:ksa-optimal-second-order-bandwidth}
\end{equation}
For every integer $t\geq T_\rho$ such that
\begin{equation}
h_t
\leq
\min\left\{
\bar h,
\sqrt{
\frac{d\underline\pi\mu_0}
{\displaystyle
(M_{\pi F}+M_\pi)\mu_2
+
\frac{16\mu_2}
{v_d(1-2^{-d})h_0^{d+2}}}
}
\right\},
\label{eq:ksa-optimal-second-order-admissibility}
\end{equation}
the KSA interval constructed using $\kappa_{h_t}$ satisfies
\begin{equation}
\begin{aligned}
\mathbb P\Bigg(
&\sup_{\bz:\,\bW_{t+1}(\bz)\in\mathcal D}
\left|
\mathbb P\left(
G_{t+1}\notin\Ct_t(\bz)
\mid\mathcal F_t,\bz_{t+1}=\bz
\right)-\alpha
\right|                                                        \\
&\leq
\frac{2(d+2)}{d\underline\pi\mu_0}
\left(
(M_{\pi F}+M_\pi)\mu_2
+
\frac{16\mu_2}
{v_d(1-2^{-d})h_0^{d+2}}
\right)                                                        \\
&\qquad\times
\left[
\frac{4dR_\delta}
{\displaystyle
(M_{\pi F}+M_\pi)\mu_2
+
\frac{16\mu_2}
{v_d(1-2^{-d})h_0^{d+2}}}
\sqrt{\frac{\log t}{t}}
\right]^{2/(d+2)}
\Bigg)\!\!
\geq1-\delta.
\end{aligned}
\label{eq:ksa-optimal-second-order-coverage}
\end{equation}

\end{enumerate}

The conditions in \eqref{eq:ksa-optimal-lipschitz-admissibility} and
\eqref{eq:ksa-optimal-second-order-admissibility} are eventually satisfied
because the corresponding bandwidths converge to zero. The two
coverage-error bounds converge to zero at the rates
$(\log t/t)^{1/(2d+2)}$ and $(\log t/t)^{1/(d+2)}$ respectively.
\end{cor}

\begin{proof}
Fix $\delta\in(0,1)$ and $t\geq T_\rho$, and take $h_t$ to be the bandwidth
specified in the relevant part of the corollary. Since $h_t$ is deterministic,
Lemma~\ref{lem:ksa-fixed-h-concentration} may be applied at time $t$ with
$h=h_t$. On the event in \eqref{eq:ksa-concentration-probability}, which has
probability at least $1-\delta$,
$$
\sup_{\bw\in\mathcal D}\sup_{y\in\mathbb R}
|A_{t,h_t}(y,\bw)-A_{h_t}(y,\bw)|
\vee
\sup_{\bw\in\mathcal D}
|B_{t,h_t}(\bw)-B_{h_t}(\bw)|
\leq
R_\delta\sqrt{\frac{\log t}{t}}.
$$
We work on this event. The ratio argument in the proof of
Theorem~\ref{apxthm:explicit_statement_individual} gives
\begin{equation}
\sup_{\bw\in\mathcal D}\sup_{y\in\mathbb R}
|\widehat F_t(y\mid\bw)-F_{h_t}(y\mid\bw)|
\leq
\frac{8R_\delta}{\underline\pi\mu_0h_t^d}
\sqrt{\frac{\log t}{t}}.
\label{eq:ksa-time-varying-ratio-bound}
\end{equation}

For part \ref{apxitm:final_cor_lip}, expanding $R_0$ and $D_L$ in
\eqref{eq:ksa-lipschitz-population-bias} gives
$$
\sup_{\bw\in\mathcal D}\sup_{y\in\mathbb R}
|F_{h_t}(y\mid\bw)-F_{Y\mid\bW}(y\mid\bw)|
\leq
\frac{2}{\underline\pi\mu_0}
\left(
\overline\pi L\mu_1
+
\frac{4\mu_2\bar h}
{v_d(1-2^{-d})h_0^{d+2}}
\right)h_t.
$$
Combining this bound with \eqref{eq:ksa-time-varying-ratio-bound} and the
coverage-conversion inequality \eqref{eq:ksa-coverage-from-CDF}, we obtain
\begin{equation}
\begin{aligned}
&\sup_{\bz:\,\bW_{t+1}(\bz)\in\mathcal D}
\left|
\mathbb P\left(
G_{t+1}\notin\Ct_t(\bz)
\mid\mathcal F_t,\bz_{t+1}=\bz
\right)-\alpha
\right|                                                        \\
&\quad\leq
\frac{4}{\underline\pi\mu_0}
\left(
\overline\pi L\mu_1
+
\frac{4\mu_2\bar h}
{v_d(1-2^{-d})h_0^{d+2}}
\right)h_t
+
\frac{16R_\delta}{\underline\pi\mu_0h_t^d}
\sqrt{\frac{\log t}{t}}.
\end{aligned}
\label{eq:ksa-expanded-lipschitz-bound}
\end{equation}
The bandwidth in \eqref{eq:ksa-optimal-lipschitz-bandwidth} is the unique
minimizer of the right-hand side of
\eqref{eq:ksa-expanded-lipschitz-bound} over $h_t>0$. At this bandwidth, the
second term on the right-hand side equals $1/d$ times the first. Hence, the
minimum is
$\frac{4(d+1)}{d\underline\pi\mu_0}
\left(
\overline\pi L\mu_1
+
\frac{4\mu_2\bar h}
{v_d(1-2^{-d})h_0^{d+2}}
\right)h_t,$ which, after substituting \eqref{eq:ksa-optimal-lipschitz-bandwidth}, is the
right-hand side of \eqref{eq:ksa-optimal-lipschitz-coverage}.

For part \ref{apxitm:final_cor_twicediff}, expanding $R_0$ and $D_2$ in
\eqref{eq:ksa-second-order-population-bias} gives
$$
\sup_{\bw\in\mathcal D}\sup_{y\in\mathbb R}
|F_{h_t}(y\mid\bw)-F_{Y\mid\bW}(y\mid\bw)|\leq
\frac{1}{\underline\pi\mu_0}
\left(
(M_{\pi F}+M_\pi)\mu_2
+
\frac{16\mu_2}
{v_d(1-2^{-d})h_0^{d+2}}
\right)h_t^2.
$$
Combining this bound with \eqref{eq:ksa-time-varying-ratio-bound} and
\eqref{eq:ksa-coverage-from-CDF} yields
\begin{equation}
\begin{aligned}
&\sup_{\bz:\,\bW_{t+1}(\bz)\in\mathcal D}
\left|
\mathbb P\left(
G_{t+1}\notin\Ct_t(\bz)
\mid\mathcal F_t,\bz_{t+1}=\bz
\right)-\alpha
\right|                                                        \\
&\quad\leq
\frac{2}{\underline\pi\mu_0}
\left(
(M_{\pi F}+M_\pi)\mu_2
+
\frac{16\mu_2}
{v_d(1-2^{-d})h_0^{d+2}}
\right)h_t^2
+
\frac{16R_\delta}{\underline\pi\mu_0h_t^d}
\sqrt{\frac{\log t}{t}}.
\end{aligned}
\label{eq:ksa-expanded-second-order-bound}
\end{equation}
The bandwidth in \eqref{eq:ksa-optimal-second-order-bandwidth} is the unique
minimizer of the right-hand side of
\eqref{eq:ksa-expanded-second-order-bound} over $h_t>0$. At this bandwidth, the
second term on the right-hand side equals $2/d$ times the first. Therefore, the
minimum is
$$
\frac{2(d+2)}{d\underline\pi\mu_0}
\left(
(M_{\pi F}+M_\pi)\mu_2
+
\frac{16\mu_2}
{v_d(1-2^{-d})h_0^{d+2}}
\right)h_t^2,
$$
which is the right-hand side of
\eqref{eq:ksa-optimal-second-order-coverage} after substituting
\eqref{eq:ksa-optimal-second-order-bandwidth}.
\end{proof}

\clearpage
\phantomsection\label{part:ec-references}
\putbib
\end{bibunit}

\end{document}